\documentclass[pra, twocolumn, nofootinbib, preprintnumbers, superscriptaddress]{revtex4-1}

\usepackage{graphicx, color, graphpap}      
\usepackage{amsmath}
\usepackage{enumitem}
\usepackage{amssymb}
\usepackage{amsthm}
\usepackage{pstricks}
\usepackage{float}
\usepackage[colorlinks=true, citecolor=blue, linkcolor=red]{hyperref}
\usepackage[T1]{fontenc}
\usepackage{bbm}
\usepackage{dsfont}
\usepackage[linesnumbered,ruled,vlined]{algorithm2e}
\SetKwInput{kwInit}{Init}
\usepackage{mathtools}
\DeclarePairedDelimiter\ceil{\lceil}{\rceil}

\usepackage{tikz}
\usepackage{graphicx}
\usetikzlibrary{positioning}
\usepackage{graphicx}
\usepackage{xcolor}
\usepackage[caption=false]{subfig}
\usepackage[ruled,vlined]{algorithm2e}
\usepackage{siunitx}
\usepackage{qcircuit}
\usepackage{pifont} 


\newcommand{\ket}[1]{| #1 \rangle}

\newcommand{\ketbra}[2]{| \hspace{1pt} #1 \rangle \langle #2 \hspace{1pt} |}

\newcommand\crule[3][black]{\textcolor{#1}{\rule{#2}{#3}}}

\newcommand{\E}{\mathcal{E}}

\newcommand{\Tr}{{\rm Tr}}

\renewcommand{\geq}{\geqslant}
\renewcommand{\leq}{\leqslant}

\renewcommand{\Re}{\text{Re}}

\renewcommand{\vec}[1]{\boldsymbol{#1}}  

\newcommand{\rhox}{\rho_{\vec{x}}}
\newcommand{\rhoxi}{\rho_{\vec{x}_i}}
\newcommand{\rhotildex}{\tilde{\rho}_{\vec{x}}}
\newcommand{\rhotildexi}{\tilde{\rho}_{\vec{x}_i}}

\newcommand{\sigmax}{\sigma_{\vec{x}}}

\newcommand{\yhat}{\hat{y}}

\newcommand{\depo}{\E_p^{\text{depo}}}          
\newcommand{\depon}{\E_p^{\text{GD}}}           

\newcommand{\flip}{\E_p^{\text{BF}}}            
\newcommand{\deph}{\E_p^{\text{deph}}}          
\newcommand{\paul}{\E_{\mathbf{p}}^{\text{P}}}  
\newcommand{\damp}{\E_p^{\text{AD}}}            

\newcommand{\moons}{{\fontfamily{cmtt}\selectfont moons}}
\newcommand{\vertical}{{\fontfamily{cmtt}\selectfont vertical}}
\newcommand{\diagonal}{{\fontfamily{cmtt}\selectfont diagonal}}

\newcommand{\defref}[1]{Definition \ref{#1}}
\newcommand{\secref}[1]{Section \ref{#1}}
\newcommand{\appref}[1]{Appendix \ref{#1}}
\newcommand{\theref}[1]{Theorem \ref{#1}}
\newcommand{\corrref}[1]{Corollary \ref{#1}}
\newcommand{\figref}[1]{Figure \ref{#1}}
\renewcommand{\eqref}[1]{(\ref{#1})}

\newtheoremstyle{example}{\topsep}{\topsep}%
{}
{}
{\bfseries}
{:}
{   }
{\thmname{#1}\thmnumber{ #2}}
\theoremstyle{example}
\newtheorem{theorem}{Theorem}
\newtheorem{lemma}{Lemma}
\newtheorem{corollary}{Corollary}

\theoremstyle{definition}
\newtheorem{definition}{Definition}

\newtheorem*{theorem*}{Theorem}

\def\orcid#1{\kern -0.4em\href{https://orcid.org/#1}{\includegraphics[keepaspectratio,width=0.7em]{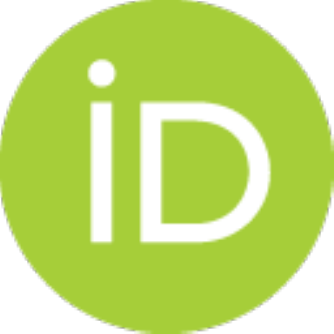}}}

\DeclareMathOperator*{\argmin}{arg\,min}

\renewcommand{\>}{\rangle}
\newcommand{\<}{\langle}

\definecolor{ForestGreen}{RGB}{34, 139, 34}

\long\def\ca#1\cb{} 
\begin{document}
\title{Robust data encodings for quantum classifiers}

\author{Ryan LaRose}
\affiliation{Department of Computational Mathematics, Science, and Engineering, Michigan State University, East Lansing, MI 48823, USA.}

\author{Brian Coyle  \orcid{0000-0002-3436-8458}}
\affiliation{School of Informatics, University of Edinburgh, 10 Crichton Street, United Kingdom.}
\email{brian.coyle@ed.ac.uk}

\begin{abstract}
    Data representation is crucial for the success of machine learning models. In the context of quantum machine learning with near-term quantum computers, equally important considerations of how to efficiently input (encode) data and effectively deal with noise arise. In this work, we study data encodings for binary quantum classification and investigate their properties both with and without noise. For the common classifier we consider, we show that encodings determine the classes of learnable decision boundaries as well as the set of points which retain the same classification in the presence of noise. After defining the notion of a robust data encoding, we prove several results on robustness for different channels, discuss the existence of robust encodings, and prove an upper bound on the number of robust points in terms of fidelities between noisy and noiseless states. Numerical results for several example implementations are provided to reinforce our findings. 
\end{abstract}


\maketitle

\section[\color{blue} Introduction]{Introduction} \label{sec:intro}

Fault-tolerant quantum computers which can efficiently simulate physics~\cite{feynman_simulating_1982} and factor prime numbers~\cite{shor_polynomial-time_1997} lie on an unclear timeline. Current quantum processors in laboratories or in the cloud~\cite{larose_overview_2019} contain fewer than one-hundred qubits with short lifetimes and noisy gate operations. Such devices cannot yet implement fault-tolerant procedures and are said to belong to the noisy, intermediate scale quantum (NISQ) era~\cite{preskill_quantum_2018}. While the gap between NISQ capabilities and required resources for, e.g., factoring 2048-bit integers is large~\cite{gidney_how_2019}, recent hardware advancements have led to the first quantum computation which classical computers cannot emulate~\cite{arute_quantum_2019}. This combination of improved hardware and unclear timeline for fault-tolerance makes the question ``What (useful) applications can NISQ computers implement?'' both interesting and important to consider. 

A leading candidate for NISQ applications is a class of algorithms known as variational quantum algorithms (VQAs)~\cite{mcclean_theory_2016}. VQAs use a NISQ computer to evaluate an objective function and a classical computer to adjust input parameters to optimize the function. VQAs have been proposed or used for many applications including quantum chemistry~\cite{peruzzo_variational_2014}, approximate optimization~\cite{farhi_quantum_2014}, quantum state diagonalization~\cite{larose_variational_2019}, quantum compilation~\cite{khatri_quantum-assisted_2019, jones_quest_2019}, quantum field theory simulation~\cite{klco_quantum-classical_2018, klco_digitization_2019}, linear systems of equations~\cite{bravo-prieto_variational_2019, xu_variational_2019, huang_near-term_2019}, and even quantum foundations~\cite{arrasmith_variational_2019}. More fundamental questions about the computational complexity~\cite{mcclean_theory_2016, biamonte_universal_2019}, trainability~\cite{mcclean_barren_2018, grant_initialization_2019, cerezo_cost-function-dependent_2020} and noise-resilience~\cite{sharma_noise_2020} of VQAs have also been considered. 

Variational quantum algorithms have strong overlap with machine learning algorithms, which seek to train a computer to recognize patterns by designing and minimizing a cost function defined over an input dataset~\cite{goodfellow_deep_2016}. In this context, VQAs can be considered quantum neural networks (QNNs)~\cite{farhi_classification_2018, benedetti_parameterized_2019}, and a multitude of applications from classical machine learning can be realized with (simulated) quantum computers. Such applications include generative modeling \cite{verdon_quantum_2017, benedetti_generative_2019, liu_differentiable_2018, coyle_born_2019, dallaire-demers_quantum_2018}, transfer learning~\cite{mari_transfer_2019}, and classification~\cite{farhi_classification_2018, schuld_circuit-centric_2018, schuld_supervised_2018, grant_hierarchical_2018, perez-salinas_data_2019, blank_quantum_2019, abbas_quantum_2020}. These QNN applications, along with additional techniques and applications based on quantum kernel methods \cite{schuld_quantum_2019, havlicek_supervised_2019, kubler_quantum_2019, suzuki_analysis_2019}, define the emerging field of quantum machine learning (QML)~\cite{wittek_quantum_2014, schuld_supervised_2018}\footnote{Although quantum machine learning is generally considered an emerging field, foundational ideas have been published more than two decades ago~\cite{behrman_quantum_1996}.}.  


Despite these applications, several fundamental questions lie at the forefront of QML. Perhaps the most pressing question is whether quantum models can provide any advantages over classical models. While some (generally) negative theoretical results have been shown for sampling complexity~\cite{arunachalam_optimal_2017} and information capacity~\cite{wright_capacity_2019}, other (generally) positive results have been shown for expressive power~\cite{du_expressive_2018, coyle_born_2019} and problem-specific sampling complexity~\cite{servedio_equivalences_2004, arunachalam_quantum_2020}, 
and the question is largely open. Any practical experiments toward demonstrating advantage give rise to additional crucial questions. One question, general to all NISQ applications, is how to deal with noise present in NISQ computers. A second question, specific to QML applications, is how to (efficiently) input or \textit{encode} data into a QML model.

For the first question, several general-purpose strategies for dealing with noise such as dynamical decoupling~\cite{viola_dynamical_1999}, probabilistic error cancellation \& zero-noise extrapolation~\cite{temme_error_2017}, and quantum subspace expansion~\cite{mcclean_decoding_2020} have been proposed. However, the robustness (resilience) of particular VQAs in the presence of noise has not been thoroughly investigated, with the exception of preliminary studies in quantum compiling~\cite{khatri_quantum-assisted_2019,sharma_noise_2020} and approximate optimization~\cite{xue_effects_2019, alam_analysis_2019, marshall_characterizing_2020}.
Understanding robustness properties of VQAs is key to making progress towards practical NISQ implementations. For the second question on inputting data, most studies in QML focus on the design of the QNN while assuming a full wavefunction representation of arbitrary data~\cite{schuld_circuit-centric_2018,farhi_classification_2018, harrow_quantum_2009, kerenidis_quantum_2016, kerenidis_quantum_2018, dervovic_quantum_2018, zhao_smooth_2018}.
This assumption is not suitable for practical implementations as it is well-known that preparing an arbitrary quantum state requires a number of gates exponential in the number of qubits~\cite{Knill_Laflamme_Milburn_2001}.

In this paper, we study both of the above questions in the context of binary quantum classification. In particular, we define different methods of encoding data and analyze their properties both with and without noise. For the noiseless case, we demonstrate that different data encodings lead to different sets of learnable decision boundaries for the quantum classifier. For the noisy case, given a quantum channel, we define the notion of a robust point for the quantum classifier --- a generalization of fixed points for quantum operations. We completely characterize the set of robust points for example quantum channels, and discuss how encoding data into robust points is a type of problem-specific error mitigation for quantum classifiers.

To these ends, the rest of the paper is organized as follows. Section~\ref{sec:definitions} presents definitions which are used in the remainder to prove our results, including a formal definition of a common binary quantum classifier we consider (Sec.~\ref{subsec:intro/quantum-classifiers}), definitions and examples of data encodings (Sec.~\ref{sec:data-encodings}), noise channels we consider (Sec.~\ref{sec:noise-models}), and definitions of robust points and robust encodings (Sec.~\ref{subsec:robustness-definition}). After this, we present analytic results and proofs for robustness in Section~\ref{sec:noise_robustness}. We begin by showing that different encodings lead to different classes of learnable decision boundaries in Sec.~\ref{ssec:classes_learnable_decision_boundaries}, then characterize the set of robust points for example quantum channels in Sec.~\ref{ssec:characterize_robust_points}. In Sec.~\ref{subsec:robustness-results}, we state and prove robustness results, and in Sec.~\ref{subsec:existence-of-robust-encodings} we discuss the existence of robust encodings. Finally, we prove an upper bound on the number of robust points in terms of fidelities between noisy and ideal states in Sec.~\ref{ssec:fidelity_bounds}. Last, we include several numerical results in Sec.~\ref{sec:numerical_results} that reinforce and extend our findings, and finally conclude in Sec.~\ref{sec:conclusions}.

\section{Preliminary Definitions} \label{sec:definitions}


\subsection{Quantum Classifiers} \label{subsec:intro/quantum-classifiers}


In classical machine learning, classification problems are a subclass of supervised learning problems in which the computer model is presented with labeled data and asked to learn some pattern. For binary classification, the input is a set of labeled feature vectors
\begin{equation} \label{eqn:labeled-data-for-classifier}
    \{(\vec{x}_i, y_i)\}_{i = 1}^{M}
\end{equation} 
where $\vec{x}_i \in \mathcal{X}$ is a feature vector, $\mathcal{X}$ is an arbitrary set%
\footnote{In practice we typically have $\mathcal{X} = \mathbb{R}^N$, but other sets --- e.g., $\mathcal{X} = \mathbb{Z}^N$ or $\mathcal{X} \in \mathbb{Z}_2^N$ --- are possible, so we write $\mathcal{X}$ for generality.}
, and ${y}_i \in \{0, 1\}$ is a binary label. Given this data, the goal of the learner is to output a rule $f: \mathcal{X} \rightarrow \{0, 1\}$ which accurately classifies the data and can be used to make predictions on new data. In practice, this is accomplished by defining a model (e.g., a neural network) and cost function, then minimizing this cost function by ``training'' the model over the input data. 

In this work, we restrict to binary (quantum) classifiers, henceforth called simply (quantum) classifiers. We remark that multi-label classification problems can be reduced to binary classification by standard methods. 

\begin{figure}
    \centering
    \includegraphics{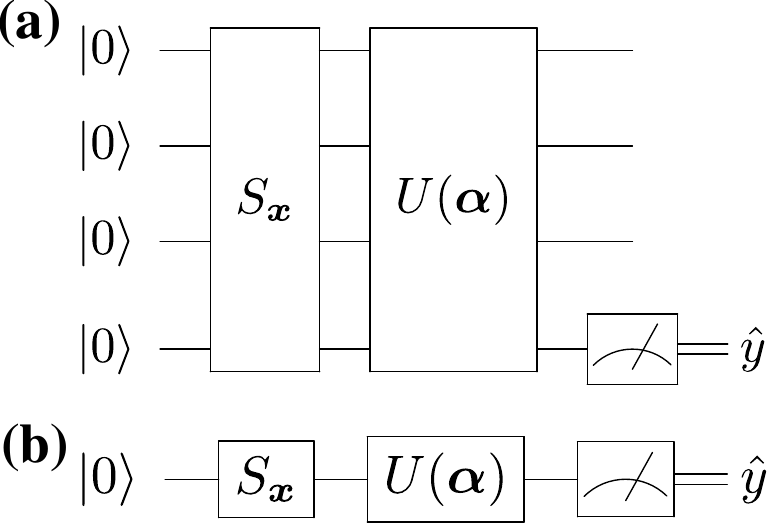}
    \caption{A common architecture for a binary quantum classifier that we study in this work. The general circuit structure is shown in \textbf{(a)} and the structure for a single qubit is highlighted in \textbf{(b)}. In both, a feature vector $\vec{x}$ is encoded into a quantum state $\rhox$ via a ``state preparation'' unitary $S_{\vec{x}}$. The encoded state $\rhox$ then evolves to $U \rhox U^\dagger =: \rhotildex$ where $U(\vec{\alpha})$ is a unitary ansatz with trainable parameters $\vec{\alpha}$. A single qubit of the evolved state $\rhotildex$ is measured to yield a predicted label $\yhat$ for the vector $\vec{x}$.}
    \label{fig:classifier}
\end{figure}

A quantum classifier is essentially the same as a classical one, the key difference being the model (ansatz) of the classifier. As mentioned in the Introduction, many QML architectures use a variational quantum algorithm, nominally consisting of parameterized one- and two-qubit gates, which is called a quantum neural network (QNN) in the context of machine learning. QNNs can be considered function approximators analogous to classical neural networks (e.g., feedforward neural networks \cite{fine_feedforward_1999}), and the procedure for training the QNN consists of adjusting gate parameters such that this function approximator outputs good predictions for the input data.

While the use of QNNs as machine learning models may present the possibility of advantage for particular problems~\cite{du_expressive_2018, coyle_born_2019},
QNNs also present key challenges for machine learning. For nearly all QML problems, a pressing challenge is inputting (arbitrary) data to the model such that the QNN can process it. We refer to this input process as \textit{data encoding}, and discuss it in detail in Sec.~\ref{sec:data-encodings}. Another potential challenge with QNNs is outputting information, since the data propagated through the QNN is a quantum state. For machine learning applications, this means that the output feature vector (amplitudes of the quantum state) cannot be accessed efficiently. Rather, as is usual in quantum mechanics, quantities of the form $\Tr[\rho \hat{O}]$ where $\rho$ is the quantum state and $\hat{O}$ is some Hermitian operator can be efficiently evaluated.

Fortunately for quantum classification, outputting information (predictions) can be done in a straightforward manner. As several authors have noted~\cite{farhi_classification_2018, schuld_circuit-centric_2018, schuld_supervised_2018, grant_hierarchical_2018, perez-salinas_data_2019, blank_quantum_2019}, it is natural to use the measurement outcome of a single qubit as a class prediction as produces a binary outcome. We adopt this strategy in our work.

Informally, we define a (binary) quantum classifier as a procedure for encoding data into a quantum circuit, processing it through trainable QNN, and outputting a (binary) predicted label. Given a feature vector $\vec{x} \in \mathcal{X}$, a concise description of such a classifier can be written
\begin{align}
    \vec{x} &\mapsto \rhox && \text{(encoding)} \label{eqn:encoding} \\
    &\mapsto \rhotildex  && \text{(processing)}  \label{eqn:processing}\\
    &\mapsto \yhat [ \rhotildex ] . && \text{(prediction)} \label{eqn:prediction}
\end{align}
%

Several remarks are in order.
First, a given data point $\vec{x}$ in the training set~\eqref{eqn:labeled-data-for-classifier} is encoded in a quantum state $\rho_{\vec{x}} \in \mathcal{D}_n$ via a state preparation unitary $S_{\vec{x}}$ (see Fig.~\ref{fig:classifier}. Throughout the paper, we use $\mathcal{D}_n \subset \mathbb{C}^{2^n \times 2^n}$ to denote the set of density operators (matrices) on $n$ qubits. We remark that each $\vec{x}$ in the training set leads to a (unique) $S_{\vec{x}}$, so the state preparation unitary can be considered a parameterized family of unitary ans\"{a}tze. We discuss encodings in detail in Sec.~\ref{sec:data-encodings}. 

For the processing step~\eqref{eqn:processing}, there have been many proposed QNN architectures in recent literature, including quantum convolutional neural networks \cite{cong_quantum_2019, henderson_quanvolutional_2019}, strongly entangling ans\"{a}tze~\cite{schuld_circuit-centric_2018}, and more~\cite{stoudenmire_supervised_2016, grant_hierarchical_2018}. In this work, we allow for a general unitary evolution $U(\vec{\alpha})$ such that
\begin{equation} \label{eqn:unitary_evolution}
     \rhotildex =  U(\vec{\alpha}) \rhox U^\dagger(\vec{\alpha}) .
\end{equation}
We remark that some QNN architectures involve intermediate measurements and conditional processing (notably \cite{cong_quantum_2019}) and so do not immediately fit into~\eqref{eqn:unitary_evolution}. Our techniques for showing robustness could be naturally extended to such architectures, however, and so we consider~\eqref{eqn:unitary_evolution} as a simple yet general model. We also note that training the classifier via minimization of a well-defined cost function is an important task with interesting questions, but we primarily focus on data encodings and their properties in this work. For this reason we often suppress the trainable parameters $\vec{\alpha}s$ and write $U$ for $U(\vec{\alpha})$. 

Finally, the remaining step is to extract information from the state $\rhotildex$ to obtain a predicted label. As mentioned, a natural method for doing this is to measure a single qubit which yields a binary outcome $0$ or $1$ taken as the predicted label $\yhat$. Since measurements are probabilistic, we measure $N_m$ times and take a ``majority vote.'' That is, if $0$ is measured $N_0$ times and $N_0 \geq N_m / 2$, we take $0$ as the class prediction, else $1$. Generalizing the finite statistics, this condition can be expressed analytically as
\begin{equation} \label{eqn:decision_rule}
    \yhat [\rhotildex] = \begin{cases}
        0 \qquad \text{if   } \Tr [ \Pi_0^c \rhotildex] \ge 1 / 2 \\
        1 \qquad \text{otherwise} 
    \end{cases} 
\end{equation}
where 
\begin{equation} \label{eqn:pi0-projector-definition}
    \Pi_0^c := |0\> \<0|_c \equiv |0\> \<0|_c \otimes I_{\bar{c}}
\end{equation}
is the projector onto the ground state of the classification qubit, labeled $c$, and the remaining qubits are labeled $\bar{c}$. For brevity we often omit these labels when it is clear from context.
Throughout the paper, we use $\yhat$ for predicted labels and $y$ for true labels, and we refer to~\eqref{eqn:decision_rule} as the \textit{decision rule} for the classifier. Equation~\eqref{eqn:decision_rule} is not the only choice for such a decision rule. In particular, one could choose a different ``weight'' $\lambda$ such that $\yhat = 0$ if $\Tr [ \Pi_0 \rhotildex] \ge \lambda$ as in Ref.~\cite{perez-salinas_data_2019}, add a bias to the classifier as in Ref.~\cite{schuld_circuit-centric_2018}, or measure the classification qubit in a different basis (e.g., the Hadamard basis instead of the computational basis). Our techniques for showing robustness (Sec.~\ref{subsec:robustness-results}) could be easily adapted for such alternate decision boundaries, and we consider~\eqref{eqn:decision_rule} as a simple yet general rule for the remainder.

The preceding discussion is summarized in the following formal definition of a quantum classifier.

\begin{definition}[Quantum Classifier] \label{def:binary-quantum-classifier}
    A (binary) quantum classifier consists of three well-defined functions: \\
    (i) an encoding function
    \begin{align} 
         E: \mathcal{X} &\rightarrow \mathcal{D}_n \label{eqn:encoding-formal} \\
         E(\vec{x}) &= \rhox,
    \end{align}
    (ii) a function which evolves the state
    \begin{align}
         \mathcal{U}: \mathbb{C}^{2^n \times 2^n} &\rightarrow \mathbb{C}^{2^m \times 2^m} \label{eqn:evolution-formal} \\
         \mathcal{U}(\rhox) &= \rhotildex,
    \end{align}
    and (iii) a decision rule
    \begin{align} \label{eqn:decision-rule-formal}
         \hat{y}: \mathbb{C}^{2^m \times 2^m} \rightarrow \{0, 1\} .
    \end{align}
    For training, a quantum classifier is provided with labeled data~\eqref{eqn:labeled-data-for-classifier}, a cost function $C$, and an optimization routine for minimizing the cost function. 
\end{definition}

Specification of the functions $E$, $\mathcal{U}$, and $\yhat$ --- along with training data, a cost function, and an optimization routine --- uniquely define a quantum classifier. In this work, we let $\mathcal{U}$ be a general unitary evolution~\eqref{eqn:unitary_evolution} and always take the decision rule $\yhat$ to be~\eqref{eqn:decision_rule}.
In the remainder, we study the effects of different encoding functions~\eqref{eqn:encoding-formal}, which we now discuss in more detail.

\subsection{Data Encodings} \label{sec:data-encodings}

\begin{figure}
    \centering
    \includegraphics[width=\columnwidth]{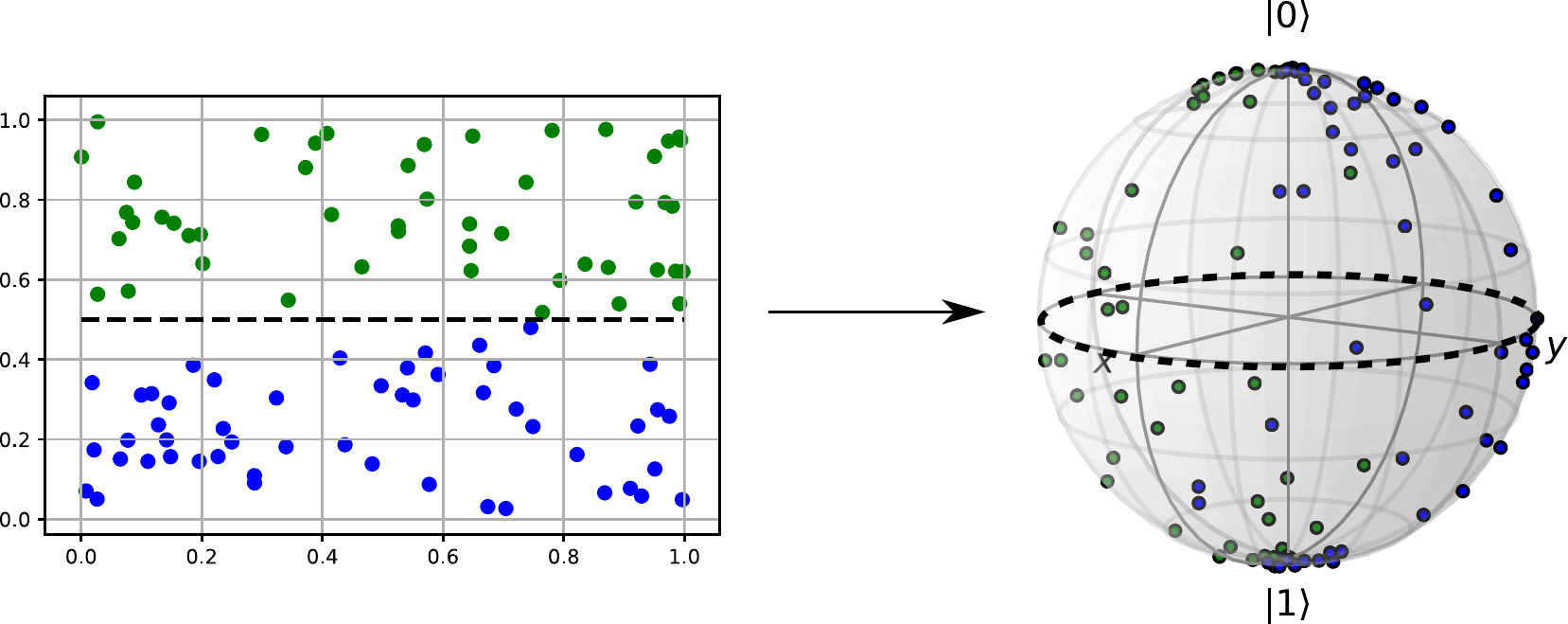}
    \caption{(Color online.) A visual representation of data encoding~\eqref{eqn:encoding-formal} for a single qubit. On the left is shown a set of randomly generated points $\{\vec{x}_i, y_i\}_{i = 1}^{M}$ normalized to lie within the unit square, separated by a true decision boundary shown by the dashed black line. A data encoding maps each $\vec{x}_i \in \mathbb{R}^2$ to a point on the Bloch sphere $\rho_{\vec{x}_i} \in \mathbb{C}^2$, here using the dense angle encoding~\eqref{eqn:dae-single-qubit}. The dashed black line on the Bloch sphere shows the initial decision boundary of the quantum classifier. During the training phase, unitary parameters are adjusted to rotate the dashed black line to correctly classify as many training points as possible. Different data encodings lead to different learnable decision boundaries and different robustness properties, as discussed in the main text.}
    \label{fig:data-encoding-schmatic-single-qubit}
\end{figure}


An encoding can be thought of as ``loading'' a data point $\vec{x} \in \mathcal{X}$ from memory into a quantum state so that it can be processed by a QNN. Unlike classical machine learning, this presents a unique challenge in QML. The ``loading'' is accomplished by an encoding~\eqref{eqn:encoding-formal} from the set $\mathcal{X}$ to $n$-qubit quantum states $\mathcal{D}_n$. As mentioned, many QML papers ~\cite{harrow_quantum_2009, kerenidis_quantum_2016, kerenidis_quantum_2018, dervovic_quantum_2018, zhao_smooth_2018}
assume a full wavefunction encoding with $n = \log_2 N$. This provides an exponential saving in ``space'' at the cost of an exponential increase in ``time.'' That is, a quantum state of $n = \log_2 N$ qubits can represent a data point with $N$ features, but in general such a quantum state takes time $O(2^n)$ to prepare~\cite{Knill_Laflamme_Milburn_2001}.

In practice, data is encoded via a state preparation circuit (unitary) $S_{\vec{x}}$ --- written in terms of one- and two-qubit gates --- which acts on an initial state $|\phi\>$, nominally the all zero state $|\phi\> = |0\>^{\otimes n}$. This realizes the encoding
\begin{equation} \label{eqn:encoding_unitary}
    \vec{x} \mapsto E(\vec{x}) = S_{\vec{x}}\ketbra{\phi}{\phi}S_{\vec{x}}^\dagger = |\vec{x} \> \< \vec{x} | =: \rhox .
\end{equation}
%
For $S_{\vec{x}}$ to be useful for a data encoding, it should have several desirable properties. First, $S_{\vec{x}}$ should have a number of gates which is at most polynomial in the number of qubits. For machine learning applications, we want the family of state preparation unitaries to have enough free parameters such that there is a unique quantum state $\rhox$ for each feature vector $\vec{x}$ --- i.e., such that the encoding function $E$ is bijective. 
Additionally, for NISQ applications, sub-polynomial depth is even more desirable, and we want $S_{\vec{x}}$ to be ``hardware efficient'' --- meaning that the one- and two-qubit gates comprising $S_{\vec{x}}$ can be realized without too much overhead due to, e.g., compiling into the computer's native gate set and implementing swap gates to connect disjoint qubits. 

Motivated by such NISQ limitations, some recent authors~\cite{schuld_supervised_2018, grant_hierarchical_2018, stoudenmire_supervised_2016, schuld_supervised_2018, cao_cost_2019} have considered a ``qubit encoding''
\begin{equation} \label{eqn:qubit_encoding_grant}
    |\vec{x}\> = \bigotimes_{i=1}^{N} \cos (x_i)\ket{0} + \sin (x_i)\ket{1} 
\end{equation}
for the feature vector $\vec{x} = [x_1, ..., x_N]^T \in \mathcal{X}^N$. (Note that for pure state encodings, we often write only the wavefunction $|\vec{x}\> \in \mathbb{C}^{2^n}$, from which the density matrix $\rhox = |\vec{x}\>\<\vec{x}| \in \mathcal{D}_n$ is implicit.) We will also refer to~\eqref{eqn:qubit_encoding_grant} as an ``angle encoding.'' The angle encoding uses $N$ qubits with a constant depth quantum circuit and is thus amenable to NISQ computers. The state preparation unitary is $S_{\vec{x}_j} = \bigotimes_{i = 1}^{N} U_i$ where
\begin{equation}
    U_i := \left[ 
        \begin{matrix}
            \cos (x_j^{(i)}) & -\sin(x_j^{(i)}) \\
            \sin(x_j^{(i)}) & \cos(x_j^{(i)})
        \end{matrix}
    \right] ,
\end{equation}
a strategy which encodes one feature per qubit.


This encoding can be slightly generalized to encode two features per qubit by exploiting the relative phase degree of freedom. We refer to this as the ``dense angle encoding'' and include a definition below.

\begin{definition}[Dense Angle Encoding] \label{def:dae}
    Given a feature vector $\vec{x} = [x_1, ..., x_N]^T \in \mathbb{R}^N$, the dense angle encoding maps $\vec{x} \mapsto E(\vec{x})$ given by
    \begin{equation} \label{eqn:dae-general}
        |\vec{x}\> = \bigotimes_{i=1}^{\ceil*{N / 2}} \cos (\pi x_{2i -1})\ket{0} + e^{2 \pi i x_{2i}} \sin (\pi x_{2i - 1})\ket{1} .
    \end{equation}
\end{definition}

For some of our analytic and numerical results, we highlight the dense angle encoding for two-dimensional data $\vec{x} \in \mathbb{R}^2$ with a single qubit given by
\begin{equation} \label{eqn:dae-single-qubit}
    |\vec{x}\> = \cos (\pi x_{1}) \ket{0} + e^{2  \pi i x_{2}} \sin (\pi x_{1})\ket{1} ,
\end{equation}
which has density matrix
\begin{equation} \label{eqn:dae-density-matrix}
    \rhox =
    \left[ \begin{matrix}
        \cos^2 \pi x_1  & e^{ - 2 \pi i x_2} \cos \pi x_1 \sin \pi x_1 \\
        e^{  2 \pi i x_2} \cos \pi x_1\sin \pi x_1& \sin^2 \pi x_1 \\
    \end{matrix} \right] .
\end{equation}

Although the angle encoding~\eqref{eqn:qubit_encoding_grant} and dense angle encoding~\eqref{eqn:dae-general} use sinuosoids and exponentials, there is nothing special about these functions (other than, perhaps, they appear in common parameterizations of qubits and unitary matrices~\cite{nielsen_quantum_2010}). We can easily abstract these to a general class of qubit encodings which use arbitrary functions.
\begin{definition}[General Qubit Encoding]
    Given a feature vector $\vec{x} = [x_1, ..., x_N]^T \in \mathbb{R}^N$, the general qubit encoding maps $\vec{x} \mapsto E(\vec{x})$ given by
    \begin{equation} \label{eqn:qubit-encoding-general}
        |\vec{x}\> = \bigotimes_{i=1}^{\ceil*{N / 2}} f_i(x_{2i - 1}, x_{2i}) \ket{0} + g_i(x_{2i - 1}, x_{2i}) \ket{1} .
    \end{equation}
    where $f, g : \mathbb{R} \times \mathbb{R} \rightarrow \mathbb{C}$ are such that $|f_i|^2 + |g_i|^2 = 1 \ \forall i$. 
\end{definition}
%

We remark that a similar type of generalization was used in~\cite{perez-salinas_data_2019} with a single qubit classifier that allowed for repeated application of an arbitrary state preparation unitary. 
While~\eqref{eqn:qubit-encoding-general} is the most general description of a qubit encoding --- and is the encoding we primarily focus on in this work --- it is of course not the most general encoding~\eqref{eqn:encoding-formal}. The previously mentioned wavefunction encoding maps $N$ features into $n = \log_2 N$ qubits as follows.

\begin{definition}[Wavefunction Encoding]
     The wavefunction encoding of a vector $\vec{x} \in \mathbb{R}^N$ is
    \begin{equation} \label{eqn:wavefunction_encoding_amplitude}
        |\vec{x}\> := \frac{1}{||\vec{x}||_2^2} \sum_{i = 1}^{N} x_i |i\> 
    \end{equation}
    where $x_i$ is the $i$th feature of $\vec{x}$. 
\end{definition}

As with the dense angle encoding, we highlight the wavefunction encoding for $\vec{x} = [x_1 \ x_2]^T$ with $||\vec{x}||_2 = 1$:
\begin{equation} \label{eqn:wavefunction-encoding-density-matrix-onequbit}
    \rhox = \left[ \begin{matrix}
        x_1^2  & x_1 x_2 \\
        x_1 x_2 & x_2^2 \\
    \end{matrix} \right] .
\end{equation}

While we do not consider it in this work, we note that the wavefunction encoding can be slightly generalized to allow for parameterizations of features (amplitudes).
\begin{definition}[Amplitude Encoding] \label{def:amplitude-encoding}
    For $\vec{x} \in \mathbb{R}^N$, the amplitude encoding maps $\vec{x} \mapsto E(\vec{x})$ given by
    \begin{equation} \label{eqn:amplitude_encoding_grant} 
        \ket{\vec{x}} = \sum\limits_{i=1}^{N} f_i(\vec{x})\ket{i}
    \end{equation}
    where $\sum_i |f_i|^2 = 1$. 
\end{definition}
The functions $f_i$ could only act on the $i$th feature, e.g. $f_i(\vec{x}) = \sin x_i$, or could be more complicated functions of several (or all) features.

Thus far, we have formally defined a data encoding~\eqref{eqn:qubit-encoding-general} and its role in a quantum classifier (Def.~\ref{def:binary-quantum-classifier}), and we have given several examples. While we have discussed different properties of state preparation circuits which implement data encodings (depth, overhead, etc.), we have not yet discussed the two main properties of data encodings we consider in this work: \textit{learnability} and \textit{robustness}. For the first property, we show in Sec.~\ref{ssec:classes_learnable_decision_boundaries} that different data encodings lead to different classes of learnable decision boundaries. For the second property, we show that different data encodings lead to different sets of robust points (to be defined) in Sec.~\ref{ssec:characterize_robust_points} --- Sec.~\ref{subsec:existence-of-robust-encodings}. 
For the latter results which constitute the bulk of our work, we first need to introduce the noise channels we consider and define the notion of a robust point, which we do in the following two sections.

\subsection{Noise in Quantum Systems} \label{sec:noise-models}


In this section, we introduce our notation for the common quantum channels we use in this work. While we provide brief exposition on quantum noise, we refer the reader desiring more background to the standard references~\cite{nielsen_quantum_2010, john_preskill_quantum_1998, Watrous_2018}.

Noise occurs in quantum systems due to interactions with the environment. Letting $\rho$ denote the quantum state of interest and $\rho_{\text{env}}$ the environment, noise can be characterized physically by the process
\begin{equation}
    \rho \mapsto \Tr_{\text{env}} \left[ U \left( \rho \otimes \rho_{\text{env}} \right) U^\dagger \right]
\end{equation}
where $U$ is a unitary on the composite Hilbert space. This can be written in the equivalent, often more convenient, operator-sum representation
\begin{equation} \label{eqn:quantum-operation}
    \rho \mapsto \sum_{k = 1}^{K} E_k \rho E_k^\dagger 
\end{equation}
where the Kraus operators $E_k$ satisfy the completeness relation
\begin{equation}
    \sum_{k = 1}^{K} E_k^\dagger E_k = I .
\end{equation}
Equation~\eqref{eqn:quantum-operation} is known as a quantum operation or quantum channel. Physically, it can be interpreted as randomly replacing the state $\rho$ by the (properly normalized) state $E_k \rho E_k^\dagger$ with probability $\Tr [ E_k \rho E_k ^\dagger ]$.

The quantum channels we study here are standard and often used in theoretical work as reasonable noise models~\cite{nielsen_quantum_2010}. 
For readers familiar with these channels, the following definitions are solely to introduce our notation.

A widely used noise model is the Pauli channel.

\begin{definition} \label{def:pauli-channel}
    The Pauli channel maps a single qubit state $\rho$ to $\paul (\rho)$ defined by
    \begin{equation} \label{eqn:pauli-channel}
        \paul (\rho) := p_I \rho + p_X X \rho X + p_Y Y \rho Y + p_Z Z \rho Z
    \end{equation}
    where $p_I + p_X + p_Y + p_Z = 1$. 
\end{definition}
\noindent
While the Pauli channel acts on a single qubit, it can be generalized to a $d$-dimensional Hilbert space via the \textit{Weyl channel}
%
\begin{equation} \label{eqn:weyl-channel-def}
    \mathcal{E}_{p}^{\text{W}} (\rho) := \sum_{k, l = 0}^{d - 1} p_{kl} W_{kl} \rho W_{kl}^\dagger
\end{equation}
where $p_{kl}$ are probabilities and the Weyl operators are
\begin{equation}
    W_{kl} := \sum_{m = 0}^{d - 1} e^{2 \pi i m k / d} |m\>\< m + 1| .
\end{equation}
For $d = 2$, Eqn.~\eqref{eqn:weyl-channel-def} reduces to Eqn.~\eqref{eqn:pauli-channel}.

Two special cases of the Pauli channel are the bit-flip and phase-flip (dephasing) channel.

\begin{definition} \label{def:bit-flip-channel}
    The bit-flip channel maps a single qubit state $\rho$ to $\deph (\rho)$ defined by
    \begin{equation} \label{eqn:bitflip-channel}
        \flip (\rho) := (1 - p) \rho + p X \rho X 
    \end{equation}
    where $0 \le p \le 1$. 
\end{definition}
\noindent
While a bit-flip channel flips the computational basis state with probability $p$, the phase-flip channel introduces a relative phase with probability $p$.
\begin{definition} \label{def:dephasing-channel}
    The phase-flip (dephasing) channel maps a single qubit state $\rho$ to $\deph (\rho)$ defined by
    \begin{equation} \label{eqn:dephasing-channel}
        \deph (\rho) := ( 1 - p) \rho + p Z \rho Z 
    \end{equation}
    where $0 \le p \le 1$. 
\end{definition}

Another special case of the Pauli channel is the depolarizing channel which occurs when each Pauli is equiprobable $p_X = p_Y = p_Z = p$ and $p_I = 1-3p$.
This channel can be equivalently thought of as replacing the state $\rho$ by the maximally mixed state $I / 2$ with probability $p$.
%
\begin{definition} \label{def:depolarizing-channel}
    The depolarizing channel maps a single qubit state $\rho$ to $\depo (\rho)$ defined by
    \begin{equation} \label{eqn:depolarizing-channel}
        \depo (\rho) := (1 - p) \rho + p I / 2 
    \end{equation}
    where $0 \le p \le 1$. 
\end{definition}
\noindent
The $d = 2^n$-dimensional generalization of Def.~\ref{def:depolarizing-channel} is straightforward.
\begin{definition} \label{def:global-depolarizing-channel}
    The global depolarizing channel maps an $n$-qubit state $\rho$ to $\depon (\rho)$ defined by
    \begin{equation} \label{eqn:global-depolarizing-channel}
        \depon (\rho) := (1 - p) \rho + p I / d
    \end{equation}
    where $0 \le p \le 1$, $d = 2^n$, and $I \equiv I_d$ is the $d$-dimensional identity.  
\end{definition}

Finally, we consider amplitude damping noise which models decay from the excited state to the ground state via spontaneous emission of a photon.
\begin{definition} \label{def:amp-damp-channel}
    The amplitude damping channel maps a single qubit state $\rho$ to $\damp (\rho)$ defined by
    \begin{equation} \label{eqn:amp-damp-channel}
        \damp (\rho) := \left[ \begin{matrix}
            \rho_{00} + p \rho_{11}     & \sqrt{1 - p} \rho_{01} \\
            \sqrt{1 - p} \rho_{10}      & (1 - p) \rho_{11} \\
        \end{matrix} \right]
    \end{equation}
    where $0 \le p \le 1$. 
\end{definition}

Now that we have introduced quantum noise and several common channels, we can define robust data encodings for quantum classifiers. We remark that we also consider measurement noise in Appendix~\ref{app:further_proofs} but omit the definition and results from the main text for brevity.

\subsection{Robust Data Encodings} \label{subsec:robustness-definition}

\begin{figure*}[t!]
    \centering
    \includegraphics[width=1\textwidth, height =0.2\textwidth]{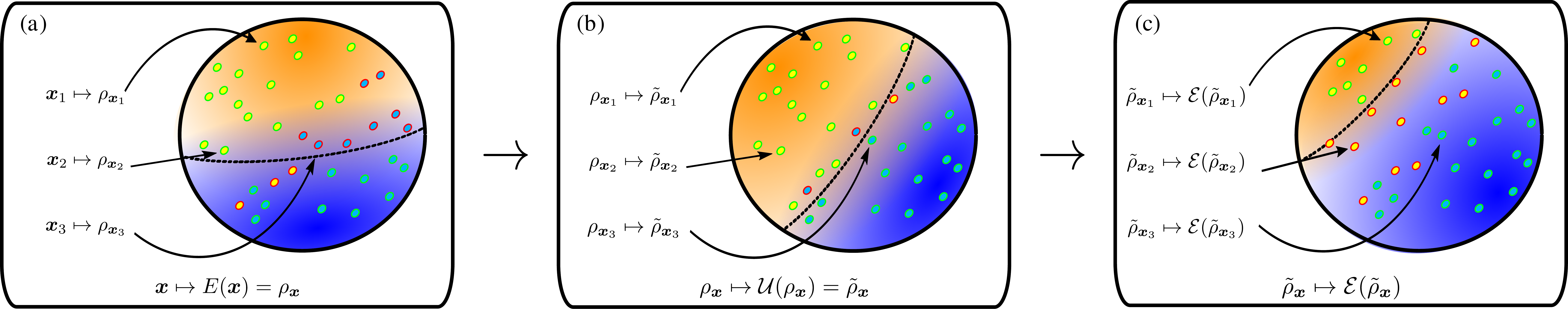}
    \caption{(Color online.) Cartoon illustration of robust points for a single qubit classifier. In panel (a), input training data points $\vec{x}_i$ with classes $y_i \in \{ \text{yellow}, \text{blue} \}$ are mapped into quantum states $\rhoxi$ according to some encoding function $E$. The dashed line through the Bloch sphere indicates the initial decision boundary. Points with a green outline are classified correctly, while points with a red outline are misclassified. In (b), data points are processed by the QNN with optimal unitary parameters (after minimizing a cost function to find such parameters). For clarity, we keep data points fixed and adjust the location of the decision boundary, which is now rotated to correctly classify more points (fewer points with red outlines). In (c), a noise process $\E$ occurs which shifts the location of the final processed points (or location of decision boundary), causing some points to be misclassified. The set of points which maintain the same classification in (b) and (c) are the robust points. Example points $\vec{x}_1$ and $\vec{x}_2$ are correctly classified in (a) and (b) then misclassified in (c) due to the noise. Example point $\vec{x}_3$ is incorrectly classified in (a), correctly classified in (b) after propagating through the QNN, and remains correctly classified in (c).}
    \label{fig:single_qubit_binary_classifier_with_noise}
\end{figure*}

In this section, we define robust points and robust data encodings of quantum classifiers. Informally, the intuition is as follows: the quantum classifier with decision rule~\eqref{eqn:decision_rule} requires only a ``coarse-grained'' measurement to extract a predicted label. For example, with a single qubit classifier, all points in the ``top'' hemisphere of the Bloch sphere are predicted to have label $0$, while all points in the ``bottom'' hemisphere are predicted to have label $1$. The effect of noise is to shift points on the Bloch sphere, but certain points can get shifted such that they get assigned the same labels they would without noise.
This is the idea of robustness, represented schematically in Fig.~\ref{fig:single_qubit_binary_classifier_with_noise}. For classification purposes, we do not require completely precise measurements, only that the point remain ``in the same hemisphere'' in order to get the same predicted label.

Formally, we define a robust point as follows.
\begin{definition}[Robust Point] \label{def:robust-point}
    Let $\mathcal{E}$ be a quantum channel, and consider a (binary) quantum classifier with decision rule $\yhat$ as defined in~\eqref{eqn:decision_rule}. We say that the state $\rhox \in \mathcal{D}_n$ encoding a data point $\vec{x} \in \mathcal{X}$ is a \textit{robust point} of the quantum classifier if and only if
    \begin{equation} \label{eqn:robust-point-definition}
        \yhat [ \mathcal{E} ( \rhotildex )] = \yhat[\rhotildex] 
    \end{equation}
    where $\rhotildex$ is the processed state via~\eqref{eqn:evolution-formal}. 
\end{definition}
%

As mentioned, for the purpose of classification, Eqn.~\eqref{eqn:robust-point-definition} is a well-motivated and reasonable definition of robustness. We remark that~\eqref{eqn:robust-point-definition} is expressed in terms of probability; in practice, additional measurements may be required to reliably determine robustness%
\footnote{\label{footnote:shots}Recall that expectation values evaluated with $N_m$ measurements have variance $1 / \sqrt{N_m}$~\cite{mcclean_theory_2016}. If $\Tr[\Pi_0 \rhotildex] \le 1/2$ and $\Tr[\Pi_0 \mathcal{E} ( \rhotildex )] = 1/2 - \epsilon$ where $\epsilon > 0$, then at least $N_m = 1 / \epsilon^2$ measurements are required to determine robustness. Intuitively, this means that points ``on the border'' between classes require more measurements to distinguish.}%
. Further, we note that~\eqref{eqn:robust-point-definition} assumes that noise occurs only after the evolution $\rhox \mapsto \rhotildex$. While this may be a useful theoretical assumption, in practice noise happens throughout a quantum circuit. We can therefore consider robustness for an \textit{ideal data encoding} as in Def.~\ref{def:robust-point}, or for a \textit{noisy data encoding} in which some noise process $\mathcal{E}_1$ occurs after encoding and another noise process $\mathcal{E}_2$ occurs after evolution:
\begin{equation} \label{eqn:robust-point-noisy-encoding}
    \yhat [ \mathcal{E}_2 (  \mathcal{U} ( \mathcal{E}_1 ( \rhox) ) ) ] =
    \yhat [ \rhotildex ] .
\end{equation}
For our results, we primarily consider~\eqref{eqn:robust-point-definition}, although we show robustness for~\eqref{eqn:robust-point-noisy-encoding} in some cases. 

Robust points~\eqref{eqn:robust-point-definition} are related but not equivalent to (density operator) fixed points of a quantum channel, and can be considered an application-specific generalization of fixed points. In Sec.~\ref{ssec:characterize_robust_points}, we characterize the set of robust points for example channels, and in Sec.~\ref{subsec:existence-of-robust-encodings} we use this connection to prove the existence of robust data encodings.  

For classification, we are concerned with not just one data point, but rather a set of points (e.g., the set $\mathcal{X}$ or training set~\eqref{eqn:labeled-data-for-classifier}). We therefore define the set of robust points, or robust set, in the following natural way.
\begin{definition}[Robust Set] \label{def:robust-set}
    Consider a (binary) quantum classifier with encoding $E : \mathcal{X} \rightarrow \mathcal{D}_n$ and decision rule $\yhat$ as defined in~\eqref{eqn:decision_rule}. Let $\mathcal{E}$ be a quantum channel.
    The set of robust points, or simply \textit{robust set}, is
    \begin{equation} \label{eqn:robust-point-set-definition}
        \mathcal{R} (\mathcal{E}, E, \yhat) := \left\{ \vec{x} \in \mathcal{X}: \yhat [ \mathcal{E} ( \rhotildex )] = \yhat[\rhotildex] \right\}
    \end{equation}
    where $\rhotildex$ is the processed state via~\eqref{eqn:evolution-formal} and $\rho_{\vec{x}} = E(\vec{x})$.
\end{definition}
%

While the robust set generally depends on the encoding $E$, there are cases in which $\mathcal{R}$ is independent of $E$. 
In this scenario, we say all encodings are robust to this channel.
Otherwise, the size of the robust set (i.e., number of robust points) can vary based on the encoding, and we distinguish between two cases. If the robust set is the set of all possible points,
we say that the encoding is \textit{completely robust} to the given noise channel. 

\begin{definition}[Completely Robust Data Encoding] \label{def:complete_noise_robustness}
    Consider a (binary) quantum classifier with encoding $E$ and decision rule $\yhat$ as defined in~\eqref{eqn:decision_rule}. Let $\vec{x} \in \mathcal{X}$ and let $\mathcal{E}$ be a quantum channel. We say that $E$ is a \textit{completely robust data encoding} for the quantum classifier if and only if
    \begin{equation} \label{eqn:completely-robust-encoding-condition}
        \mathcal{R}(\mathcal{E}, E, \yhat) = \mathcal{X} .
    \end{equation}
\end{definition}

We note that in practice (e.g., for numerical results), complete robustness is determined relative to the training set~\eqref{eqn:labeled-data-for-classifier}. That is, we empirically observe that $E$ is a completely robust data encoding if and only if
    \begin{equation}
        \mathcal{R}(\mathcal{E}, E, \yhat) = \{\vec{x}_i\}_{i = 1}^{M} .
    \end{equation}

Complete robustness can be a strong condition, so we also consider a partially robust data encoding, defined as follows.

\begin{definition}[Partially Robust Data Encoding]
    \label{def:partial_noise_robustness}
    Consider a (binary) quantum classifier with encoding $E$ and decision rule $\yhat$ as defined in~\eqref{eqn:decision_rule}. Let $\vec{x} \in \mathcal{X}$ and let $\mathcal{E}$ be a quantum channel. We say that $E$ is a \textit{partially robust data encoding} for the quantum classifier if and only if
    \begin{equation} \label{eqn:partially-robust-encoding-condition}
        \mathcal{R}(\mathcal{E}, E, \hat{y}) \subsetneq \mathcal{X} .
    \end{equation}
\end{definition}

Similar to complete robustness, partial robustness is determined in practice relative to the training set. For $0 \le \delta \le 1$, we say that $E$ is a $\delta$-robust data encoding if and only if
\begin{equation}\label{eqn:delta_robust_encoding_condition}
    | \mathcal{R}(\mathcal{E}, E, \hat{y})  | = \delta M 
\end{equation}
where $|\cdot|$ denotes cardinality so that $ | \mathcal{R}(\mathcal{E}, E, \hat{y}) | \in [M]$.

\section{Analytic Results} \label{sec:noise_robustness}

Using the definitions from Sec.~\ref{sec:definitions}, we now state and prove results about data encodings. First, we show that different encodings lead to different classes of decision boundaries in Sec.~\ref{ssec:classes_learnable_decision_boundaries}. Next, we characterize the set of robust points for example quantum channels in Sec.~\ref{ssec:characterize_robust_points}. In Sec.~\ref{subsec:robustness-results}, we prove several robustness results for different quantum channels, and in Sec.~\ref{subsec:existence-of-robust-encodings} we discuss the existence of robust encodings as well as an observed tradeoff between learnability and robustness.
Finally, in Sec.~\ref{ssec:fidelity_bounds}, we prove an upper bound on the number of robust points in terms of fidelities between noisy and noiseless states.

\subsection{Classes of Learnable Decision Boundaries} \label{ssec:classes_learnable_decision_boundaries}

We defined several different encodings in Sec.~\ref{sec:data-encodings} and discussed differences in the state preparation circuits which realize the encodings. Here, we show that different encodings lead to different sets of decision boundaries for the quantum classifier, thereby demonstrating that the success of the quantum classifier in Def.~\ref{def:binary-quantum-classifier} depends crucially on the data encoding~\eqref{eqn:qubit-encoding-general}. 

The decision boundary according to the decision rule~\eqref{eqn:decision_rule} is implicitly defined by
\begin{equation} \label{eqn:decision_boundary_defining rule}
     \Tr [ \Pi_0 \rhotildex] = 1 / 2 .
\end{equation}
Consider a single qubit encoding~\eqref{eqn:qubit-encoding-general} so that
\begin{equation}
    \rhox = \left[ \begin{matrix}
        f(x_1, x_2)^2 & f(x_1, x_2) g(x_1, x_2)^* \\
        f(x_1, x_2) g(x_1, x_2) & |g(x_1, x_2)|^2 
    \end{matrix} \right]
\end{equation}
where we assumed without loss of generality that $f$ is real valued.
Let the unitary $U$ such that $\rhotildex = U \rhox U^\dagger$ have matrix elements $U_{ij}$. Then, one can write the decision boundary~\eqref{eqn:decision_boundary_defining rule} as (see~\eqref{eqn:useful-00mtx-elt})
\begin{equation}\label{eqn:projector_expanded}
   |U_{00}|^2 f^2 + 2 \Re [ U_{00}^* U_{01} f g] + |U_{01}|^2 |g|^2 = 1 / 2
\end{equation}
where we have let $f = f(x_1, x_2)$ and $g = g(x_1, x_2)$ for brevity. Eqn.~\eqref{eqn:projector_expanded} implicitly defines the decision boundary in terms of the data encoding $f$ and $g$. The unitary matrix elements $U_{ij}$ act as hyperparameters to define a class of learnable decision boundaries. 

Eqn.~\eqref{eqn:projector_expanded} can be solved numerically for different encodings, and we do so in Sec.~\ref{ssec:classes-decision-boundaries-numerical} (Fig.~\ref{fig:random_decision_boundaries_encodings}) to visualize decision boundaries for single qubit classifiers. At present, we can proceed further analytically with a few inconsequential assumptions to simplify the equations.

For the wavefunction encoding, we have $f(x_1, x_2) = x_1$ and $g(x_1, x_2) = x_2$. Suppose for simplicity that matrix elements $U_{00} \equiv a$ and $U_{01} \equiv b$ are real. Then, Eqn.~\eqref{eqn:projector_expanded} can be written
\begin{equation} \label{eqn:learnable-decision-boundary-wavefunction}
    (a x_1 + b x_2)^2 = 1/2 ,
\end{equation}
which defines a line $x_2 = x_2(x_1)$ with slope $-a / b$ and intercept $1 / \sqrt{2} b$. Thus, a single qubit classifier in Def.~\ref{def:binary-quantum-classifier} which uses the wavefunction encoding~\eqref{eqn:wavefunction_encoding_amplitude} can learn decision boundaries that are straight lines. 

Now consider the dense angle encoding~\eqref{eqn:dae-general} on a single qubit, for which $f(x_1, x_2) = \cos(\pi x_1)$ and $g(x_1, x_2) = e^{2 \pi i x_2} \sin (\pi x_1)$. Supposing again that matrix elements $U_{00} \equiv a$ and $U_{01} \equiv b$ are real, we can write~\eqref{eqn:projector_expanded} as
\begin{multline}
    a^2 \cos^2 \pi x_1 + 2 a b \cos \pi x_1 \sin \pi x_1 \cos 2 \pi x_2 \\+ b^2 \sin^2 \pi x_1 = 1/2 .
\end{multline}
This can be rearranged to
%
\begin{equation} \label{eqn:learable_decision_boundary_dae}
    \cos 2 \pi x_2 = \frac{1 - 2 a^2  + (2a^2 -  2b^2) \sin^2 \pi x_1}{a b \sin 2 \pi x_1} ,
\end{equation}
which defines a class of sinusoidal functions $x_2 = x_2(x_1)$. (See Sec.~\eqref{ssec:classes-decision-boundaries-numerical} and Fig.~\ref{fig:random_decision_boundaries_encodings}.) 

The different decision boundaries defined by~\eqref{eqn:learnable-decision-boundary-wavefunction} and~\eqref{eqn:learable_decision_boundary_dae} emphasize the effect that encoding has on learnability. A classifier may have poor performance due to its encoding, and switching the encoding may lead to better results. We note that a similar phenomenon occurs in classical machine learning --- a standard example being that a dot product kernel cannot separate data on a spiral, but a Gaussian kernel can. 
It may not be clear \textit{a priori} what encoding to use (similarly in classical machine learning with kernels), but different properties of the data may lead to educated guesses. We note that Lloyd \textit{et al.}~\cite{lloyd_quantum_2020} consider training over hyperparameters to find good encodings, and we introduce a similar idea in \secref{ssec:encoding_learn_alg} to find good \emph{robust} encodings.

In Sec.~\ref{ssec:classes-decision-boundaries-numerical}, we numerically evaluate decision boundaries for additional single-qubit encodings, as well as two-qubit encodings, to further illustrate the differences that arise from different encodings.


\subsection{Characterization of Robust Points} \label{ssec:characterize_robust_points}

For a given quantum channel $\mathcal{E}$, it is a standard exercise to characterize the set of density operator fixed points, i.e., states $\rho \in \mathcal{D}_n$ such that
\begin{equation} \label{eqn:fixed-point-definition}
    \mathcal{E}(\rho) = \rho .
\end{equation}
In this section, we characterize the set of robust points for example quantum channels. This demonstrates the relationship between robust points and fixed points which we further elaborate on in Sec.~\ref{subsec:existence-of-robust-encodings}. We remark that the characterizations similar to the ones in this Section may be of independent interest from a purely theoretical perspective, as robust points can be considered a type of generalized fixed point, or symmetry, of quantum channels. 

The pure states which are fixed points of the dephasing channel~\eqref{eqn:dephasing-channel} are $\Pi_0 := |0\>\<0|$ and $\Pi_1 := |1\>\<1|$, and
\begin{equation} \label{eqn:fixed-points-of-dephasing-channel}
    \rho = a \Pi_0 + b \Pi_1 
\end{equation}
with $a + b = 1$ is the general mixed-state density operator fixed point. In contrast, let us now consider the robust points of the same dephasing channel, which satisfy
\begin{equation} \label{eqn:robust-condition-for-dephasing}
    \yhat[ \deph (\rho) ] = \yhat [\rho]
\end{equation}
instead of~\eqref{eqn:fixed-point-definition}. Certainly the state in Eqn.~\eqref{eqn:fixed-points-of-dephasing-channel} will satisfy~\eqref{eqn:robust-condition-for-dephasing} --- i.e., any fixed point is a robust point --- but the set of robust points may contain more elements. To completely characterize the robust set, we seek the set of $\rho \in \mathcal{D}_2$ such that
\begin{equation} \label{eqn:robust-condition-for-dephasing-inequality1}
    \Tr [ \Pi_0 \rho ] \ge 1/2 \implies \Tr[ \Pi_0 \deph (\rho)] \ge 1 / 2
\end{equation}
and
\begin{equation} \label{eqn:robust-condition-for-dephasing-inequality2}
    \Tr [ \Pi_0 \rho ] < 1/2 \implies \Tr[ \Pi_0 \deph (\rho)] < 1 / 2 .
\end{equation}
Using simple properties of the trace and Pauli matrices (see Appendix~\ref{app:useful-formulae} if desired), we can write
\begin{equation}
    \Tr[\Pi_0 \deph(\rho)] = (1 - p) \Tr[\Pi_0 \rho] + p \Tr[\Pi_0 Z \rho Z] = \Tr[ \Pi_0 \rho ]. 
\end{equation}
Thus~\eqref{eqn:robust-condition-for-dephasing-inequality1} and~\eqref{eqn:robust-condition-for-dephasing-inequality2} are satisfied for all density operators $\rho \in \mathcal{D}_{2}$.
That is, every data point $\vec{x} \in \mathcal{X}$ is a robust point of the dephasing channel (independent of the encoding) for the quantum classifier in Def.~\ref{def:binary-quantum-classifier}. 

Consider now an amplitude damping channel~\eqref{eqn:amp-damp-channel} with $p = 1$, for which the only fixed point is the pure state $\Pi_0$. By evaluating
\begin{equation}
    \Tr [ \Pi_0 \damp (\rho) ] = (1 - p) \Tr [ \Pi_0 \rho ] + p ,
\end{equation}
we see that a robust point $\sigma$ must satisfy $\Tr[\Pi_0 \sigma] = 1$. That is, the only robust point is $\Pi_0$, and in this case the set of robust points is identical to the set of fixed points.




The previous two examples illustrate how to find the robust points of a quantum channel, and the relationship between robust points and fixed points for the given channels. As expected from~\eqref{eqn:robust-point-definition} and~\eqref{eqn:fixed-point-definition}, these examples confirm that
\begin{equation} \label{eqn:fixed-points-are-subset-of-robust-points}
    \mathcal{F}(\mathcal{E}) \subseteq \mathcal{R}(E, \mathcal{E}, \hat{y}) 
\end{equation}
where $\mathcal{F}(\mathcal{E})$ denotes the set of fixed points of $\mathcal{E}$. In Sec.~\ref{subsec:existence-of-robust-encodings}, we use this connection to generalize the above discussion and prove the existence of robust data encodings.

\subsection{Robustness Results} \label{subsec:robustness-results}


In this section, we state and prove results on robust encodings. In particular, we prove robustness results for Pauli, depolarizing and amplitude damping channels, given certain conditions on the noise parameters in each channel.

First, we consider when robustness can be achieved for a Pauli channel.

\begin{theorem} \label{thm:robustness_pauli_noise_xy}
    Let $\paul$ be a Pauli channel~\eqref{eqn:pauli-channel} and consider a quantum classifier on data from the set $\mathcal{X}$. Then, for any encoding $E: \mathcal{X} \rightarrow \mathcal{D}_2$, we have complete robustness
    \begin{equation}
        \mathcal{R} (\paul, E, \yhat) = \mathcal{X}
    \end{equation}
    if $p_X + p_Y \le 1/2$. (Recall that $\vec{p} = [p_I, p_X, p_Y, p_Z]$.)
\end{theorem}



\begin{proof}
    The predicted label in the noisy case is identical to~\eqref{eqn:decision_rule} with $\rhotildex$ replaced by $\paul ( \rhotildex )$. That is,
    \begin{equation} \label{eqn:classification-scheme-pauli-noise}
        \yhat [ \paul (\rhotildex) ] = \begin{cases}
            0 \qquad \text{if   } \Tr [ \Pi_0 \paul( \rhotildex )] \ge 1 / 2 \\
            1 \qquad \text{otherwise} 
        \end{cases} .
    \end{equation}
    By definition~\eqref{eqn:pauli-channel}, we have
    \begin{equation} \label{eqn:pauli-proof1}
        \begin{split}
            \Tr [ \Pi_0  \paul( \rhotildex )] = p_I & \Tr[ \Pi_0 \rhotildex] + p_X \Tr [ \Pi_0 X \rhotildex X] \\
            + \ p_Y \Tr& [ \Pi_0 Y \rhotildex Y] + p_Z \Tr [ \Pi_0 Z \rhotildex Z] .
        \end{split}
    \end{equation}
    Using straightforward substitutions (Appendix~\ref{app:useful-formulae}),
    we may write~\eqref{eqn:pauli-proof1} as
    \begin{equation}
        \Tr [ \Pi_0  \paul( \rhotildex )] = (p_I + p_Z) \Tr [ \Pi_0 \rhotildex ] + (p_X + p_Y) \Tr[ \Pi_1 \rhotildex ] .
    \end{equation}
    By resolution of the identity
    \begin{equation}
        1 = \Tr [ \rhotildex ] = \Tr[\Pi_0 \rhotildex] + \Tr[ \Pi_1 \rhotildex ] ,
    \end{equation}
    we come to the simplified expression
    \begin{equation}
        \Tr [ \Pi_0  \paul( \rhotildex )] = \left[ 1 - 2 \nu \right] \Tr [ \Pi_0 \rhotildex ] + \nu .
    \end{equation}
    where $\nu := p_X + p_Y$. 
    
    Suppose the noiseless classification is $\yhat = 0$ so that $\Tr[ \Pi_0 \rhotildex ] \ge 1 / 2$. Since $\nu \le 1/2$, we have
    \begin{equation}
        \Tr [ \Pi_0  \paul( \rhotildex )] \ge \left[ 1 - 2 \nu \right] \frac{1}{2} + \nu  = \frac{1}{2}
    \end{equation}
    Hence, classification of data points with label $\yhat = 0$ is robust for any encoding. 
    
    Suppose the noiseless classification is $\yhat = 1$ so that $\Tr[ \Pi_0 \rhotildex ] < 1 / 2$. Since $\nu \le 1/2$, we have
    \begin{equation}
        \Tr [ \Pi_0  \paul( \rhotildex )] < \left[ 1 - 2 \nu \right] \frac{1}{2} + \nu = \frac{1}{2}.
    \end{equation}
    Hence, classification of data points with label $\yhat = 1$ is also robust for any encoding. 
\end{proof}

Returning to the condition, $p_X + p_Y \le 1/2$, one can imagine a NISQ computer in which either $p_X$ or $p_Y$ were large enough such that this condition is not satisfied. In this regard, we note two things. The first is that if this condition is \emph{not} satisfied, then not every encoding strategy will be robust to the Pauli channel in this model. In particular, the set of robust points will now be \textit{dependent} on the encoding strategy. This is similar to the behavior of the amplitude damping channel (which we demonstrate shortly), and we illustrate in \secref{sec:numerical_results}. Secondly, the requirement $p_X + p_Y \le 1/2$ appears because the decision rule uses a measurement in the computational basis. In this case, we can still achieve robustness by using a modified decision rule which measures in a different basis. 

\begin{corollary} \label{corr:pauli_x_robustness}
    Consider a quantum classifier on data from the set $\mathcal{X}$ with modified decision rule
    \begin{equation} \label{eqn:classification-scheme-measure-hadamard-basis}
        \hat{z} [ \rhotildex ] = \begin{cases}
            0 \qquad \text{if   } \Tr [ \Pi_+ \rhotildex ] \ge 1 / 2 \\
            1 \qquad \text{otherwise} 
        \end{cases} .
    \end{equation}
    Here, $\Pi_+ := |+\> \< + |$ is the projector onto the $+1$ eigenstate $|+\>$ of Pauli $X$. Then, for any $E: \mathcal{X} \rightarrow \mathcal{D}_2$,
    \begin{equation}
        \mathcal{R} (\paul, E, \hat{z}) = \mathcal{X}
    \end{equation}
    for a Pauli channel $\paul$ such that $p_Y + p_Z \le 1/2$ .
\end{corollary}

The proof mimics that of Theorem~\ref{thm:robustness_pauli_noise_xy}. We note that the analogous statement for measurements in the $Y$-basis also holds.
These results suggest that device-specific encoding strategies may be important for achieving robustness in practice on NISQ computers. 

Theorem~\ref{thm:robustness_pauli_noise_xy} also implies the following result for dephasing noise, which is a Pauli channel with $p_X = p_Y = 0$.
\begin{theorem} \label{thm:robustness-single-qubit-dephasing-noise}
    Let $\deph$ be a dephasing channel~\eqref{eqn:dephasing-channel}, and consider a quantum classifier on data from the set $\mathcal{X}$. Then, for any encoding $E: \mathcal{X} \rightarrow \mathcal{D}_2$,
    \begin{equation}
        \mathcal{R} (\deph, E, \yhat) = \mathcal{X} .
    \end{equation}
\end{theorem}
This result follows immediately from the discussion of the dephasing channel in the above Section \ref{ssec:characterize_robust_points}.

Similar to Corollary~\ref{corr:pauli_x_robustness}, we can consider a modified decision rule to achieve robustness for a bit-flip channel.

\begin{corollary} \label{corr:bit-flip-robustness}
    Consider a quantum classifier on data from the set $\mathcal{X}$ with modified decision rule $\hat{z}$ defined in Eqn.~\eqref{eqn:classification-scheme-measure-hadamard-basis}. Then, for any encoding $E: \mathcal{X} \rightarrow \mathcal{D}_2$,
    \begin{equation}
        \mathcal{R} (\flip, E, \hat{z}) = \mathcal{X} .
    \end{equation}
\end{corollary}

We note that a decision rule which measures in the $Y$-basis yields robustness to combined bit/phase-flip errors. (That is, the error channel $\mathcal{E}(\rho) = (1 - p) \rho + p Y \rho Y$.) v

We now consider robustness for depolarizing noise~\eqref{eqn:depolarizing-channel}. A simple calculation shows that
\begin{equation}
    \Tr [ \Pi_0 \depo (\rhotildex ) ] = p / 2 + (1 - p) \Tr [ \Pi_0 \rho ] .
\end{equation}
If $\yhat = 0$ so that $\Tr[ \Pi_0 \rho ] \ge 1/2$, then we have that $\Tr [ \Pi_0 \depo (\rhotildex ) ] \ge 1/2$. Similarly for the case $\yhat = 1$. Thus, we have shown the following.

\begin{theorem} \label{thm:robustness_to_depolarizing_on_unitary_ansatz_single_qubit}
    Let $\depo$ be a depolarizing channel~\eqref{eqn:depolarizing-channel}, and consider a quantum classifier on data from the set $\mathcal{X}$. Then, for any encoding $E: \mathcal{X} \rightarrow \mathcal{D}_2$,
    \begin{equation}
        \mathcal{R} (\depo, E, \yhat) = \mathcal{X} .
    \end{equation}
\end{theorem}

We remark that Theorem~\ref{thm:robustness_to_depolarizing_on_unitary_ansatz_single_qubit} holds with measurements in any basis, not just the computational basis. Further, we will soon generalize this result to (i) multi-qubit classifiers and (ii) noisy data encoding~\eqref{eqn:robust-point-noisy-encoding}.

We now consider amplitude damping noise, for which the robust set $\mathcal{R}$ depends on the encoding $E$. From the channel definition~\eqref{eqn:amp-damp-channel}, it is straightforward to see that
\begin{equation} \label{eqn:amp-damp-simple-expansion}
    \Tr [ \Pi_0 \damp ( \rhotildex )] = \Tr [ \Pi_0 \rhotildex ] + p \Tr [ \Pi_1 \rhotildex ] .
\end{equation}
Suppose first that the noiseless prediction is $\yhat = 0$ so that $\Tr [ \Pi_0 \rhotildex ] \ge 1/2$. Then, certainly $\Tr [ \Pi_0 \damp ( \rhotildex )] \ge 1/2$ because $p \ge 0$ and $\Tr [ \Pi_1 \rhotildex ] \ge 0$. Thus, the noisy prediction is always identical to the noiseless prediction when the noiseless prediction is $\yhat = 0$. This can be understood intuitively because an amplitude damping channel models the $|1\rangle \mapsto |0\rangle$ transition \cite{john_preskill_quantum_1998} which only increases the probability of the ground state. 

Suppose now that the noiseless prediction is $\yhat = 1$. From~\eqref{eqn:amp-damp-simple-expansion}, we require that
\begin{equation}
    \Tr [ \Pi_0 \damp ( \rhotildex )] = \Tr [ \Pi_0 \rhotildex ] + p \Tr [ \Pi_1 \rhotildex ] < 1/2
\end{equation}
to achieve robustness. We use resolution of the identity
\begin{equation}
    \Tr[ \Pi_1 \rhotildex] = 1 - \Tr [ \Pi_0 \rhotildex] 
\end{equation}
to arrive at the condition
\begin{equation} \label{eqn:amp-damp-robustness-condition}
    \Tr [ \Pi_1 \rhotildex ] > \frac{1}{2(1 - p)} .
\end{equation}
Let $\rhox$ be given by the general qubit encoding~\eqref{eqn:qubit-encoding-general} so that~\eqref{eqn:amp-damp-robustness-condition} can be written (see~\eqref{eqn:useful-11mtx-elt})
\begin{equation*}
    |U_{10}|^2 f^2 + 2 \Re [ U_{11}^* U_{10} f g^*] + |U_{11}|^2 |g|^2 > \frac{1}{2(1 - p)}
\end{equation*}
where $U_{ij}$ denote the optimal unitary matrix elements.

We have thus shown the following.
\begin{theorem} \label{thm:robustness-amp-damp}
    Consider a quantum classifier on data from the set $\mathcal{X}$, and let $\damp$ denote the amplitude damping channel~\eqref{eqn:amp-damp-channel}. Then, for any qubit encoding $E$ defined in~\eqref{eqn:qubit-encoding-general} which satisfies
    \begin{equation}\label{eqn:amp_damp_robustness_condition_expanded_unitary}
        |U_{10}|^2 f^2 + 2 \Re [ U_{11}^* U_{10} f g^*] + |U_{11}|^2 |g|^2 > \frac{1}{2(1 - p)} ,
    \end{equation}
    we have
    \begin{equation}
        \mathcal{R} (\paul, E, \yhat) = \mathcal{X} .
    \end{equation}
    If $E$ is not completely robust, the set of points $\vec{x}$ such that that~\eqref{eqn:amp_damp_robustness_condition_expanded_unitary} holds define the partially robust set.
\end{theorem}

We note that~\eqref{eqn:amp_damp_robustness_condition_expanded_unitary} depends on the optimal unitary $U$ as well as the encoding $E$. This is expected as the final state $\rhotildex$ has been processed by the QNN. In practice, since we do not know the optimal unitary parameters \textit{a priori}, it remains a question of how large the (partially) robust set will for a given an encoding. To address this point, we discuss in Sec.~\ref{ssec:encoding_learn_alg} how training over hyperparameters in the encoding function can help find the robust region even after application of the \textit{a priori} unknown optimal unitary. Additionally, in the next Section we discuss whether we can find an encoding which satisfies~\eqref{eqn:amp_damp_robustness_condition_expanded_unitary}, or more generally whether a robust encoding exists for a given channel. 

Given the robustness condition~\eqref{eqn:amp_damp_robustness_condition_expanded_unitary} for the amplitude damping channel, it is natural to ask whether such an encoding exists. In Sec.~\ref{subsec:existence-of-robust-encodings}, we show the answer is yes by demonstrating there always exists a robust encoding for any trace preserving quantum operation. This encoding may be trivial, which leads to the idea of a tradeoff between learnability and robustness. (See Sec.~\ref{subsec:existence-of-robust-encodings}.)

We now consider global depolarizing noise on a multi-qubit classifier. It turns out that any encoding is completely robust to this channel applied at any point throughout the circuit. To clearly state the theorem, we introduce the following notation. First, let 
\begin{equation} \label{eqn:short-notation-depo-channel}
    \E_{p_i} (\rho) = p_i \rho + (1 - p_i) I_d / d 
\end{equation}
be shorthand for a global depolarizing channel with probability $p_i$. (Note $p_i$ and $1 - p_i$ are intentionally reversed compared to Def.~\ref{def:global-depolarizing-channel} to simplify the proof.) Then, let
\begin{equation} \label{eqn:depo-robustness-notation-of-unitary-noise-application}
    \tilde{\rho}_{\vec{x}}^{(m)} \equiv \left[ \prod_{i = 1}^{m} U_i \circ \E_{p_i} \right] \circ \rhox
\end{equation}
denote the state of the encoded point $\rhox$ after $m$ applications of a global depolarizing channel and unitary channel. For instance, $m = 1$ corresponds to 
\begin{equation*}
    U_1 \circ \E_{p_1} \circ \rhox \equiv U_1( \E_{p_1} ( \rhox ) )
\end{equation*}
and $m = 2$ corresponds to
\begin{equation*}
    U_2 \circ \E_{p_2} \circ U_1 \circ \E_{p_1} \circ \rhox \equiv U_2 ( \E_{p_2} ( U_1( \E_{p_1} ( \rhox ) ) ) ) .
\end{equation*}
We remark that $U_i$ can denote any unitary in the circuit.

With this notation, we state the theorem as follows.

\begin{theorem} \label{thm:robustness_to_depolarizing_on_unitary_ansatz_multiple_qubits}
    Consider a quantum classifier on data from the set $\mathcal{X}$ with decision rule $\hat{y}$ defined in Eqn.~\eqref{eqn:classification-scheme-measure-hadamard-basis}. Then, for any encoding $E: \mathcal{X} \rightarrow \mathcal{D}_n$,
    \begin{equation}
        \mathcal{R} \left(\depon , E, \hat{y} \right) = \mathcal{X} .
    \end{equation}
    where $\depon$ denotes the composition of global depolarizing noise acting at any point in the circuit --- i.e., such that the final state of the classifier is given by~\eqref{eqn:depo-robustness-notation-of-unitary-noise-application}. 
\end{theorem}

To prove Theorem~\ref{thm:robustness_to_depolarizing_on_unitary_ansatz_multiple_qubits}, we use the following lemma.

\begin{lemma} \label{lem:lemma-for-global-depo-proof}
    The state in Eqn.~\eqref{eqn:depo-robustness-notation-of-unitary-noise-application} can be written as (adapted from~\cite{sharma_noise_2020})
    \begin{equation} \label{eqn:lemma-for-global-depo-proof}
        \tilde{\rho}_{\vec{x}}^{(m)} = 
        \prod_{i = 1}^{m} p_i U_m \cdots U_1 \rhox U_1^\dagger \cdots U_m^\dagger + \left( 1 - \prod_{i = 1}^{m} p_i \right) \frac{I_d}{d}
    \end{equation}
    where $d = 2^n$ is the dimension of the Hilbert space. 
\end{lemma}

\begin{proof}
    Using the definition of the global depolarizing channel~\eqref{eqn:short-notation-depo-channel}, it is straightforward to evaluate
    \begin{equation*}
        \tilde{\rho}_{\vec{x}}^{(1)} = U_1 \circ \E_{p_1} \circ \rhox = p_1 U_1 \rhox U_1^\dagger + (1 - p_1) I_d / d . 
    \end{equation*}
    Thus~\eqref{eqn:lemma-for-global-depo-proof} is true for $m = 1$. Assume~\eqref{eqn:lemma-for-global-depo-proof} holds for $m = k$. Then, for $k + 1$ we have
    \begin{align*}
        \tilde{\rho}_{\vec{x}}^{(k + 1)} &= U_{k + 1} \circ \E_{p_{k + 1}} \circ \tilde{\rho}_{\vec{x}}^{(k)} \\
        &= p_{k + 1} U_{k + 1} \tilde{\rho}_{\vec{x}}^{(k)} U_{k + 1}^\dagger + (1 - p_{k + 1}) I_d / d .
    \end{align*}
    The last line can be simplified to arrive at
    \begin{align}
        \tilde{\rho}_{\vec{x}}^{(k + 1)} = 
        \prod_{i = 1}^{k + 1} p_i U_{k + 1} \cdots U_1 \rhox U_1^\dagger \cdots U_{k + 1}^\dagger \nonumber \\ 
        + \left( 1 - \prod_{i = 1}^{k + 1} p_i \right) I / d , \nonumber
    \end{align}
    which completes the proof.
\end{proof}

We can now prove Theorem~\ref{thm:robustness_to_depolarizing_on_unitary_ansatz_multiple_qubits} as follows. Let $l$ denote the total number of alternating unitary gates with depolarizing noise in the classifier circuit so that~\eqref{eqn:lemma-for-global-depo-proof} can be written
\begin{equation} \label{eqn:final-state-global-depo-circuit}
    \tilde{\rho}_{\vec{x}}^{(l)} = \bar{p} \rhotildex + (1 - \bar{p}) I / d.
\end{equation}
Here, we have let $\bar{p} := \prod_{i = 1}^{l} p_i$ and noted that $U_l \cdots U_1 \rhox U_1^\dagger \cdots U_l^\dagger = \rhotildex$ is the final state of the noiseless circuit before measuring. Eqn.~\eqref{eqn:final-state-global-depo-circuit} is thus the final state of the noisy circuit before measuring. We can now evaluate
\begin{align}
    \Tr [ \Pi_0 \tilde{\rho}_{\vec{x}}^{(l)} ] = \bar{p} \Tr[ \Pi_0 \rhotildex ] + (1 - \bar{p}) / 2
\end{align}
where we have used $\Tr[ \Pi_0 I_d ] = 2^{d - 1}$. 
To prove robustness, suppose that $\yhat [ \rhotildex ] = 0$ so that $\Tr[ \Pi_0 \rhotildex ] \ge 1/2$. Then,
\begin{equation}
    \Tr [ \Pi_0 \tilde{\rho}_{\vec{x}}^{(l)} ] \ge \bar{p} / 2 + (1 - \bar{p}) / 2 = 1 / 2 
\end{equation}
so that $\yhat [ \tilde{\rho}_{\vec{x}}^{(l)} ] = 0$. Similarly for the case $\yhat[ \rhotildex ] = 1$, which completes the proof of Theorem~\ref{thm:robustness_to_depolarizing_on_unitary_ansatz_multiple_qubits}.


%

Thus, any encoding strategy exhibits complete robustness to global depolarizing noise. We remark again (see footnote on Page~\pageref{footnote:shots}) that our definition of robustness (Def.~\ref{def:robust-point}) is in terms of probability, meaning that more measurements for sampling may be required to reliably evaluate robustness. With this remark, we note an interesting connection to explain a phenomenon observed in recent literature: In Ref.~\cite{grant_hierarchical_2018}, the authors found that classification accuracy decreased under the presence of depolarizing noise. 
Theorem~\ref{thm:robustness_to_depolarizing_on_unitary_ansatz_multiple_qubits} implies this was a feature exclusively of the finite shot noise used to obtain the predicted label.


While global depolarizing noise admits a clean robustness result for an arbitrary $d$-dimensional circuit, general channels can lead to complicated equations which are best handled numerically. We include several numerical results in Sec.~\ref{sec:numerical_results}, and we discuss avenues for proving more analytical results with certain classes of channels in future work in Sec.~\ref{sec:conclusions}. To close the present discussion, we highlight the special case of multi-qubit classifiers with ``factorizable noise,'' for which it is straightforward to apply previous results proved in this section.

In particular, suppose that $\E : \mathcal{D}_n \rightarrow \mathcal{D}_n$ is a noise channel which factorizes into single qubit channels, e.g.
\begin{equation} \label{eqn:factorizable-noise-channel-into-single-qubits-definition}
    \E = \E_1 \otimes \cdots \otimes \E_n
\end{equation}
where $\E_i : \mathcal{D}_2 \rightarrow \mathcal{D}_2$ for $i \in [n]$. Without loss of generality, let the classification qubit be the $n$th qubit. Then, if the processed state of the classification qubit is robust to the channel $\E_n$, the encoded state will be robust to the entire channel $\E$ in~\eqref{eqn:factorizable-noise-channel-into-single-qubits-definition}. This result, which is precisely stated and proved in Appendix~\ref{app:factorizable-noise}, also holds for general $n - 1$ qubit channels which act on every qubit except the classification qubit. Although this is relatively straightforward, the result could be used as a building block to better understand more intricate robustness properties of quantum classifiers.


\subsection{Existence of Robust Encodings} \label{subsec:existence-of-robust-encodings}



In Sec.~\ref{ssec:characterize_robust_points}, we considered example channels and characterized their robust points and fixed points. We found that the set of fixed points $\mathcal{F}(\mathcal{E})$ is always a subset of the robust set $\mathcal{R}(\mathcal{E}, E, \yhat)$ in~\eqref{eqn:fixed-points-are-subset-of-robust-points}. Here, we use this connection to show that there always exists a robust encoding for a trace-preserving channel $\mathcal{E}$ (regardless of optimal unitary parameters which may appear in the robustness condition, e.g.~\eqref{eqn:amp_damp_robustness_condition_expanded_unitary}).

\begin{theorem}(Existence of Fixed Points~\cite{schauder_fixpunktsatz_1930, nielsen_quantum_2010}) \label{thm:schauder_fixed_point}
Any trace-preserving quantum operation has at least one density operator fixed point~\eqref{eqn:fixed-point-definition}.
\end{theorem}

Using this and the observation that $\mathcal{F}(\mathcal{E}) \subset \mathcal{R}(\mathcal{E}, E, \yhat)$, we have the following existence theorem for robust encodings.
\begin{theorem}
\label{thm:robust_encodings_existence}
    Given a data point $\vec{x} \in \mathcal{X}$, a trace-preserving quantum channel $\mathcal{E}$, and decision rule $\yhat$ defined in~\eqref{eqn:decision_rule}, there exists an encoding $E$ such that
    \begin{equation} \label{eqn:existenstence-robust-encoding}
        \yhat [ \mathcal{E} ( E(\vec{x}) ) ] = \yhat[ E(\vec{x}) ] .
    \end{equation}
%
\end{theorem}

We note that the optimal unitary of the QNN affects the ``location'' of the robust set, but not the existence. 

We emphasize that Theorem~\ref{thm:robust_encodings_existence} is with respect to a single data point $\vec{x} \in \mathcal{X}$. As mentioned in Sec.~\ref{subsec:robustness-definition}, it is more relevant for applications to consider the training set~\eqref{eqn:labeled-data-for-classifier} or entire set $\mathcal{X}$. 
Appropriately, one can ask whether a completely robust encoding (Def.~\ref{def:complete_noise_robustness}) exists for a given channel $\E$. This answer also turns out to be yes, but in a potentially trivial way.

In particular, suppose that there is a unique fixed point $\sigma$ of the channel $\E$, e.g. depolarizing noise or amplitude damping noise with $p = 1$. Then, consider the encoding
\begin{equation}
    \E (\vec{x}) = \sigma 
\end{equation}
for all $\vec{x} \in \mathcal{X}$. From a robustness perspective, this has the desirable property of complete robustness. From a machine learning perspective, however, this has very few desirable properties: all training data is mapped to the same point so that it is impossible to successfully train a classifier%
\footnote{In principle, one can achieve an encoding which is completely robust and able to correctly classify all data if there are at least two orthogonal fixed points in $\mathcal{F}(\E)$. For example, if $\E$ the bit flip channel, the encoding $\vec{x}_i \mapsto \ket{0} + (-1)^{y_i} \ket{1}$ is both completely robust and completely learnable (the optimal unitary is a Hadamard gate), but assumes the true labels $y_i$ are known.}
.

The previous example, while extreme, serves to illustrate the tradeoff between learnability and robustness. By ``learnability,'' we mean the ability of the classifier to predict correct labels (without regard to noise), and by robustness we mean that the prediction is the same with or without noise (without regard to correctness). The two links are schematically connected below:
\begin{equation} \label{eqn:learnability-robustness-links}
    y[\vec{x}] \ \ \xleftrightarrow{\text{Learnability}} \ \ \yhat [ \rhotildex ] \ \
    \xleftrightarrow{\text{Robustness}} \ \ \yhat [ \E( \rhotildex ) ] 
\end{equation}
The tradeoff we observe is that more learnability comes at the price of less robustness, and \textit{vice versa}. See Sec.~\ref{ssec:encoding_learn_alg} for a discussion. 
\subsection{Upper Bounds on Partial Robustness} \label{ssec:fidelity_bounds}

In this section, we consider a slightly modified binary quantum classifier which embeds the cost function in the circuit and computes the cost by measuring expectation values. In contrast to the classifier in Def.~\ref{def:binary-quantum-classifier}, the output of this circuit is thus the cost $C$ instead of an individual predicted label $\yhat$. Correspondingly, the input to the circuit is all data points in the training set~\eqref{eqn:labeled-data-for-classifier} (using a ``mixed state encoding'' discussed below) instead of a single data point $\vec{x}$. Such a classifier was recently introduced
by Cao \textit{et al.} in Ref.~\cite{cao_cost_2019} and presents an interesting framework to analyze in the context of noise, which we do in this Section.

Since the output of the circuit is the cost $C$ for all points instead of a predicted label $\yhat$ for an individual point, the definition of a single robust point does not immediately apply to this classifier. However, it is still natural to compare the noisy and noiseless outcomes --- in the same spirit as robustness --- by comparing the difference between the output cost $C_\E$ when some noise channel $\E$ occurs in the circuit to the output cost $C$ from an ideal (noiseless) circuit. In fact, we show this quantity provides an upper bound on the size of the partially robust set and therefore can be used as a proxy to assess robustness of different encodings.

To do so, we consider the indicator cost function
\begin{equation} \label{eqn:indicator_cost_over_dataset}
    C := \frac{1}{M} \sum_{i = 1}^{M}  \mathcal{I} (\hat{y}_i(\rhotildexi) \neq y_i) .
\end{equation}
Here, the indicator $\mathcal{I}$ evaluates to the truth value of its argument --- i.e., $\mathcal{I}(\yhat_i \neq  y_i) = 0$ if $y_i = \yhat_i$, else $1$.
We note again that $C = C(\vec{\alpha})$ is parameterized by some angles $\vec{\alpha}$ but we omit $\vec{\alpha}$ for brevity.

The indicator cost function~\eqref{eqn:indicator_cost_over_dataset} relates naturally to the robust set in Def.~\ref{def:robust-set}. Even though we cannot say individually which points are robust, a decrease in the cost due to some noise channel implies that some points were misclassified (assuming we had perfect classification, in the absence of the channel). Hence, how much the noisy cost function decreases is a useful proxy of robustness. We quantify this as
\begin{equation} \label{eqn:change-in-cost}
    \Delta_\E C := | C_\E - C |.
\end{equation}
%

In the encoding strategy of Cao \textit{et al}~\cite{cao_cost_2019}, each feature vector $\vec{x}$ is encoded along with its true label $y$ on an ancilla qubit as per
\begin{equation} \label{eqn:pure-state-of-mixed-state-encoding-for-embedded-cost-classifier}
    \sigmax = E(\vec{x}) \otimes |y\>\<y| = \rhox \otimes |y\>\<y| .
\end{equation}
where $E$ is an encoding function. Then, the entire dataset~\eqref{eqn:labeled-data-for-classifier} is prepared in the mixed state
\begin{equation}
    \sigma = \frac{1}{M} \sum_{i = 1}^{M} \rhoxi \otimes |y_i\>\<y_i| .
\end{equation}
Such a mixed state encoding may not be reliably preparable on NISQ computers, but in principle could be prepared using purification or by probabilistically preparing one of the pure states~\eqref{eqn:pure-state-of-mixed-state-encoding-for-embedded-cost-classifier} --- which could be amenable to NISQ computers depending on the encoding $E$. The QNN acts only on the ``data subsystem'' so that the evolved state before measurement is
\begin{equation}\label{eqn:mixed_state_over_data}
    \tilde{\sigma} = \frac{1}{M} \sum_{i = 1}^{M} \rhotildexi \otimes |y_i\>\<y_i| 
\end{equation}
%

We now consider the application of a noisy channel $\E$ so that
\begin{equation} \label{eqn:mixed_state_over_data_noisy_with_label}
    \E (\tilde{\sigma} ) = \frac{1}{M} \sum_{i=i}^M \E \left( \rhotildexi \otimes \ketbra{y_i}{y_i} \right) .
\end{equation}
While $\E$ could most generally act on the entire system, to match with previous analyses we assume that
\begin{equation} \label{eqn:mixed_state_over_data_noisy}
    \E (\tilde{\sigma} ) = \frac{1}{M} \sum_{i=i}^M \E \left( \rhotildexi \right) \otimes \ketbra{y_i}{y_i} 
\end{equation}
for simplicity. That is, we assume the true labels are invariant with respect to the noise channel%
\footnote{
Interestingly, the case where the true label is corrupted by noise can be linked to a commonly studied case in classical supervised learning --- i.e., ``learning from noisy examples''~\cite{angluin_learning_1988}. 
}%
.

The cost\footnote{Note this actually gives a slightly more general cost function than the indicator we use here~\eqref{eqn:indicator_cost_over_dataset}, but can be related by a simple transformation. See Ref.~\cite{cao_cost_2019}.} $C$ can be evaluated by measuring the expectation of
\begin{equation} \label{eqn:observable_cost}
    D := I^{\otimes n-1} \otimes Z_c \otimes Z_l 
\end{equation}
where $c$ and $l$ denote classification and label qubits, respectively. (See Ref.~\cite{cao_cost_2019} for more details.) That is, the (noiseless) cost is given by
\begin{equation} \label{eqn:quantum_cost}
    C = \Tr(D \tilde{\sigma}) ,
\end{equation}
and the noisy cost is identical with $\tilde{\sigma}$ replaced by $\E (\tilde{\sigma})$. 

We can now evaluate the change in cost due to noise~\eqref{eqn:change-in-cost} as (following \cite{gentini_noise-assisted_2019})
%
\begin{align}
    \Delta_\E C &:= \left| C_\E - C \right| \nonumber \\[1.0ex]
    &= \left|  \Tr [ D ( \E (\tilde{\sigma} ) - \tilde{\sigma}) ] \right| \nonumber \\[1.0ex]
    &\le ||D||_\infty || \E (\tilde{\sigma}) - \tilde{\sigma} ||_1 \nonumber \\[1.0ex]
    &\le 2 \sqrt{ 1 - F(\E (\tilde{\sigma}), \tilde{\sigma}) } . \label{eqn:fidelity_bound_mixed_state}
\end{align}
Here, $F$ is the fidelity of states $\tau, \omega \in \mathcal{D}_n$ defined by
\begin{equation}\label{eqn:fidelity_defintion}
    F(\tau, \omega) := \Tr \left[ \sqrt{ \sqrt{ \tau } \omega \sqrt {\tau } } \right] ^2 .
\end{equation}
The third line in this derivation follows from H\"{o}lders inequality and the last line from the Fuchs-van de Graaf inequality~\cite{gentini_noise-assisted_2019, fuchs_cryptographic_1999}. We also used the fact that $||D||_\infty := \max_j | \lambda_j(D) | = 1$.

We can also derive an alternative inequality based on the average trace distance between the individual encoded states, namely
\begin{equation} \label{eqn:fidelity_bound_average}
    \Delta_\E C \leq \frac{2}{M} \sum_{i=1}^M \sqrt{1-F(\E(\rhotildexi), \rhotildexi)} .
\end{equation}
A proof is included in Appendix~\ref{app:further_proofs}.

Due to our choice of cost function~\eqref{eqn:indicator_cost_over_dataset}, the quantity $\Delta_\E C$ corresponds exactly to the $\delta$-robustness of the model in \defref{def:partial_noise_robustness} since a classification difference in a single point due to noise causes the error to increase by $1/M$. In particular, the quantity $\Delta_\E C$ is exactly the fraction of robust points in the dataset~\eqref{eqn:delta_robust_encoding_condition}. Specifically, we have
\begin{equation}
    |\mathcal{R}(\mathcal{E}, E, \hat{y})| = M \Delta_\E C 
\end{equation}

Thus, Eqn.~\eqref{eqn:fidelity_bound_mixed_state} provides an upper bound on how large the (partially) robust set can be, namely
\begin{equation} \label{eqn:partially_robust_set_bound}
     |\mathcal{R}(\mathcal{E}, E, \hat{y})| \leq 2 \sum_{i=1}^M \sqrt{1-F(\E(\rhotildexi), \rhotildexi)} .
\end{equation}

In Sec.~\ref{ssec:fidelity_analysis_exp}, we use these inequalities to bound the size of the robust set for several different encodings on an example implementation.

\section{Numerical Results} \label{sec:numerical_results}

In this Section, we present numerical evidence to reinforce the theoretical results proved in Sec.~\ref{sec:noise_robustness} and build on the discussions. In Sec.~\ref{ssec:classes-decision-boundaries-numerical}, we show classes of learnable decision boundaries for example encodings, building on the previous discussion in Sec.~\ref{ssec:classes_learnable_decision_boundaries}. We then plot the robust sets for partially robust encodings in Sec.~\ref{subsec:robust-sets-for-partially-robust-encodings} to visualize the differences that arise from different encodings. We also generalize some encodings defined in Sec.~\ref{sec:data-encodings} to include hyperparameters and study the effects. This leads us to attempt to train over these hyperparameters, and we present an ``encoding learning algorithm'' in Sec.~\ref{ssec:encoding_learn_alg} to perform this task. Finally, in Sec.~\ref{ssec:fidelity_analysis_exp} we compute upper bounds on the size of robust sets based on \secref{ssec:fidelity_bounds}. We note that we include code to reproduce all results in this Section at Ref.~\cite{coyle_noiserobustclassifier_2020}.
For all numerical results in the following sections related to single qubit classifier, we use three simple datasets; the first is the ``moons'' dataset from {\fontfamily{cmtt}\selectfont scikit-learn}, \cite{pedregosa_scikit-learn_2011}, and two we denote ``vertical'' and ``diagonal''. Representative examples can be found in \appref{app:numerical_results}.

\subsection{Decision Boundaries and Implementations} \label{ssec:classes-decision-boundaries-numerical}

In Sec.~\ref{sec:data-encodings}, we defined an encoding~\eqref{eqn:encoding-formal} and gave several examples. In Sec.~\ref{ssec:classes_learnable_decision_boundaries}, we showed that a classifier with the wavefunction encoding~\eqref{eqn:wavefunction-encoding-density-matrix-onequbit} can learn decision boundaries that are straight lines, while the same classifier with the dense angle encoding~\eqref{eqn:dae-single-qubit} can learn sinusoidal decision boundaries. We show this in Fig.~\ref{fig:random_decision_boundaries_encodings}, and we build on this discussion in the remainder of this section.

\begin{figure}
    \includegraphics[width=\columnwidth]{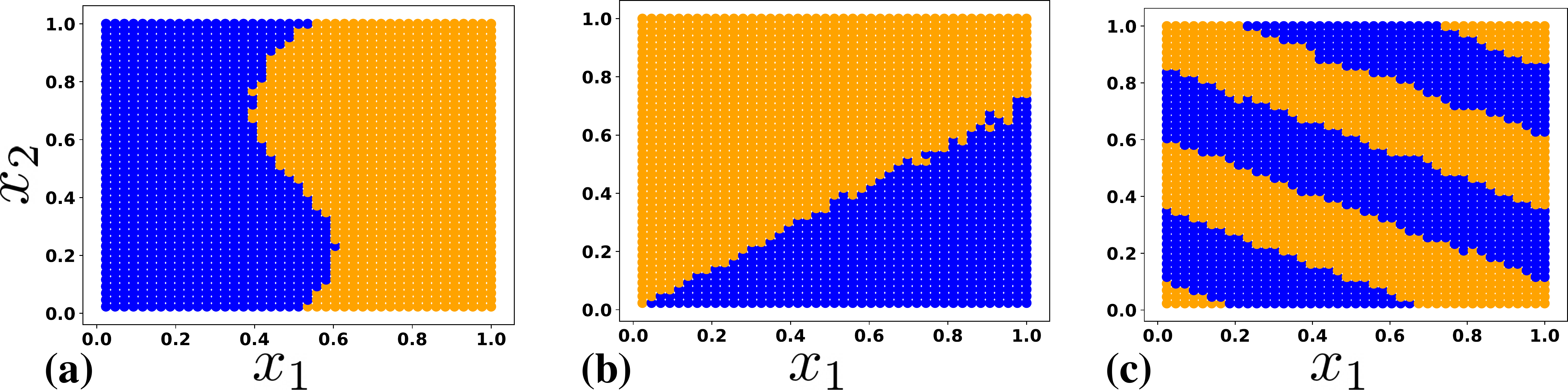}
    \caption{(Color online.) Examples of learnable decision boundaries for a single qubit classifier with the (a) dense angle encoding, (b) wavefunction encoding, and (c) superdense angle encoding where $\theta = \pi$ and $\phi = 2 \pi$.
    Colors denote class labels. The QNN used here consisted of an arbitrary single qubit rotation (see Fig.~\ref{fig:specific_classifier_circuits}) with random parameters.}
    \label{fig:random_decision_boundaries_encodings}
\end{figure}

Figure~\ref{fig:random_decision_boundaries_encodings}(c) shows a ``striped'' decision boundary which was learned by a ``superdense'' angle encoding, defined below. The superdense encoding introduces a linear combination of features into the qubit (angle) encoding~\eqref{eqn:qubit_encoding_grant}.
\begin{definition}[Superdense Angle Encoding (SDAE)] \label{def:sdae}
    Let $\vec{x} = [x_1, ..., x_N]^T \in \mathbb{R}^N$ be a feature vector and $\vec{\theta}, \vec{\phi} \in \mathbb{R}^N$ be parameters. Then, the superdense angle encoding maps $\vec{x} \mapsto E(\vec{x})$ given by
    \begin{equation} \label{eqn:sdae-general}
        |\vec{x}\> = \bigotimes_{i=1}^{\ceil*{N / 2}}\cos (\theta_i x_{2i-1} + \phi_i x_{2i}) \ket{0} + \cos (\theta_i x_{2i-1} + \phi_i x_{2i}) \ket{1} .
    \end{equation}
\end{definition}
For a single qubit, the SDAE is
\begin{equation} \label{eqn:superdae_encoding_single_qubit}
    |\vec{x}\> := \cos \left(\theta x_1+ \phi x_2\right) |0\> + \sin \left(\theta x_1+ \phi x_2\right)|1\> .
\end{equation}
We observe that $\phi = 0$ recovers the qubit (angle) encoding~\eqref{eqn:qubit_encoding_grant} considered by \cite{stoudenmire_supervised_2016, schuld_supervised_2018, cao_cost_2019} and~\eqref{eqn:sdae-general} encodes two features per qubit.

We note that Def.~\ref{def:sdae} includes hyperparameters $\mathbf{\theta}$ and $\mathbf{\phi}$. The reason for this will become clear in Sec.~\ref{ssec:encoding_learn_alg} when we consider optimizing over encoding hyperparameters to increase robustness. As previously mentioned, a similar idea was investigated by Lloyd \textit{et al.} in Ref.~\cite{lloyd_quantum_2020} for the purpose of (in our notation) learnability. 

%
%
%
%
%
%
%
%
%
%
%
%
%
%
%
%
%
%
%
%
%
%
%
%
%
%




As a final example to explore the importance of encodings, we consider an example implementation on a standard dataset using different encodings. The dataset we consider is the Iris flower dataset~\cite{fisher_use_1936} in which each flower is described by four features so that $\vec{x} \in \mathbb{R}^4$. The original dataset includes three classes (species of flower) but we only consider two for binary classification. A quantum classifier using the qubit angle encoding~\eqref{eqn:qubit_encoding_grant} and a tree tensor network (TTN) ansatz was considered in Ref.~\cite{grant_hierarchical_2018}. Using this encoding and QNN, the authors were able to successfully classify all points in the dataset. 

Since the angle encoding maps one feature into one qubit, a total of four qubits was used for the example in~\cite{grant_hierarchical_2018}. Here, we consider encodings which map two features into one qubit and thus require only two qubits. Descriptions of the encodings, QNN ansatze, and overall classification accuracy are shown in Table~\ref{table:two_qubit_iris_class_accuracy}. 
\begin{table}[h!]
    \centering
     \begin{tabular}{| c | c | c | c | c |} 
     \hline
     \textbf{Encoding} & \textbf{QNN} & $\boldsymbol{N_P}$ & $\boldsymbol{n} $ & \textbf{Accuracy} \\ [0.5ex] 
     \hline
     Angle & TTN & 7 & 4 & 100\% \\ 
     \hline
     Dense Angle & $U(4)$ & 12 & 2 & 100\% \\
     \hline
     Wavefunction  & $U(4)$ & 12 & 2 & 100\% \\
     \hline
     Superdense Angle & $U(4)$ & 12 & 2 & 77.6\% \\
      \hline
    \end{tabular}
    \caption{Classification accuracy achieved on the Iris dataset using different encodings and QNNs in the quantum classifier. The top row is from Ref.~\cite{grant_hierarchical_2018} and the remaining rows are from this work. The heading $N_p$ indicates number of parameters in the QNN and $n$ is the number of qubits in the classifier. The accuracy is the overall performance using a train-test ratio of $80\%$ on classes $0$ and $2$. (See Ref.~\cite{coyle_noiserobustclassifier_2020} for full implementation details.)}
    \label{table:two_qubit_iris_class_accuracy}
\end{table}

As can be seen, we are able to achieve 100\% accuracy using the wavefunction and dense angle encoding. For the SDAE, the accuracy drops. 
Because the SDAE performs worse than other encodings, this implementation again highlights the importance of encoding on learnability. Additionally, the fact that we can use two qubits instead of four highlights the importance of encodings from a resource perspective. Specifically, NISQ applications with fewer qubits are less error prone due to fewer two-qubit gates, less crosstalk between qubits, and reduced readout errors. The reduction in the number of qubits here due to data encoding parallels, e.g., the reduction in the number of qubits in quantum chemistry applications due to qubit tapering~\cite{bravyi_tapering_2017}. For QML, such a reduction is not always beneficial as the encoding may require a significantly higher depth. For this implementation, however, the dense angle encoding has the same depth as the angle encoding, so the reduction in number of qubits is meaningful.


\subsection{Robust Sets for Partially Robust Encodings} \label{subsec:robust-sets-for-partially-robust-encodings}

In Sec.~\ref{subsec:robustness-results}, we proved conditions under which an encoding is robust to a given error channel. Typically in practice, encodings may not completely satisfy such robustness criteria, but will exhibit partial robustness --- i.e., some number of training points will be robust, but not all. In this section, we characterize such robust sets for different partially robust encodings. We emphasize two points that (i) the number of robust points is different for different encodings, and (ii) the ``location'' of robust points is different for different encodings.

To illustrate the first point, we consider amplitude damping noise --- which has robustness condition~\eqref{eqn:amp_damp_robustness_condition_expanded_unitary} --- for two different encodings: the dense angle encoding and the wavefunction encoding. For each, we use a dataset which consists of 500 points in the unit square separated by a vertical decision boundary at $x_1 = 0.5$.

\begin{figure}
\subfloat[\label{subfig:de-encoding-linear-boundary-no-noise}]{
         \begin{tikzpicture}
  \node (img)  {\includegraphics[width = 0.15\textwidth, height = 0.155\textwidth]{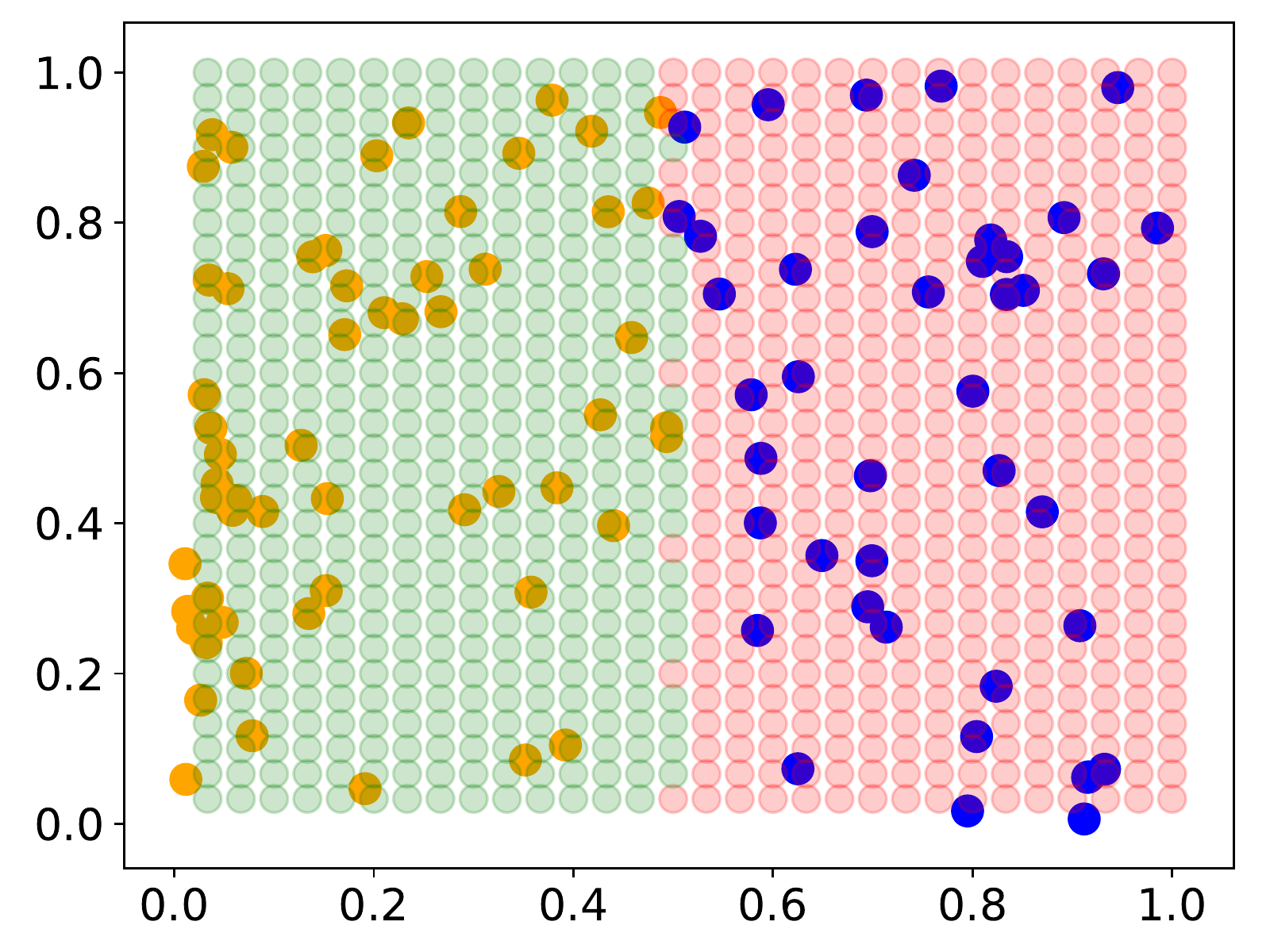}};
   \node[above=of img, node distance=0cm,  anchor=center, yshift=-1cm,font=\color{black}] {\scriptsize{Test Acc.: $98.99\%$}};
  \node[below=of img, node distance=0cm, yshift=1.2cm,font=\color{black}] {$x_1$};
  \node[left=of img, node distance=0cm, rotate=90, anchor=center,yshift=-0.95cm,font=\color{black}] {$x_2$};
\end{tikzpicture}
}\kern-1.5em
\subfloat[ \label{subfig:de-encoding-linear-boundary-amp-damp-noise-0_04ampdamp}]{
         \begin{tikzpicture}
  \node (img)  {\includegraphics[width=0.16\textwidth, height = 0.16\textwidth]{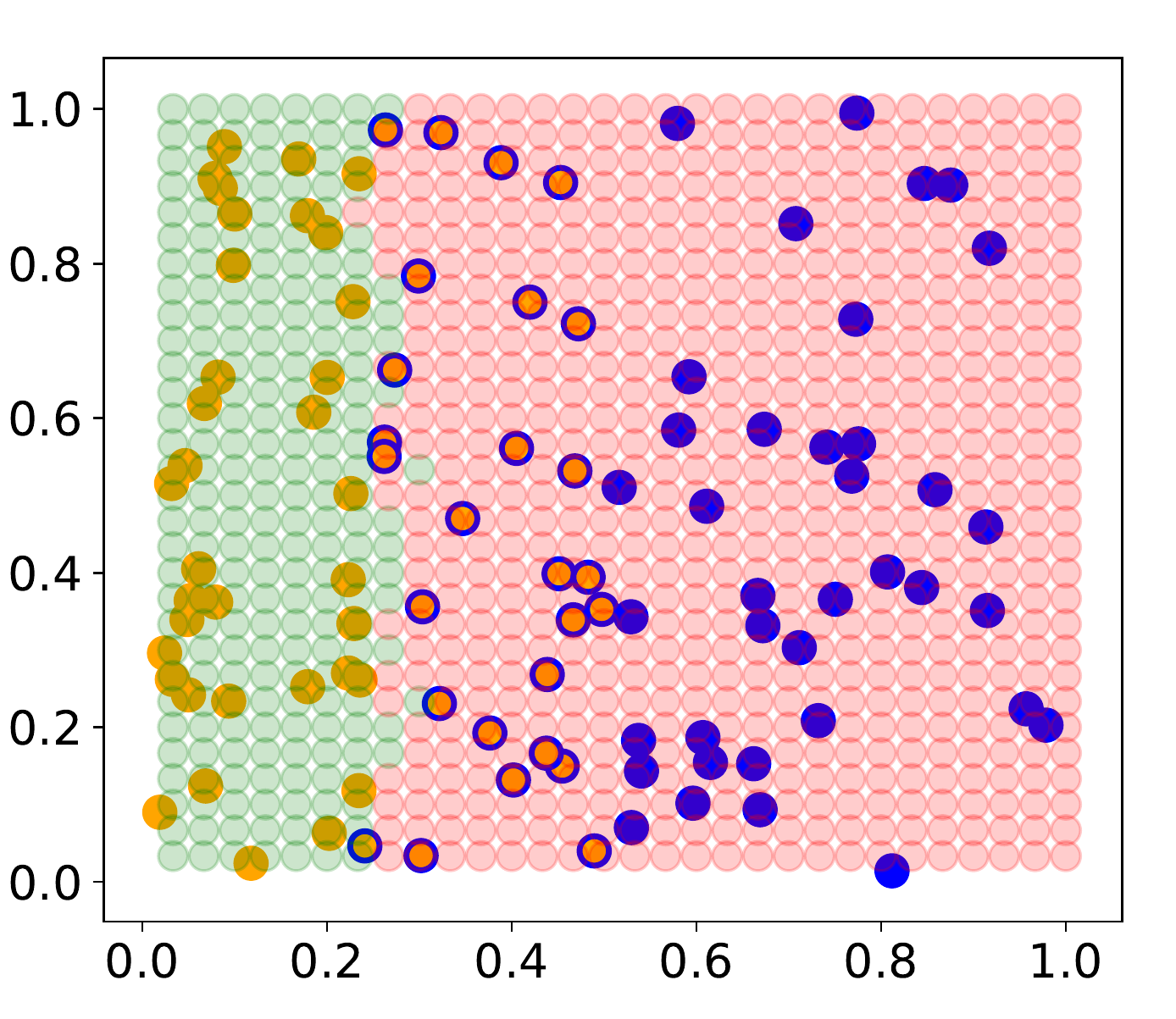}};
  \node[below=of img, node distance=0cm, yshift=1.2cm,font=\color{black}] {$x_1$};
    \node[above=of img, node distance=0cm,  anchor=center,yshift=-1.1cm,font=\color{black}] {\scriptsize{Test Acc.: $77.78\%$}};
\end{tikzpicture}
}\kern-2.2em
\subfloat[\label{subfig:de-encoding-linear-boundary-analytic-misclassification-0_04ampdamp}]{
         \begin{tikzpicture}
  \node (img)  {\includegraphics[width=0.15\textwidth, height = 0.16\textwidth]{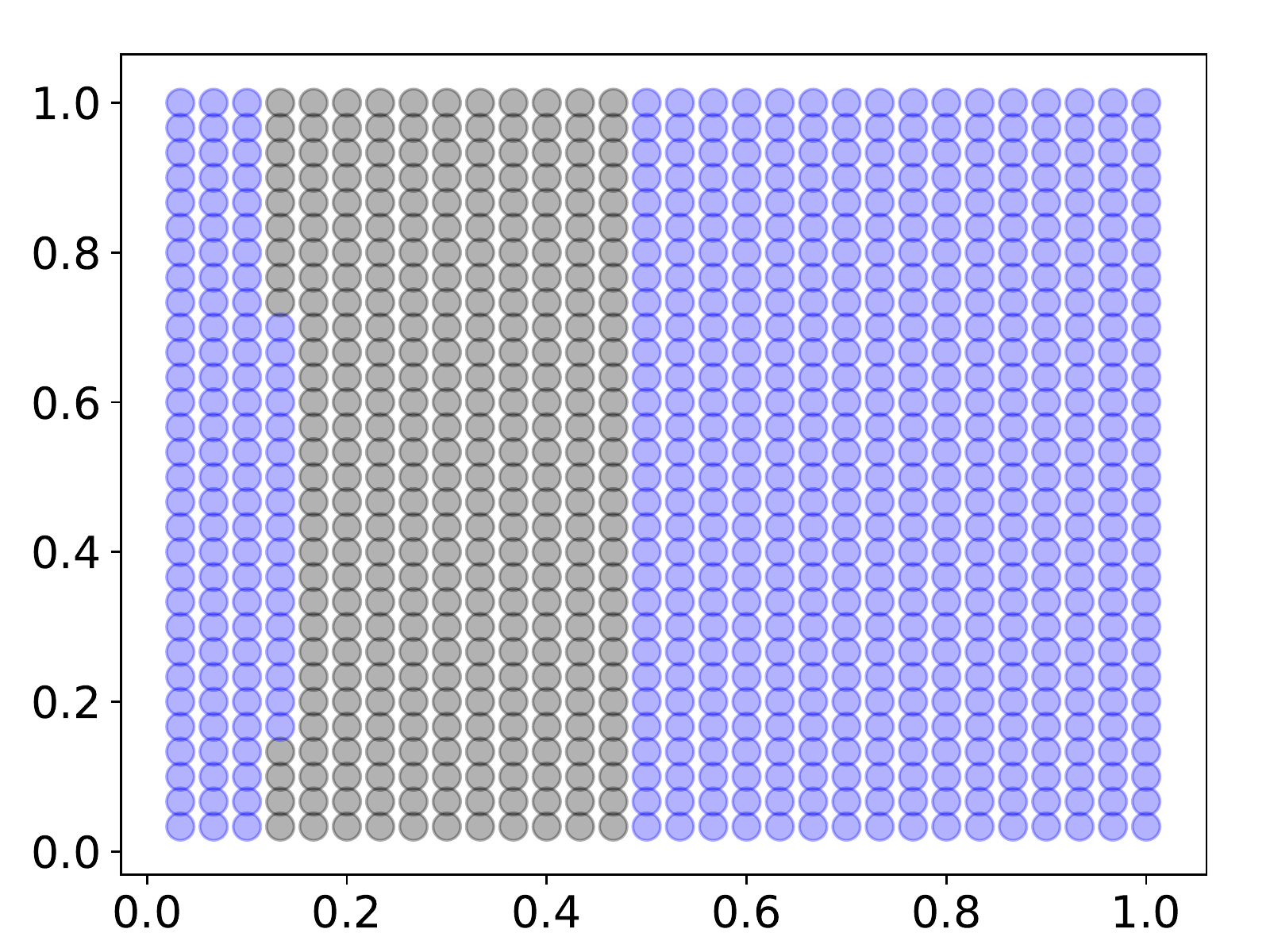}};
  \node[below=of img, node distance=0cm, yshift=1.2cm,font=\color{black}] {$x_1$};
  \node[left=of img, node distance=0cm, rotate=90, anchor=center,yshift=-0.95cm,font=\color{black}] {};
\end{tikzpicture}

}
\caption{(Color online.) Partial robustness for the dense angle encoding. The dataset consists of 500 points in the unit square separated by a vertical decision boundary, and we use a train-test split of $80\%$. Panel (a) shows the classifier test accuracy after optimizing the unitary without noise. Panel (b) shows the reduced accuracy after amplitude damping noise of strength $p = 0.4$ is added. The robust set is at the far left and far right of the unit square, explicitly shown in Panel (c). Here, [\crule[blue]{0.2cm}{0.2cm}] indicates the robust set and [\crule[black]{0.2cm}{0.2cm}] indicates its complement.}
\label{fig:dae_encoding_linear_boundary}
\end{figure}
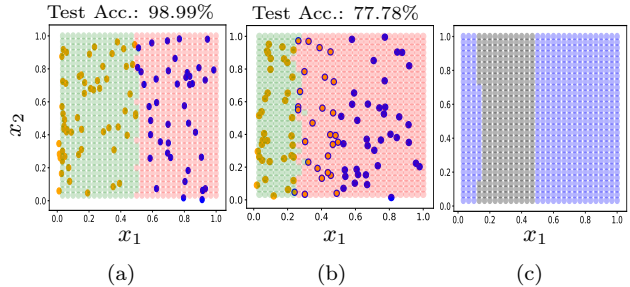

The results for the dense angle encoding are shown in Fig.~\ref{fig:dae_encoding_linear_boundary}. Without noise, the classifier is able to reach an accuracy of $\sim 99\%$ on the training set. When the amplitude damping channel with strength $p = 0.2$ is added, the test accuracy reduces to $\sim 78\%$. This encoding is thus partially robust, and the set of robust points is shown explicitly in Fig.~\ref{fig:dae_encoding_linear_boundary}(c).


The results for the wavefunction encoding are shown in Fig.~\ref{fig:wf-encoding-linear-boundary}. Here, the classifier is only able to reach $\sim 82\%$ test accuracy without noise. When the same amplitude damping channel with strength $p = 0.4$ is added, the test accuracy drops to $\sim 43\%$. We consider also the effect of amplitude damping noise with strength $p = 0.2$ in Fig.~\ref{fig:wf-encoding-linear-boundary}, for which the classifier achieves test accuracy $\sim 61\%$. The robust set for both channels is also shown in Fig.~\ref{fig:wf-encoding-linear-boundary}.


%
\begin{figure}
\subfloat[\label{subfig:wf-encoding-linear-boundary-no-noise}]{
         \begin{tikzpicture}
            \node (img)  {\includegraphics[width = 0.14\textwidth, height = 0.14\textwidth]{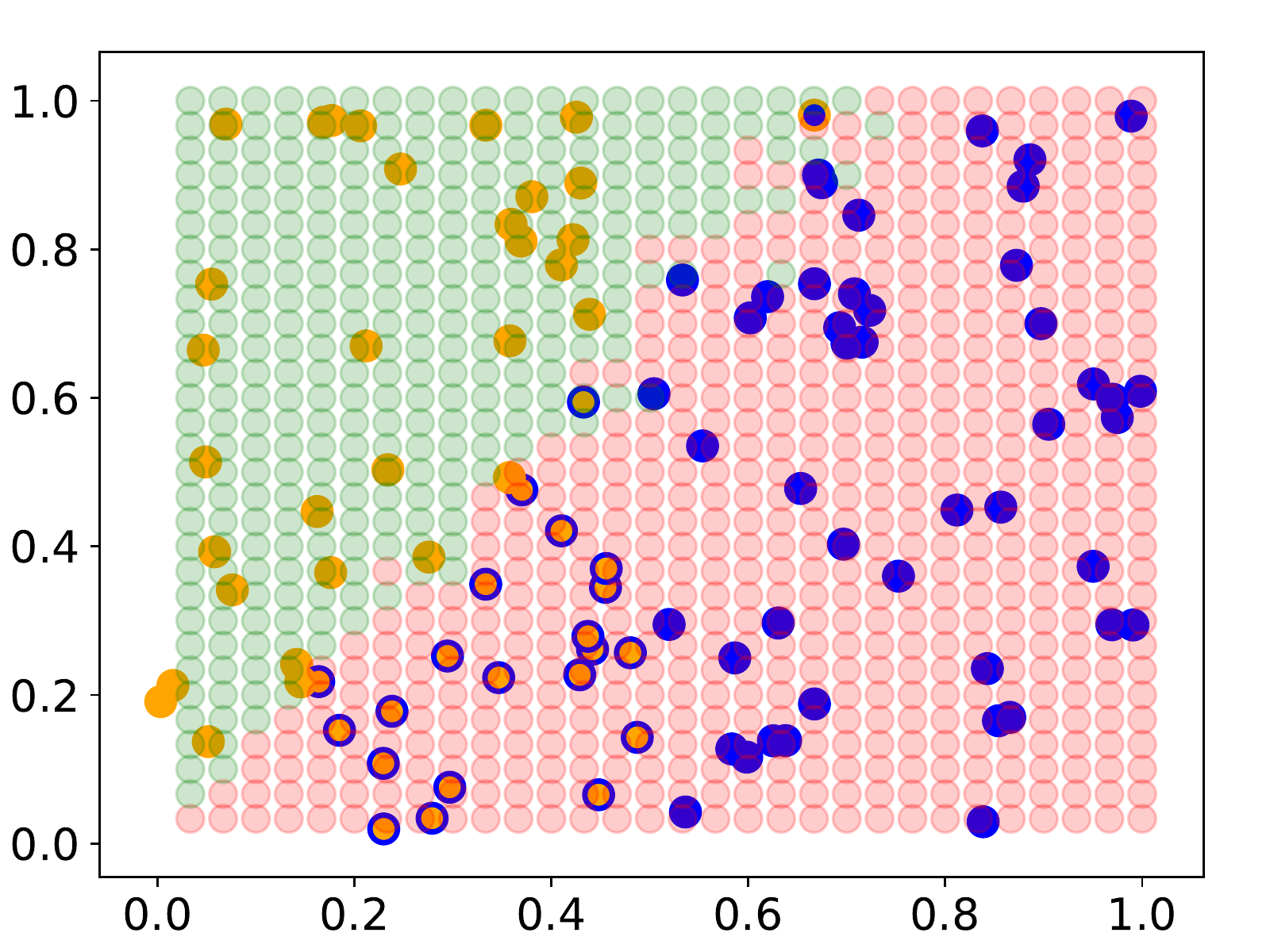}};
            \node[below=of img, node distance=0cm, yshift=1.2cm,font=\color{black}] {$x_1$};
            \node[left=of img, node distance=0cm, rotate=90, anchor=center,yshift=-0.95cm,font=\color{black}] {$x_2$};
            \node[above=of img, node distance=0cm,  anchor=center, yshift=-1.1cm,font=\color{black}] {\scriptsize{Test Acc.: $81.81\%$}};
        \end{tikzpicture}
    }\kern-1.5em
\subfloat[\label{subfig:wf-encoding-linear-boundary-amp-damp-noise-0_04ampdamp}]{
        \begin{tikzpicture}
            \node (img)  {\includegraphics[width = 0.14\textwidth, height = 0.14\textwidth]{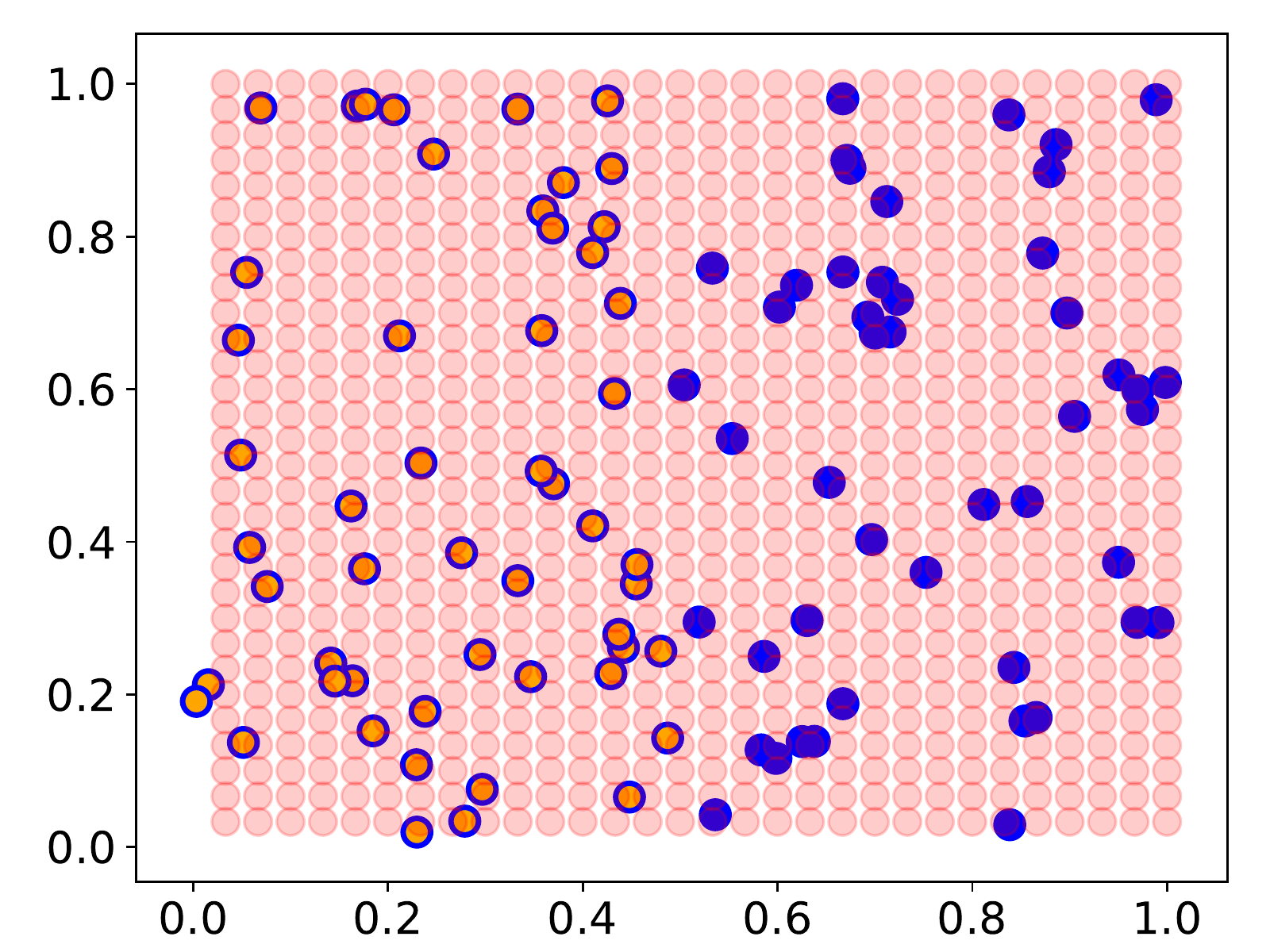}};
            \node[below=of img, node distance=0cm, yshift=1.2cm,font=\color{black}] {$x_1$};
            \node[above=of img, node distance=0cm,  anchor=center,yshift=-1.0cm,font=\color{black}] {\scriptsize{Test Acc.: $43.43\%$}};
        \end{tikzpicture}
}\kern-1.5em
\subfloat[\label{subfig:wf-encoding-linear-boundary-analytic-misclassification-0_04ampdamp}]{
         \begin{tikzpicture}
            \node (img)  {\includegraphics[width = 0.14\textwidth, height = 0.14\textwidth]{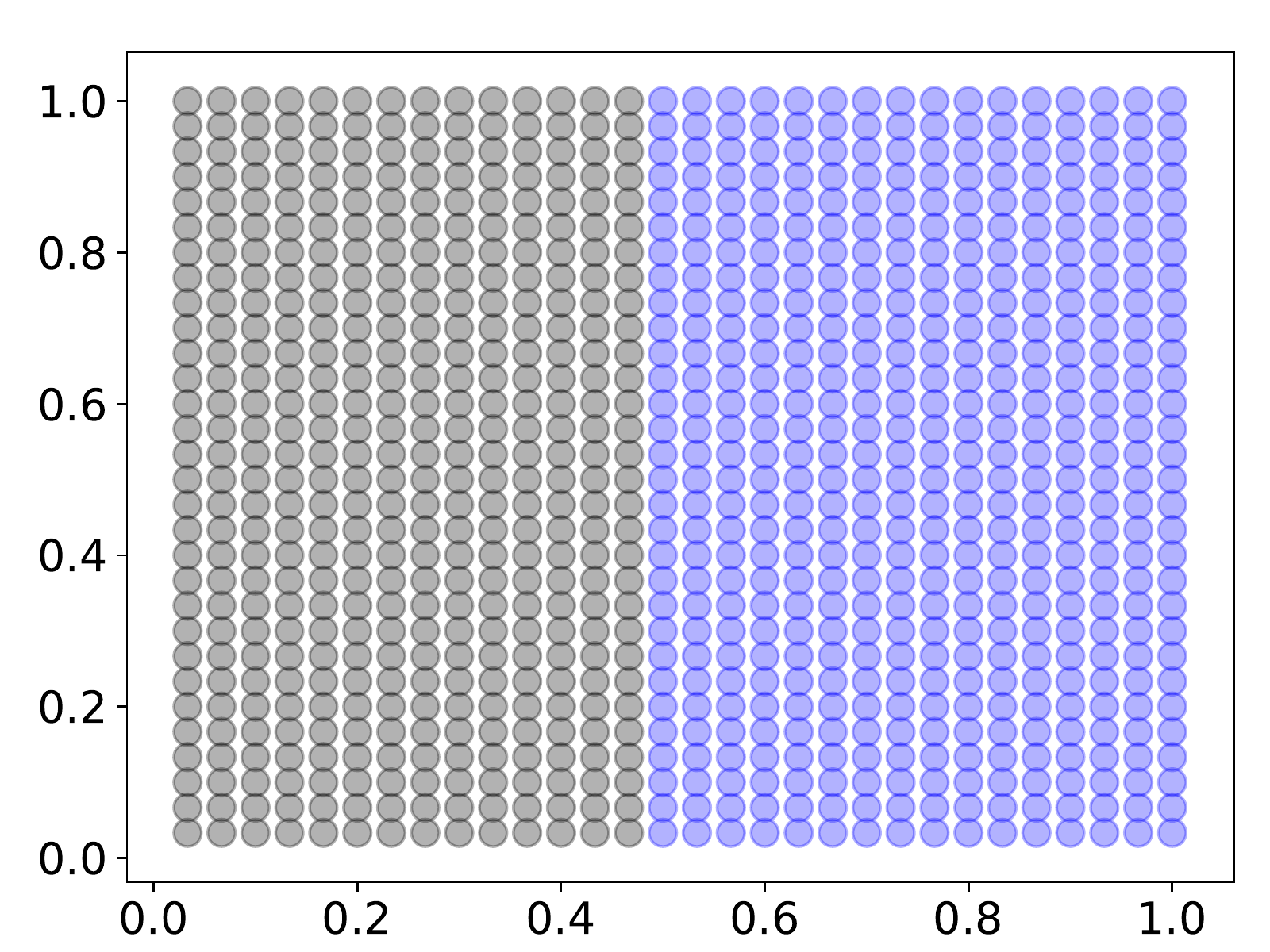}};
            \node[below=of img, node distance=0cm, yshift=1.2cm,font=\color{black}] {$x_1$};
        \end{tikzpicture}
}\kern-2em\\ \vspace{-1em}
\subfloat[\label{subfig:wf-encoding-linear-boundary-amp-damp-noise-0_02ampdamp}]{
    \begin{tikzpicture}
        \node (img)  {\includegraphics[width = 0.18\textwidth, height = 0.145\textwidth]{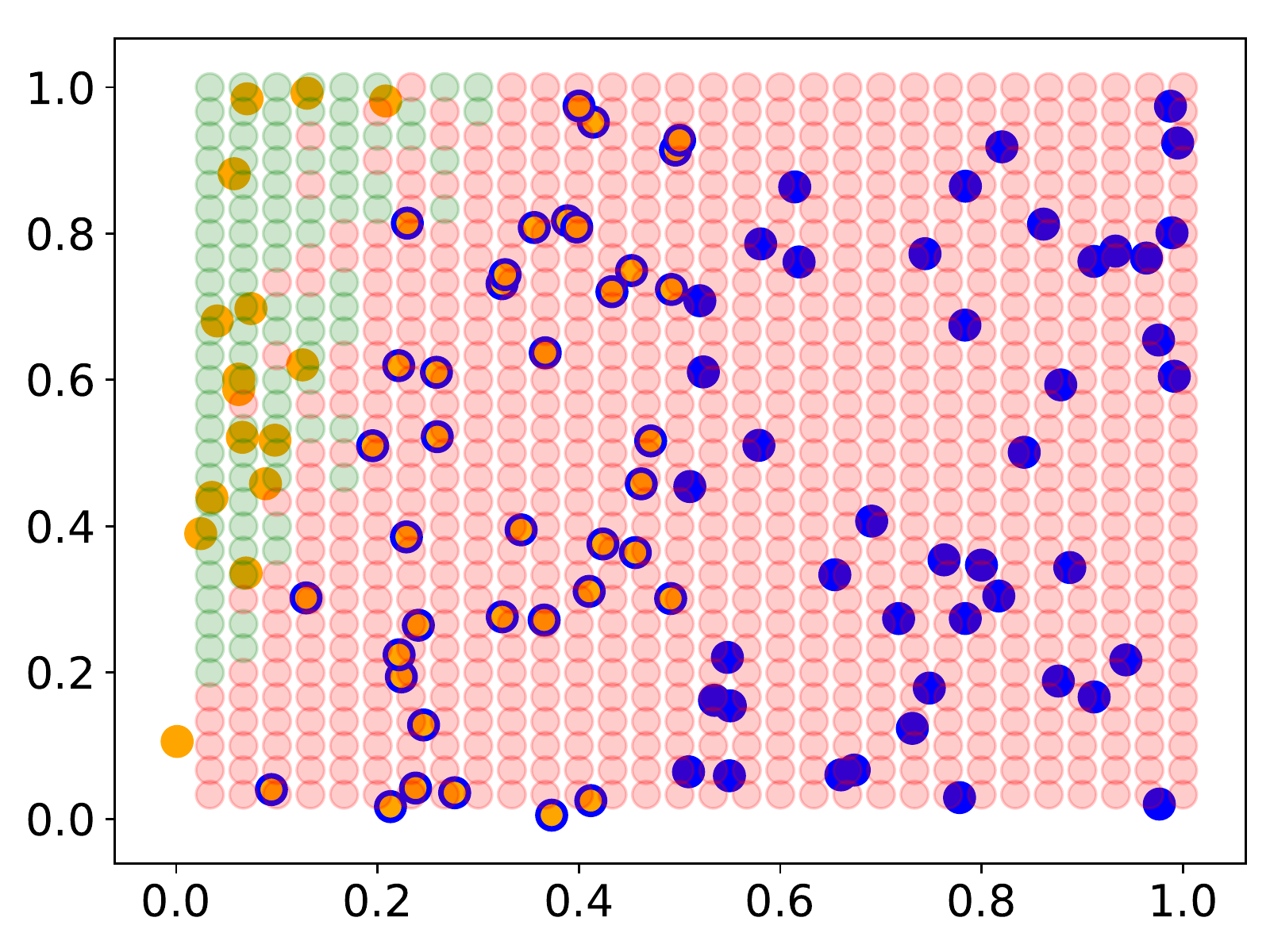}};
        \node[below=of img, node distance=0cm, yshift=1.2cm,font=\color{black}] {$x_1$};
        \node[left=of img, node distance=0cm, rotate=90, anchor=center,yshift=-0.95cm,font=\color{black}] {$x_2$};
        \node[above=of img, node distance=0cm,  anchor=center,yshift=-1.0cm,font=\color{black}] {\scriptsize{Test Acc.: $60.6\%$}};
    \end{tikzpicture}
}\kern-1.5em
\subfloat[\label{subfig:wf-encoding-linear-boundary-analytic-misclassification-0_02ampdamp}]{
    \begin{tikzpicture}
        \node (img)  {\includegraphics[width = 0.18\textwidth, height = 0.15\textwidth]{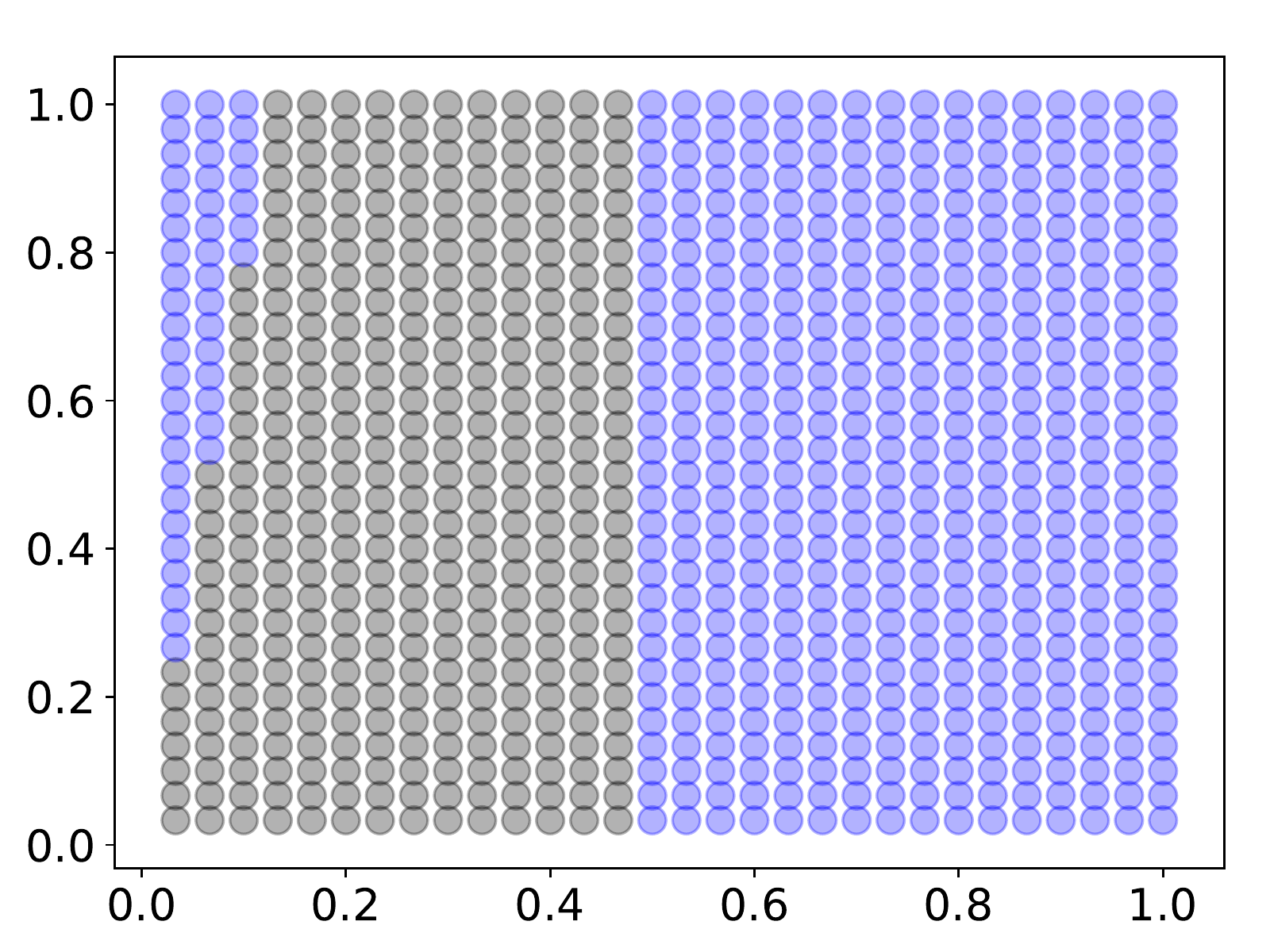}};
        \node[below=of img, node distance=0cm, yshift=1.2cm,font=\color{black}] {$x_1$};
    \end{tikzpicture}
    }
\caption{
(Color online.) Partial robustness for the wavefunction encoding. The dataset consists of 500 points in the unit square separated by a vertical decision boundary, and we use a train-test split of $80\%$. Panel (a) shows classifier test accuracy after optimizing the unitary without noise. Panel (b) shows the reduced accuracy after adding amplitude damping noise with strength $p = 0.4$. The robust set is shown explicitly in Panel (c) where [\crule[blue]{0.2cm}{0.2cm}] indicates a robust point and [\crule[black]{0.2cm}{0.2cm}] indicates a misclassified point.
Panel (d) is the same as (b) but with decreased strength $p = 0.2$ of the amplitude damping channel. Test accuracy reduces from $81\%$ to $60\%$ in this case. Panel (e) shows the robust set for (d).
} 
\label{fig:wf-encoding-linear-boundary}
\end{figure}
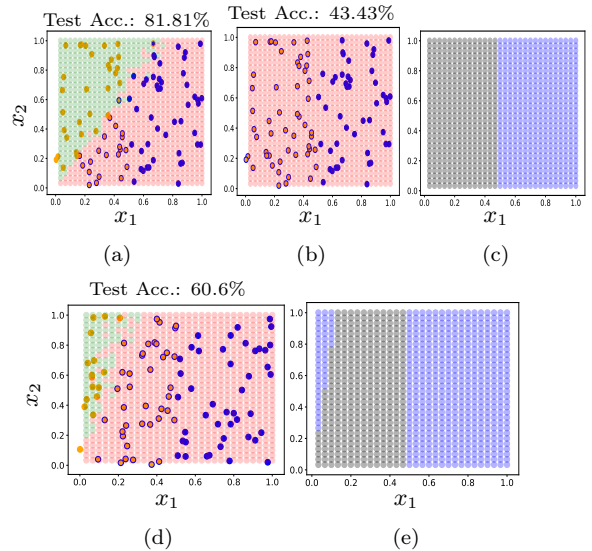


\subsection{Encoding Learning Algorithm}\label{ssec:encoding_learn_alg}


For this purpose, we introduce an ``encoding learning algorithm'' to try and search for good encodings. The goal is crudely illustrated in  Fig.~\ref{fig:single_qubit_binary_classifier_encoding_learning}. As mentioned above, Ref.~\cite{lloyd_quantum_2020} trains over hyperparameters using the re-uploading structure of Ref.~\cite{perez-salinas_data_2019} to increase learnability. Here, the encoding learning algorithm adapts to noise to increase robustness. We note the distinction that in our implementations we train the unitary in a noiseless environment and do not alter its parameters during the encoding learning algorithm.

\begin{figure}
    \includegraphics[width=\columnwidth]{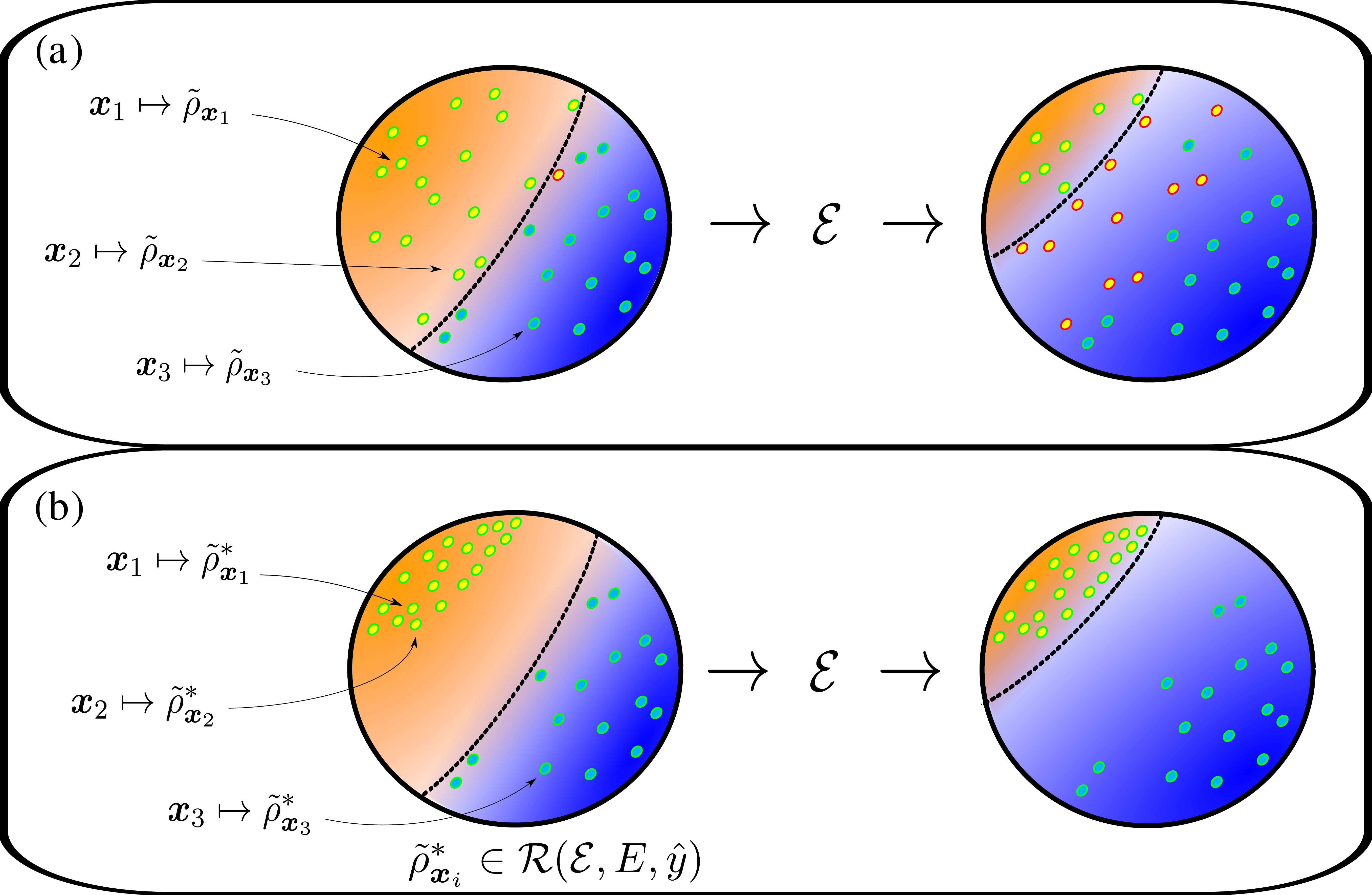}
    \caption{(Color online.) Cartoon illustration of the encoding learning algorithm with a single qubit classifier. In (a), a preset encoding with no knowledge of the noise misclassifies a large number of points. In (b), the encoding learning algorithm detects misclassifications and tries to adjust points to achieve more robustness, attempting to encode into the robust set for the channel.}
    \label{fig:single_qubit_binary_classifier_encoding_learning}
\end{figure}
%

The encoding learning algorithm is similar to the Rotoselect algorithm~\cite{ostaszewski_quantum_2019} which is used to learn circuit structure. For each function pair $(f_j, g_j)$ from a discrete set of parameterized functions $\{f_i(\vec{\theta}_i), g_i(\vec{\theta}_i)\}_{i = 1}^{K}$  we train the unitary $U(\vec{\alpha})$ to minimize the cost while keeping the encoding (hyper)parameters $\vec{\theta}_j$ fixed. Next, we add a noise channel $\E$ which causes some points to be misclassified. Now, we optimize the encoding parameters $\vec{\theta}_j$ in the noisy environment. For this optimization, the same cost function is used, and the goal is to further decrease the cost (and hence increase the set of robust points) by varying the encoding hyperparameters. Pseudocode for the algorithm is shown in Appendix~\ref{app:numerical_results}.

We test the algorithm on linearly separable and non-linearly separable datasets in Fig.~\ref{fig:encoding_learning_algorithm}. In particular, we use three different encodings on three datasets. The encodings used are the dense angle encoding, superdense angle encoding, and ``generalized wavefunction encoding (GWFE)'' given by

\begin{equation} \label{eqn:gwf_single_qubit}
    | \vec{x} \> := \frac{\sqrt{1+\theta x_2^2}}{||\vec{x}||_2} x_1 |0\> + \frac{\sqrt{1-\theta x_1^2}}{||\vec{x}||_2} x_2 |1\>.
\end{equation}
for a single qubit.

Using these encodings and the datasets shown in Appendix~\ref{app:numerical_results}, we study performance for the noiseless case, noisy case, and the effect of the encoding learning algorithm. We observe that the algorithm is not only capable of recovering the noiseless classification accuracy achieved, but is actually able to outperform it in some cases, as can be seen in Fig.~\ref{fig:encoding_learning_algorithm}.

Finally, we consider the discussion in Sec.~\ref{subsec:existence-of-robust-encodings} about the tradeoff between learnability and robustness. We make this quantitative in Fig.~\ref{fig:dae_encoding_learnability_versus_robustness} by plotting accuracy (percent learned correctly) and robustness against hyperparameters $\theta$ and $\phi$ in a generalized dense angle encoding
\begin{equation}
     |\vec{x}\> = \cos (\theta x_{1}) \ket{0} + e^{i \phi x_{2}} \sin (\theta x_{1})\ket{1} .
\end{equation}
More specifically, in Fig.~\ref{fig:dae_encoding_learnability_versus_robustness}, we illustrate how the noise affects the hyperparameters, $\theta^*$ and $\phi^*$, which give maximal classification \emph{accuracy} in both the noiseless and noisy environments, and also those which give maximal \emph{robustness} (in the sense of \defref{def:partial_noise_robustness}). Fig.~\ref{fig:dae_encoding_learnability_versus_robustness}(a), shows the percentage misclassified in the noiseless environment, where red indicates the lowest accuracy on the test set, and blue indicates the highest accuracy. We then repeat this in Fig.~\ref{fig:dae_encoding_learnability_versus_robustness}(b) and Fig.~\ref{fig:dae_encoding_learnability_versus_robustness}(c) to find the parameters which maximize accuracy in the presence of noise, and the maximize robustness. As expected, for the amplitude damping channel, the best parameters (with noise) are closer to the fixed point of the channel (i.e. $\theta^* \rightarrow 0$ implies encoding in the  $\ket{0}$ state), thereby demonstrating the tradeoff between learnability and robustness. 
%

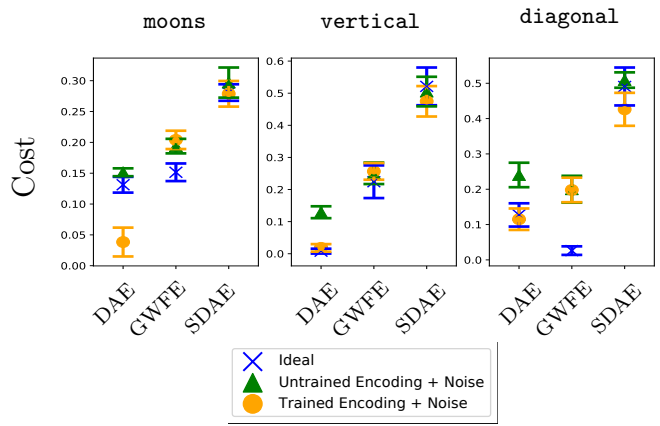
\begin{figure}[ht!]
\begin{tikzpicture}
  \node (img)  {\includegraphics[width = 0.48\textwidth, height = 0.18\textwidth]{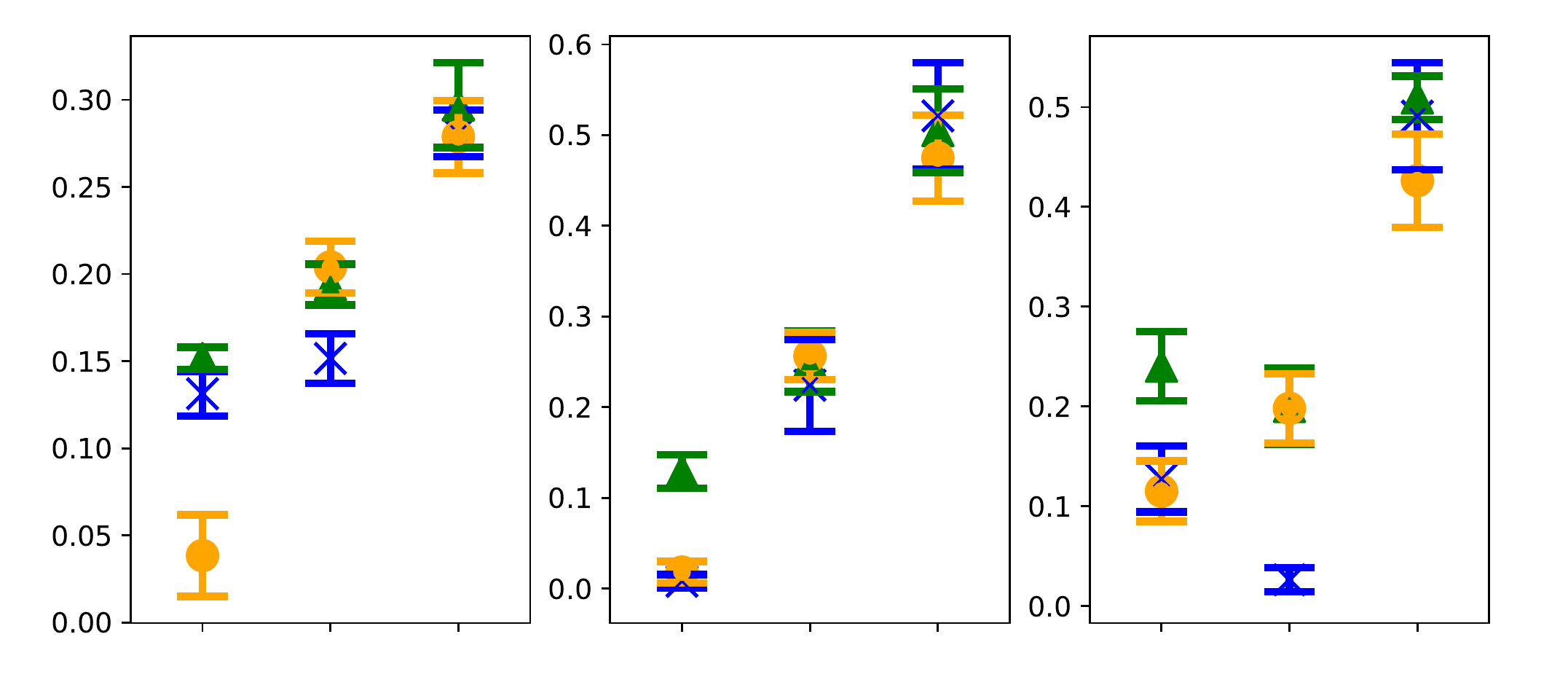}};
  \node[above=of img, node distance=0cm, xshift=-2.5cm, yshift=-1cm,font=\color{black}] {\moons};
    \node[above=of img, node distance=0cm, xshift=0.1cm, yshift=-1cm,font=\color{black}] {\vertical};  
    \node[above=of img, node distance=0cm, xshift=2.75cm, yshift=-1cm,font=\color{black}] {\diagonal};
  \node[left=of img, node distance=0cm, rotate=90, anchor=center,yshift=-0.9cm,font=\color{black}] {\large{Cost}};
      \node[below=of img, node distance=0cm, xshift=-2.4cm, rotate=45,
      yshift=1.5cm,font=\color{black}] {\footnotesize{\textnormal{DAE}}};
         \node[below=of img, node distance=0cm, xshift=-1.9cm, rotate=45, yshift=1.4cm,font=\color{black}] {\footnotesize{\textnormal{GWFE}}};
    \node[below=of img, node distance=0cm, xshift=-1.2cm, rotate=45, yshift=1.4cm,font=\color{black}] {\footnotesize{\textnormal{SDAE}}};
          \node[below=of img, node distance=0cm, xshift=0.3cm, rotate=45, yshift=1.5cm,font=\color{black}] {\footnotesize{\textnormal{DAE}}};
         \node[below=of img, node distance=0cm, xshift=0.8cm, rotate=45, yshift=1.4cm,font=\color{black}] {\footnotesize{\textnormal{GWFE}}};
    \node[below=of img, node distance=0cm, xshift=1.6cm, rotate=45, yshift=1.4cm,font=\color{black}] {\footnotesize{\textnormal{SDAE}}};
          \node[below=of img, node distance=0cm, xshift=3cm, rotate=45, yshift=1.5cm,font=\color{black}] {\footnotesize{\textnormal{DAE}}};
         \node[below=of img, node distance=0cm, xshift=3.5cm, rotate=45, yshift=1.4cm,font=\color{black}] {\footnotesize{\textnormal{GWFE}}};
    \node[below=of img, node distance=0cm, xshift= 4.2cm, rotate=45, yshift=1.4cm,font=\color{black}] {\footnotesize{\textnormal{SDAE}}};
       \node [below=of img, yshift=0.5cm]{\includegraphics[width = 0.2\textwidth, height = 0.06\textwidth]{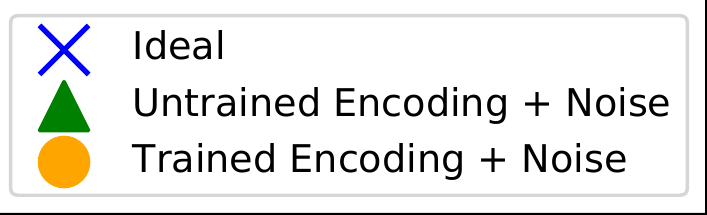}};
\end{tikzpicture}
\caption{Minimum cost achieved (vertical axis) from applying the encoding learning algorithm to three example datasets (each plot) using three different encodings (horizontal axis). The blue points [\crule[blue]{0.2cm}{0.2cm}] show the ``idea'' of training over unitary parameters only without any noise present.  The green points [\crule[ForestGreen]{0.2cm}{0.2cm}] show the same case with the addition of amplitude damping noise of strength $p = 0.3$. The orange points [\crule[orange]{0.2cm}{0.2cm}] show minimum cost after applying the encoding learning algorithm. The dense angle encoding (DAE) is seen to perform well on all datasets and is capable of adapting well to noise, even outperforming the ideal case without noise and fixed encoding. The superdense angle encoding (SDAE) does not perform well on any shown dataset since the generated decision boundary is highly nonlinear and cannot correctly classify more than half the dataset. The generalized wavefunction encoding (GWFE) performs well on the diagonal boundary, since it generates a suitable decision boundary.}
\label{fig:encoding_learning_algorithm}
\end{figure}

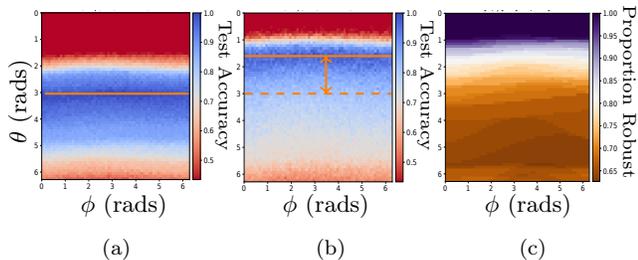
\begin{figure}[!ht]
\subfloat[ \label{subfig:learnability_v_robustness_a}]{
\begin{tikzpicture}
  \node (img)  {\includegraphics[width = 0.135\textwidth, height = 0.135\textwidth]{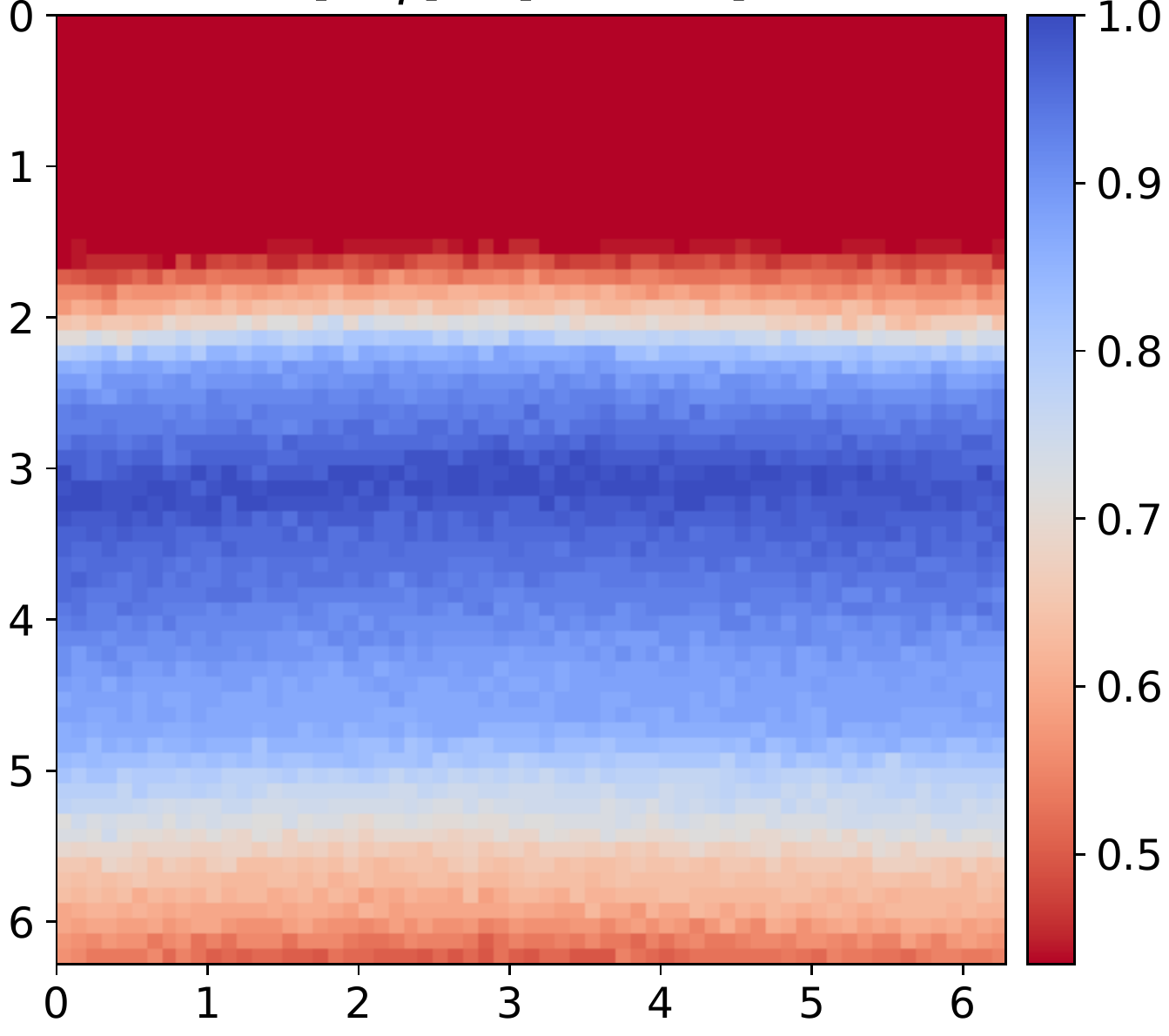}};
  \node[below=of img, node distance=0cm, yshift=1.2cm,font=\color{black}] {$\phi$ (rads)};
  \node[left=of img, node distance=0cm, rotate=90, anchor=center,yshift=-0.95cm,font=\color{black}] {$\theta$ (rads)};
\draw[thick, orange, thick] (-1.05, 0.1) -- (0.86, 0.1);
  \node[right=of img, node distance=0cm,xshift=-1cm, yshift=1.2cm,font=\color{black}, rotate=270] {\scriptsize{Test Accuracy}};
\end{tikzpicture}
}\kern-1.2em
\subfloat[ \label{subfig:learnability_v_robustness_b}]{
\begin{tikzpicture}
  \node (img)  {\includegraphics[width = 0.135\textwidth, height = 0.135\textwidth]{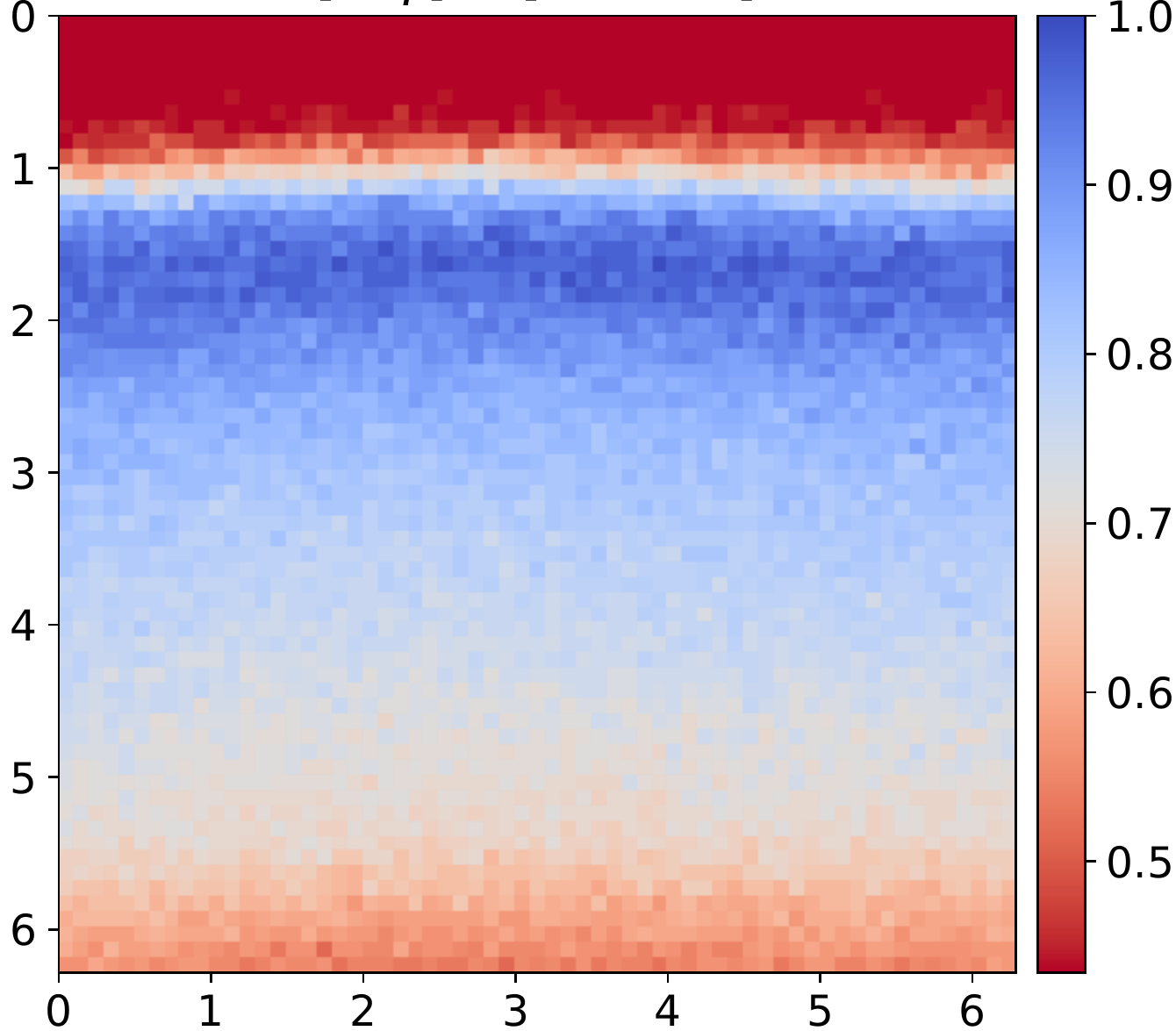}};
  \node[below=of img, node distance=0cm, yshift=1.2cm,font=\color{black}] {$\phi$ (rads)};
  \draw[thick, orange, thick] (-1.08, 0.6) -- (0.87, 0.6);
  \draw[thick, orange, dashed] (-1.08, 0.1) -- (0.87, 0.1);
    \draw [thick, orange, ->] (0, 0.6) edge (0, 0.1);
  \node[right=of img, node distance=0cm,xshift=-1cm, yshift=1.2cm,font=\color{black}, rotate=270] {\scriptsize{Test Accuracy}};
\end{tikzpicture}
}\kern-1.2em
\subfloat[ \label{subfig:learnability_v_robustness_c}]{
\begin{tikzpicture}
  \node (img)  {\includegraphics[width = 0.135\textwidth, height = 0.135\textwidth]{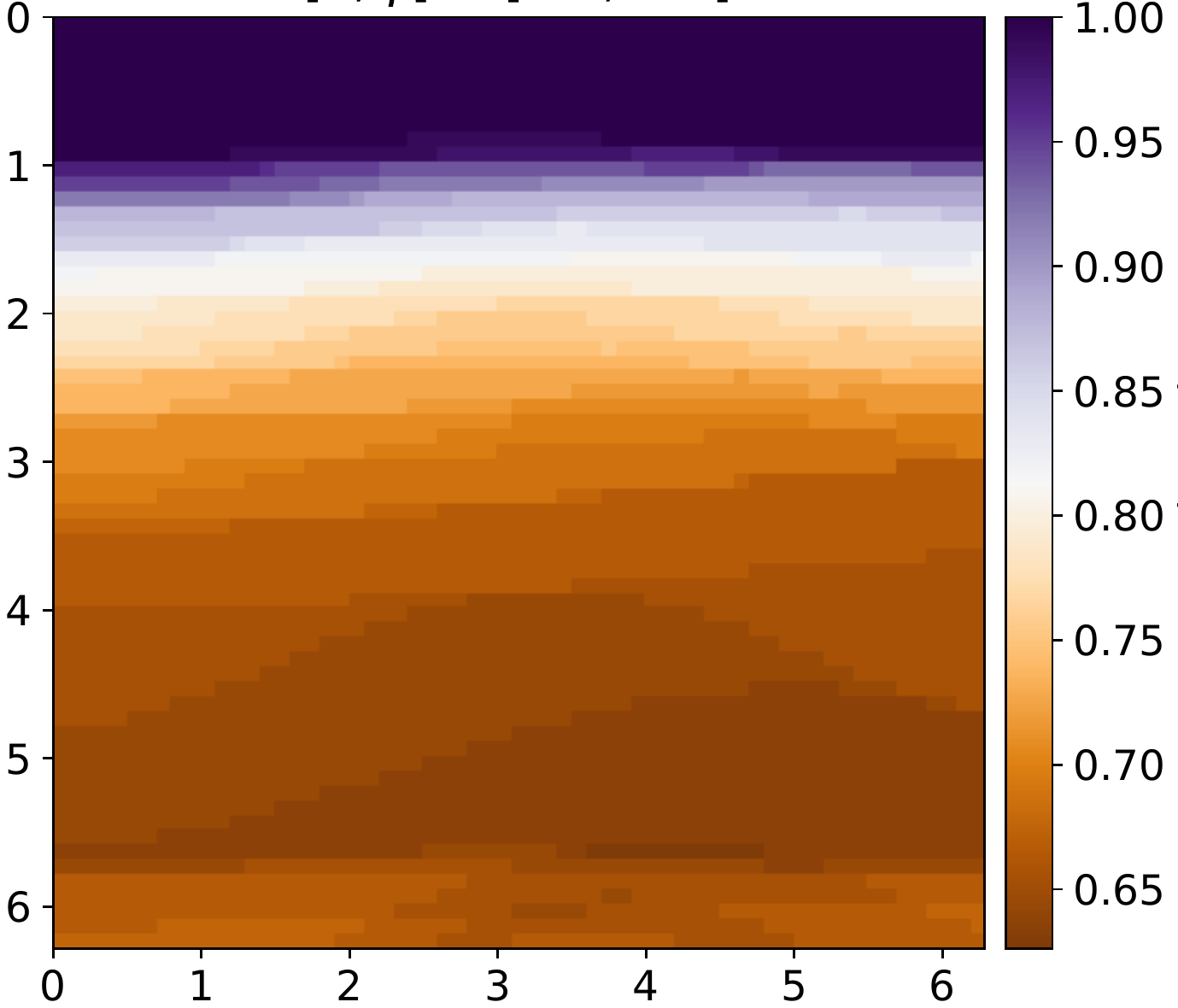}};
  \node[below=of img, node distance=0cm, yshift=1.2cm,font=\color{black}] {$\phi$ (rads)};
  \node[right=of img, node distance=0cm,xshift=-1cm, yshift=1.4cm,font=\color{black}, rotate=270] {\scriptsize{Proportion Robust}};
\end{tikzpicture}\vspace{-1em}
}
\caption{Learnability versus robustness of on the ``vertical'' dataset using the parameterized dense angle encoding. The horizontal and vertical axes show the encoding hyperparameters $\phi$ and $\theta$, respectively. Panels (a) and (b) show the classifier accuracy while Panel (c) shows the proportion of robust points. Panel (a) shows accuracy without noise as a function of encoding parameters.
Panel (b) shows accuracy with the addition of an amplitude damping channel of strength $p = 0.3$. Panel (c) shows $\delta$ - robustness for different parameter values. As expected, the robust set is largest when all points are encoded into the zero state, i.e. $\theta=0$. This leads to all points labeled $1$ being misclassified, with a resulting accuracy of approximately $50\%$. The orange [\crule[orange]{0.2cm}{0.2cm}] line indicates optimal $\theta$ parameters in each panel. From (a) to (b), the $\theta$ parameters corresponding to highest accuracy are shifted towards the robust points (i.e., towards $\theta = 0$) in (c). See \appref{app:numerical_results} for the optimal parameters found in each case.} 
\label{fig:dae_encoding_learnability_versus_robustness}
\end{figure}

\subsection{Fidelity Bounds on Partial Robustness} \label{ssec:fidelity_analysis_exp}

As a final numerical implementation, we compute the upper bounds on partial robustness proved in Sec.~\ref{ssec:fidelity_bounds} for several different encodings and error channels. The implementation we consider is the previously-discussed Iris datasest classification problem using two qubits. The results are shown in Fig.~\ref{fig:fidelity_analysis_iris}. In this Figure, each plot corresponds to a different error channel with strength varied across the horizontal axis. Each curve in the top row corresponds to the fidelity of noisy and noiseless states using different encodings. Each curve in the bottom row shows the upper bounds on partial robustness proved in Sec.~\ref{ssec:fidelity_bounds}. 


As can be seen in the bottom row of Fig.~~\ref{fig:fidelity_analysis_iris}, upper bounds on partial robustness are different for different encodings, particularly at small noise values. (Recall that a trivial upper bound on the size of partial robustness is one so that curves at large channel strengths above one are mostly uninformative.) 
For such low values of noise, they give us some information about the maximum cost function deviation we can expect.  Based on the average fidelity over the datasets, in Figs.~(\ref{subfig:fidelity_compare_iris_bit_flip} - \ref{subfig:fidelity_compare_iris_depo}) all three encodings behave qualitatively the same. However, the cost function error for the three encodings is significantly different, especially for bit flip and dephasing errors, Figures~(\ref{subfig:fidelity_bound_iris_bit_flip} - \ref{subfig:fidelity_bound_iris_dephase}). As expected, a depolarizing channel causes no misclassification, as seen in \figref{subfig:fidelity_bound_iris_depo}, despite the decrease in fidelity of the states. The apparent erratic behavior of the cost function error is largely due to the low number of samples in the Iris dataset. (Recall that the superdense angle encoding was not able to achieve perfect classification accuracy on the Iris dataset, so under amplitude damping noise, e.g., the cost function error can only decrease by about $25\%$ ($\sim 77\% \rightarrow 50 \%$).)
We can also observe that the dense angle encoding is less susceptible to bit flip and phase errors than the wavefunction encoding in Fig.~\ref{subfig:fidelity_bound_iris_bit_flip} and  Fig.~\ref{subfig:fidelity_bound_iris_dephase}.

\begin{figure*}[!ht]
\subfloat[\label{subfig:fidelity_compare_iris_bit_flip}]{
         \begin{tikzpicture}
  \node (img)  {\includegraphics[width = 0.2\textwidth, height = 0.2\textwidth]{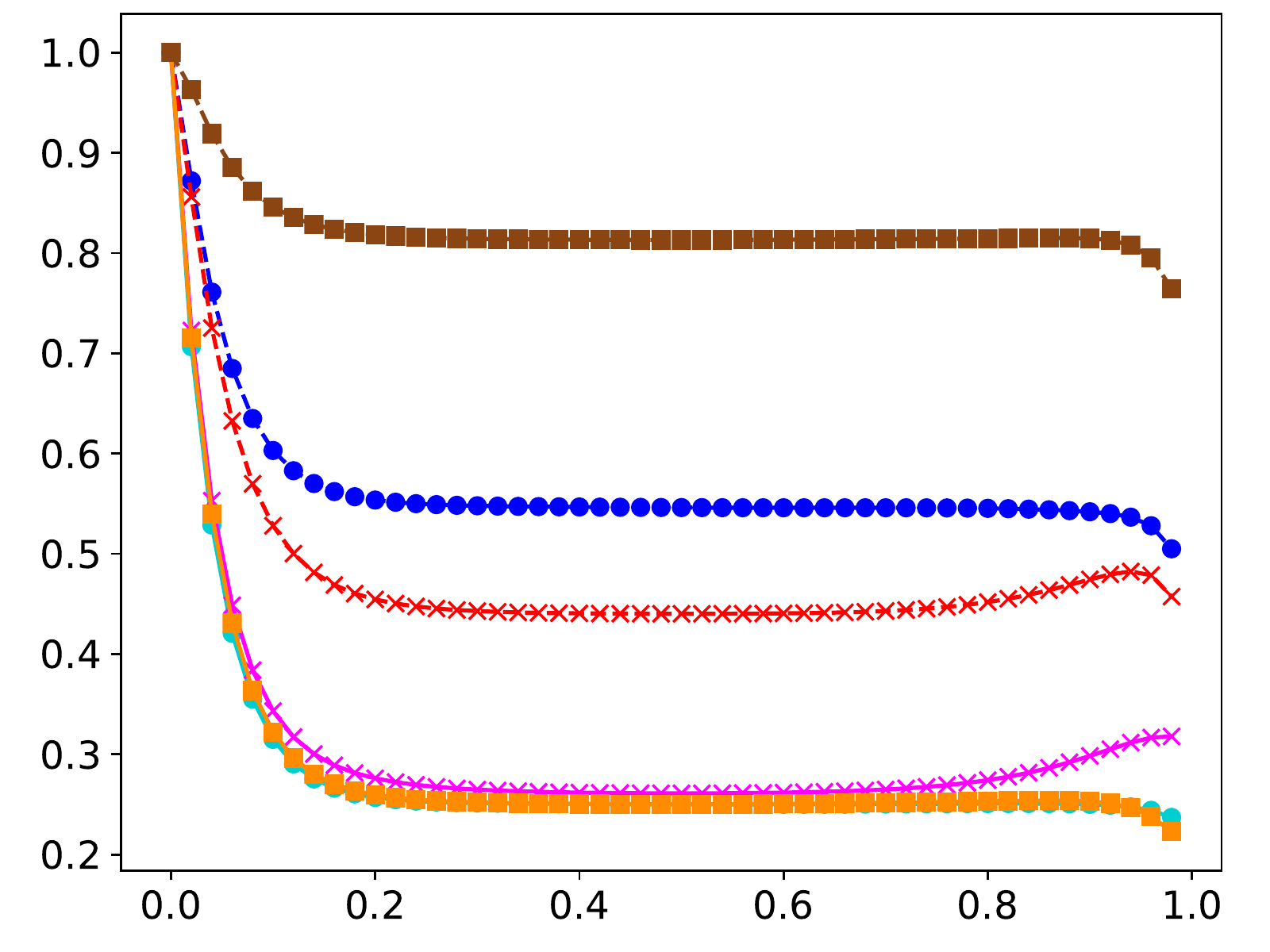}};
  \node[below=of img, node distance=0cm, yshift=1.2cm,font=\color{black}] {$p_X$};
\end{tikzpicture}
}\kern-1.5em
\subfloat[\label{subfig:fidelity_compare_iris_amp_damp}]{
         \begin{tikzpicture}
  \node (img)  {\includegraphics[width = 0.2\textwidth, height = 0.2\textwidth]{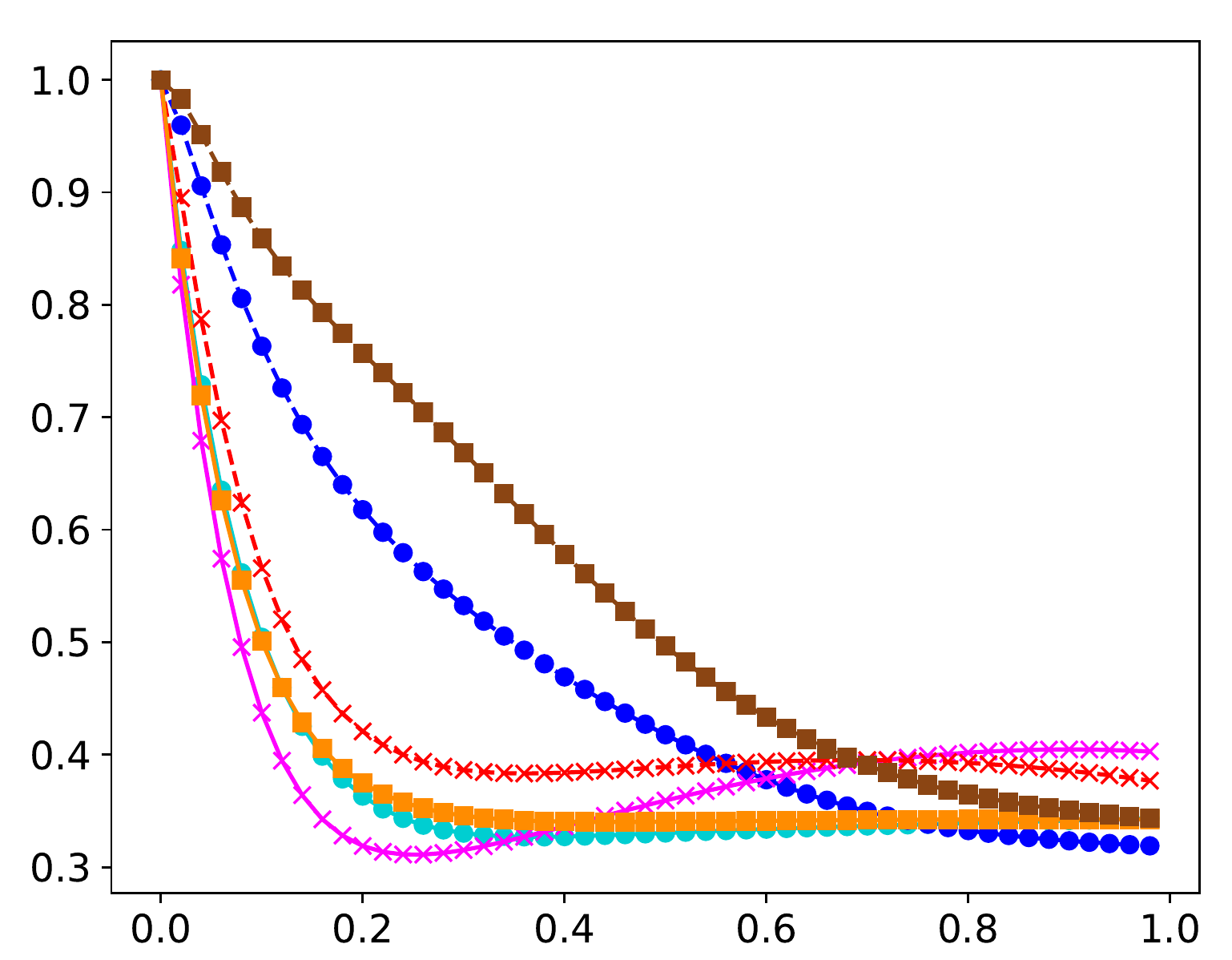}};
  \node[below=of img, node distance=0cm, yshift=1.2cm,font=\color{black}] {$p_{damp}$};
\end{tikzpicture}
}\kern-1.5em
\subfloat[\label{subfig:fidelity_compare_iris_dephase}]{
         \begin{tikzpicture}
  \node (img)  {\includegraphics[width = 0.2\textwidth, height = 0.2\textwidth]{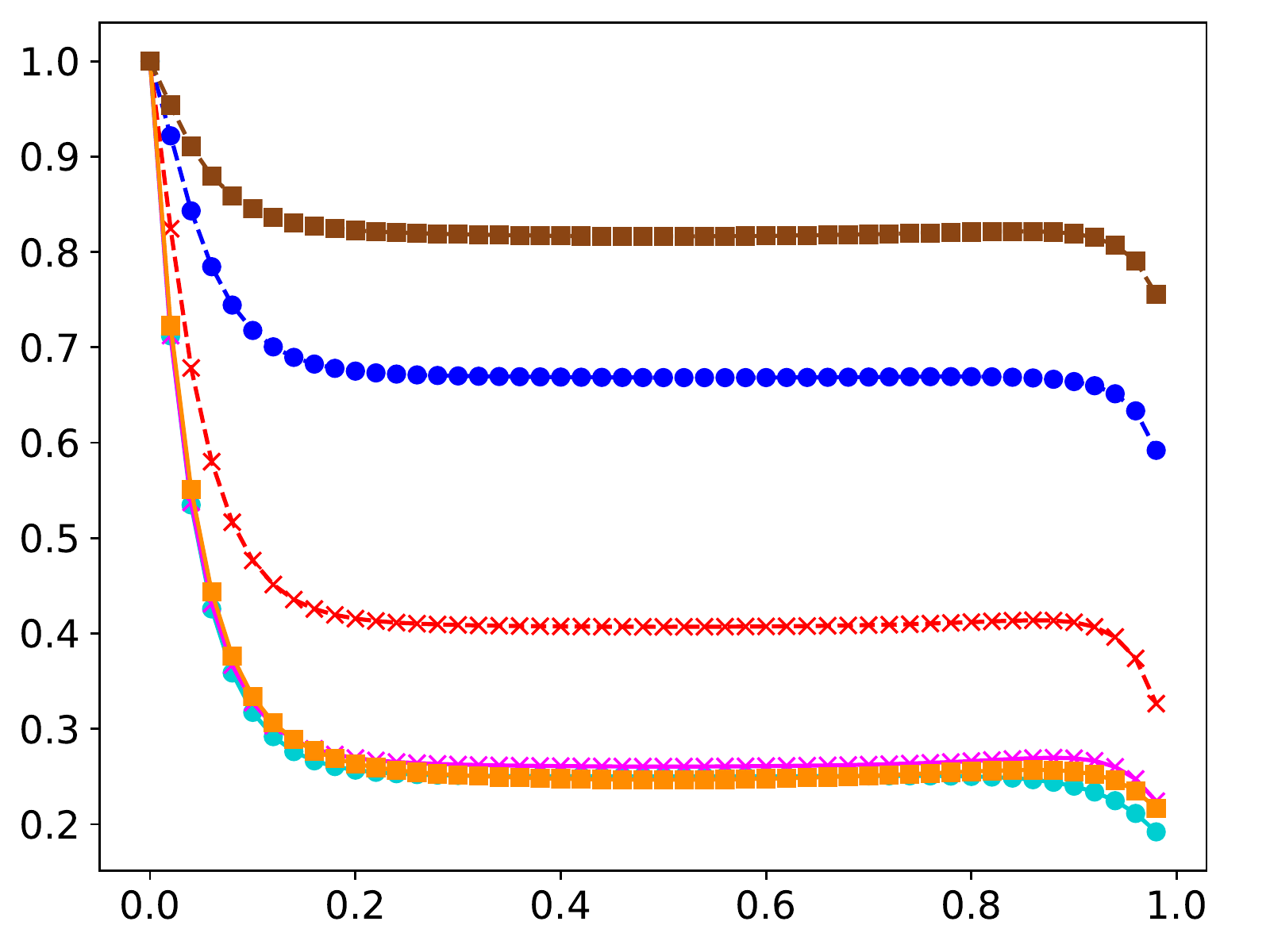}};
  \node[below=of img, node distance=0cm, yshift=1.2cm,font=\color{black}] {$p_{dephase}$};
\end{tikzpicture}
}\kern-1.5em
\subfloat[\label{subfig:fidelity_compare_iris_depo}]{
         \begin{tikzpicture}
  \node (img)  {\includegraphics[width = 0.2\textwidth, height = 0.2\textwidth]{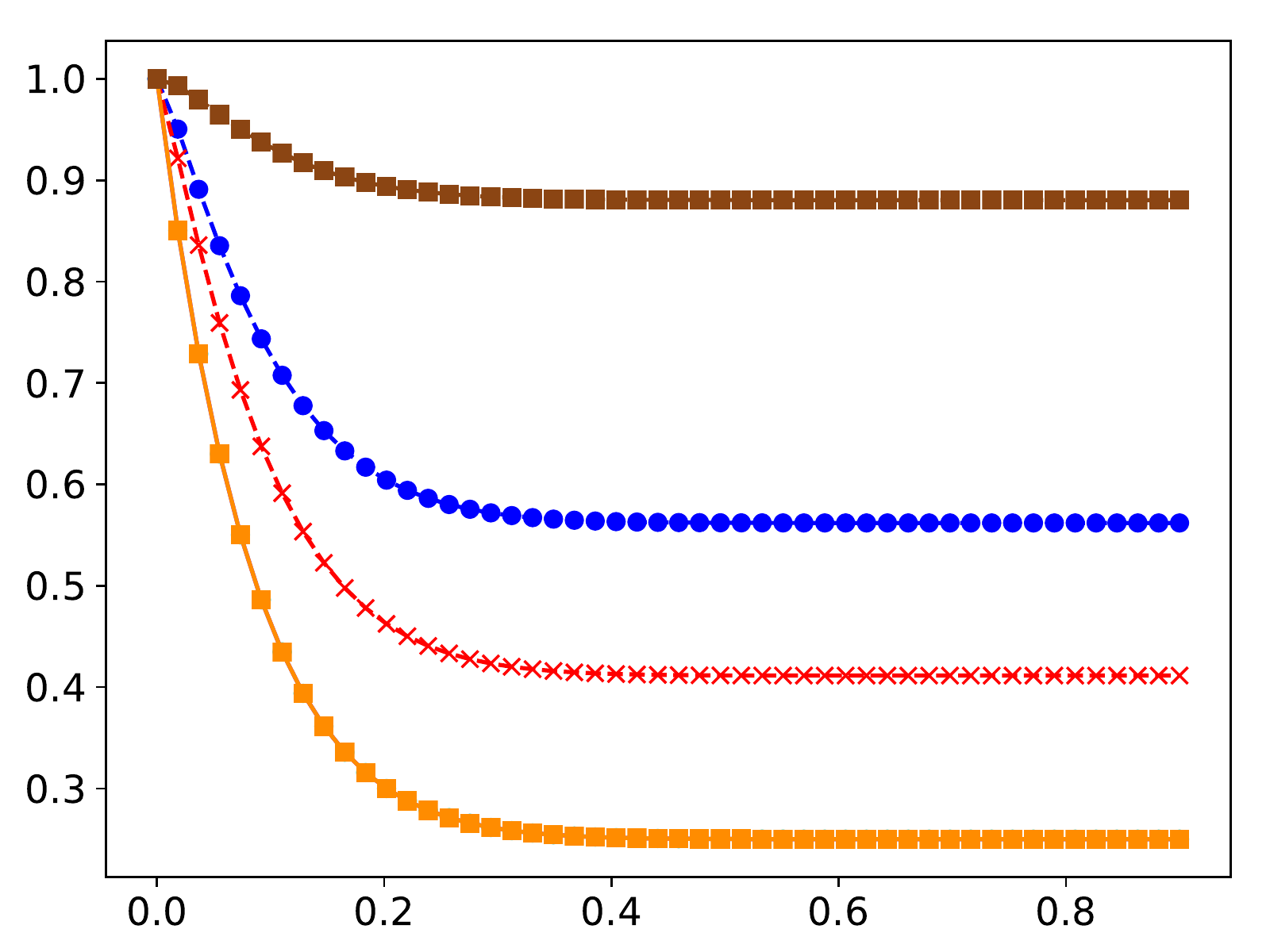}};
  \node[below=of img, node distance=0cm, yshift=1.2cm,font=\color{black}] {$p_{depo}$};
\end{tikzpicture}
}\kern-1.5em
\subfloat{
\vspace{-3em}
         \begin{tikzpicture}
  \node (img)  {\includegraphics[width = 0.17\textwidth, height = 0.14\textwidth]{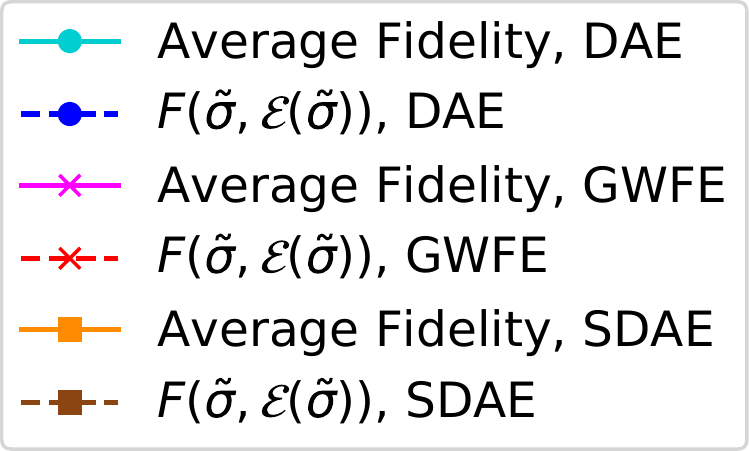}};
\end{tikzpicture}
}\\ \vspace{-1em}
\subfloat[ \label{subfig:fidelity_bound_iris_bit_flip}]{
\setcounter{subfigure}{5}
         \begin{tikzpicture}
  \node (img)  {\includegraphics[width=0.2\textwidth, height = 0.2\textwidth]{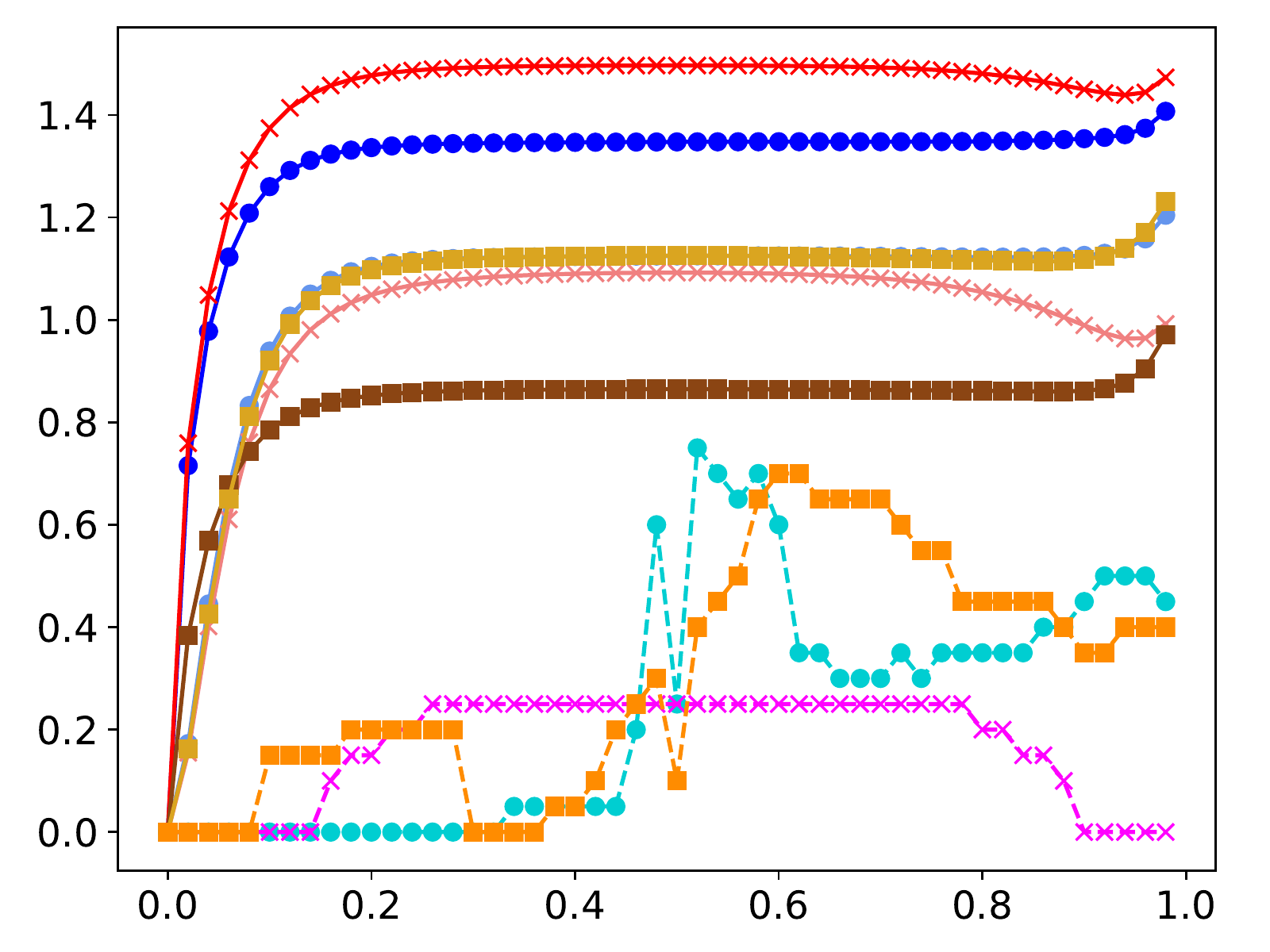}};
  \node[below=of img, node distance=0cm, yshift=1.2cm,font=\color{black}] {$p_X$};
\end{tikzpicture}
}\kern-1.5em
\subfloat[ \label{subfig:fidelity_bound_iris_amp_damp}]{
         \begin{tikzpicture}
  \node (img)  {\includegraphics[width=0.2\textwidth, height = 0.2\textwidth]{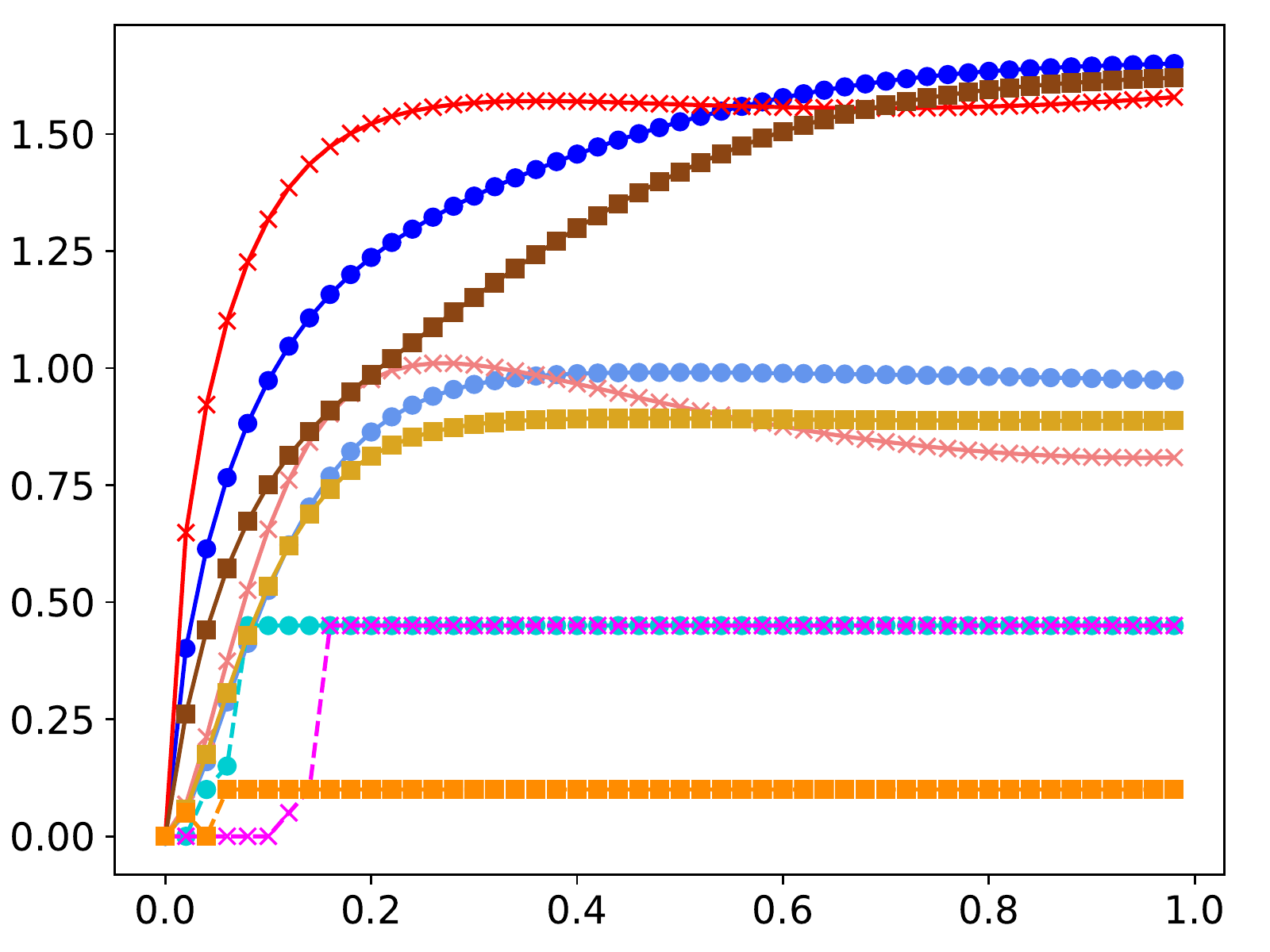}};
  \node[below=of img, node distance=0cm, yshift=1.2cm,font=\color{black}] {$p_{damp}$};
\end{tikzpicture}
}\kern-1.5em
\subfloat[\label{subfig:fidelity_bound_iris_dephase}]{
         \begin{tikzpicture}
  \node (img)  {\includegraphics[width = 0.2\textwidth, height = 0.2\textwidth]{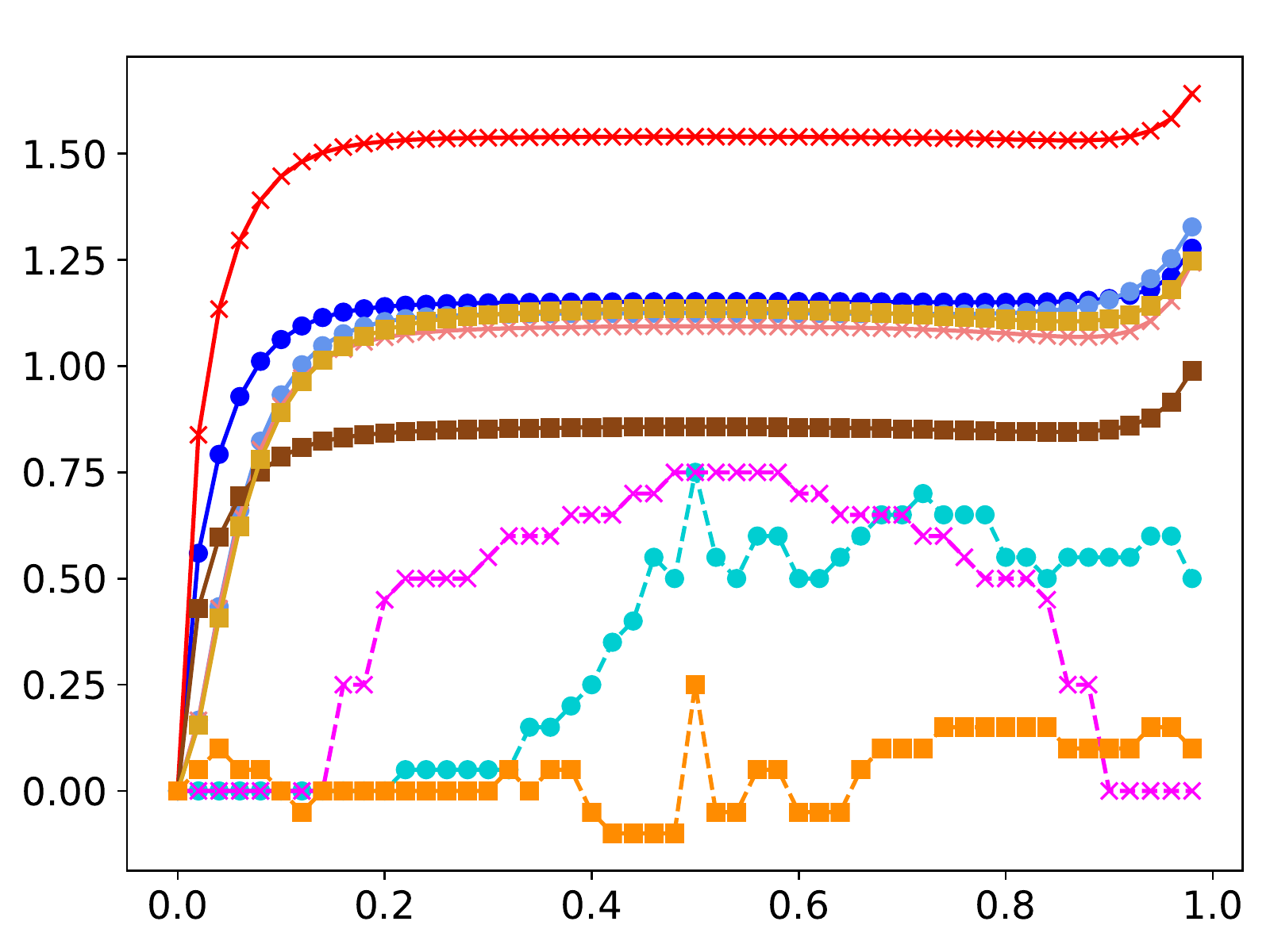}};
  \node[below=of img, node distance=0cm, yshift=1.2cm,font=\color{black}] {$p_{dephase}$};
\end{tikzpicture}
}\kern-1.5em
\subfloat[\label{subfig:fidelity_bound_iris_depo}]{
         \begin{tikzpicture}
  \node (img)  {\includegraphics[width = 0.2\textwidth, height = 0.2\textwidth]{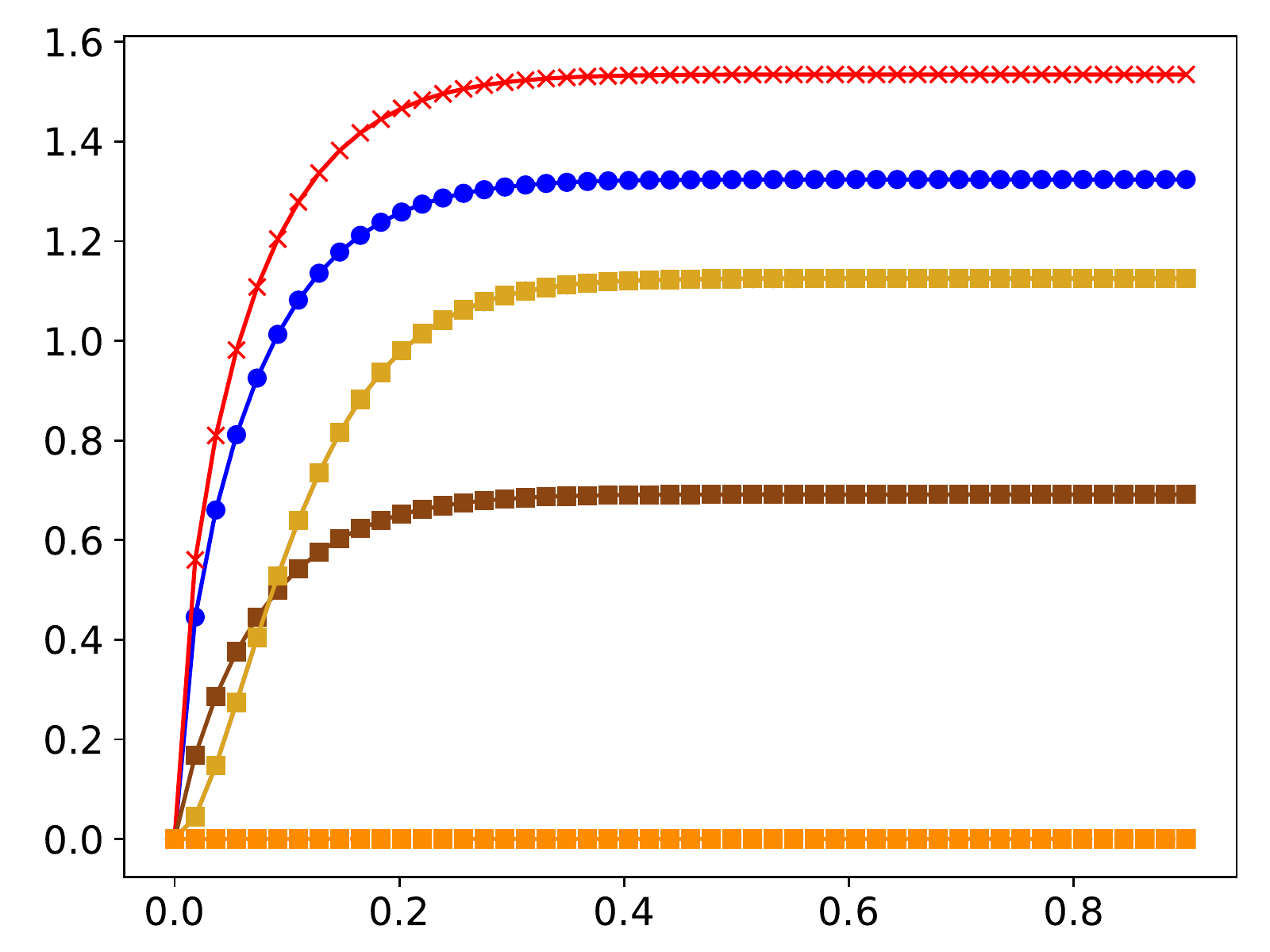}};
  \node[below=of img, node distance=0cm, yshift=1.2cm,font=\color{black}] {$p_{depo}$};
\end{tikzpicture}
}\kern-1.5em
\subfloat{
         \begin{tikzpicture}
  \node (img)  {\includegraphics[width = 0.2\textwidth, height = 0.2\textwidth]{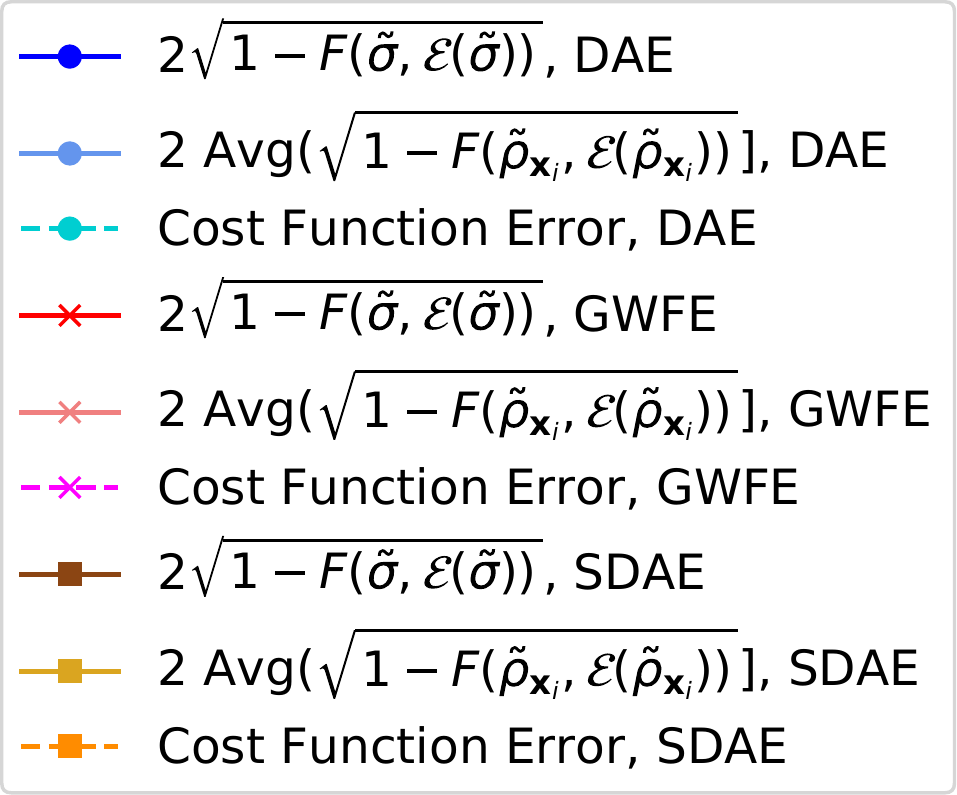}};
    \node[below=of img, node distance=0cm, yshift=2cm,font=\color{black}] {};
  \end{tikzpicture}
}
\caption{Top row: Fidelity between noisy and noiseless states for different encodings on the Iris dataset. From left to right, the noise models are bit-flip noise, amplitude damping noise, dephasing noise, and global depolarizing noise, with strengths varied across the horizontal axis. Bottom row: Upper bounds on partial robustness in terms of fidelity using the bounds~\eqref{eqn:fidelity_bound_mixed_state} and~\eqref{eqn:fidelity_bound_average}. For each plot (in both rows), three different data encodings are considered --- curves corresponds three product qubit encodings: dense angle encoding, wavefunction encoding, and superdense angle.}
\label{fig:fidelity_analysis_iris}
\end{figure*}
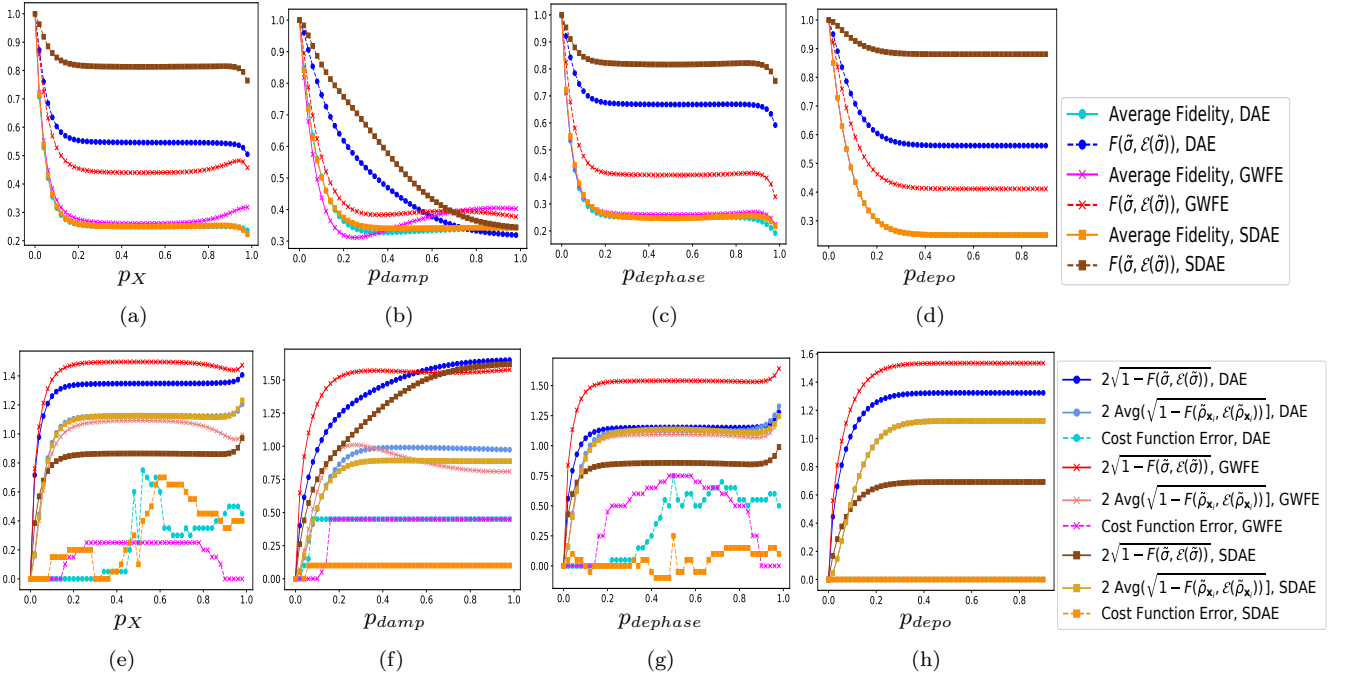

\section{Discussion and Conclusions} \label{sec:conclusions}

In this paper, we have formally defined a model for a binary quantum classifier common in recent literature and studied encoding functions in detail. In particular, we showed that different encodings lead to different learnable decision boundaries and thus have an important effect on the overall success of the classifier. We introduced and formally defined the concept of robust points as well as robust sets and (completely) robust encodings. In addition to studying the robustness criteria for several common noise models, we have characterized robust sets for example channels and discussed their relationship to the fixed points of the channels. We used this connection to provide an existence proof of robust encodings and discussed an empirically observed tradeoff between learnability and robustness. Finally, we considered an embedded cost function classifier using an indicator cost function and provided upper bounds on the robustness of an encoding in terms of fidelities between ideal and noisy evolution.

In addition to the above theoretical results, we performed several numerical implementations to confirm and extend our findings. Specifically, we numerically evaluated decision boundaries for different encodings and performed example implementations on standard datasets in machine learning. Additionally, we used numerics to show that different encodings lead to different robust sets, and quantified the size and location of such sets for different encodings. Finally, we presented an ``encoding learning algorithm'' which optimizes over hyperparameters in the encoding to attempt to find robust sets. We provided proof of principle implementations on three datasets to show the performance of this algorithm. Finally, we showed upper bounds on partial robustness by computing the fidelities between ideal and noisy final states for an example implementation. We provide code to reproduce all numerical results at Ref.~\cite{coyle_noiserobustclassifier_2020}.  

Within quantum machine learning, this work constitutes one of a relatively small number of papers to study encoding functions in detail. Our concept of robust points and robust encodings is novel to this work, and some of our analytical results help explain phenomena observed in recent literature --- namely, that misclassifications due to depolarizing noise found in Ref.~\cite{grant_hierarchical_2018} are solely due to finite shot statistics. 

As this work introduces new concepts and discusses a relatively under-studied area in quantum machine learning, there are several avenues for future research. While we have provided multiple analytic results, the framework we introduced for proving these results is perhaps more impactful. In future work, these ideas can be applied to prove more robustness results for different channels and different decision boundaries than the one we considered in this work. Specific tasks we leave to future work include generalizing results to classes of quantum channels (e.g., unital channels), quantifying the tradeoff between learnability and robustness, and characterizing the conditions under which a \textit{completely} robust encoding exists. It is also interesting from a purely theoretical perspective to extend the notion of a robust point to a ``generalized fixed point'' of a channel, i.e. points which satisfy $f(\E(\rho)) = f(\rho)$ for some function $f$. (When $f = \yhat$, the ``generalized fixed point'' is a robust point, but other arbitrary functions $f$ could be considered.)

From an applications perspective, a clear task for future work is to test the ideas introduced here on a NISQ computer. While the channels we considered are standard theoretical tools for analyzing noise, more complicated effects such as crosstalk occur in real devices. Implementations on NISQ computers would assess our results \textit{in situ} and potentially give additional insight into how robust encodings can be designed. Also, for future numerical work one could consider additional datasets (e.g., the MNIST dataset~\cite{lecun_gradient-based_1998}) and test the performance of different encodings both in terms of learnability and robustness. Last, to further incorporate with recent literature, one could consider robustness in the context of ``data re-uploading''~\cite{perez-salinas_data_2019} or in the adversarial setting of Refs.~\cite{lu_quantum_2019, kechrimparis_channel_2019}.

More broadly for NISQ applications, our work introduces a problem-specific strategy for dealing with errors. In contrast to error mitigation or even error correction techniques which allow errors to occur then attempt to mitigate (correct) them, our strategy of robust data encoding attempts to set up the problem such that any errors which do occur have no effect on the final result. 
We exploit the natural machine learning concept of data representation to achieve this effect, but in principle a similar idea could be used in other settings. This work defines fundamental concepts and proves several results for important questions that must be addressed on the path to practical applications with near-term quantum computers. 

\section*{Acknowledgments} \label{sec:acknowledgements}

We thank Nana Liu and Arkin Tikku for useful discussions and feedback. RL thanks Filip Wudarski, Stuart Hadfield, and Tad Yoder for helpful comments. BC thanks Niraj Kumar, Matty Hoban, and Atul Mantri for helpful comments.
BC was supported by the Engineering and Physical Sciences Research Council (grant EP/L01503X/1), EPSRC Centre for Doctoral Training in Pervasive Parallelism at the University of Edinburgh, School of Informatics and Entrapping Machines, (grant FA9550-17-1-0055). 
RL acknowledges partial support from the \href{https://unitary.fund}{Unitary Fund}.



\newpage
\appendix


\section{Preliminaries and useful formulae} \label{app:useful-formulae}

The single qubit Pauli operators are
\begin{equation} \label{eqn:paulis}
    I := \left[ \begin{matrix}
        1 & 0\\
        0 & 1 \\
    \end{matrix} \right], \ 
    X := \left[ \begin{matrix}
        0 & 1\\
        1 & 0 \\
    \end{matrix} \right], \
    Y := \left[ \begin{matrix}
        0 & -i \\
        i & 0 \\
    \end{matrix} \right], \
    Z := \left[ \begin{matrix}
        1 & 0 \\
        0 & -1 \\
    \end{matrix} \right]
\end{equation}

Let $\rho$ be a single qubit state with matrix elements $\rho_{ij}$, i.e.
\begin{equation}
    \rho := \left[ \begin{matrix}
        \rho_{00} & \rho_{01} \\
        \rho_{10} & \rho_{11} \\
    \end{matrix} \right] .
\end{equation}
Then,
\begin{align*}
    X \rho X &= \left[ \begin{matrix}
        \rho_{11} & \rho_{10} \\
        \rho_{01} & \rho_{00} \\
    \end{matrix} \right]   \\
    Y \rho Y &= \left[ \begin{matrix}
        \rho_{11} & - \rho_{10} \\
        - \rho_{01} & \rho_{00} \\
    \end{matrix} \right]   \\
    Z \rho Z &= \left[ \begin{matrix}
        \rho_{00} & - \rho_{01} \\
        - \rho_{10} & \rho_{11} \\
    \end{matrix} \right]   
\end{align*}

Defining the projectors $\Pi_0 = |0\>\<0|$ and $\Pi_1= |1\> \< 1|$, one can show
\begin{align}
    \Tr [ \Pi_0 X \rho X] &= \Tr [ \Pi_1 \rho] \\
    \Tr [ \Pi_0 Y \rho Y] &= \Tr [ \Pi_1 \rho] \\
    \Tr [ \Pi_0 Z \rho Z] &= \Tr [ \Pi_0 \rho] ,
\end{align}

For any hermitian matrix $A = [A_{ij}]$ and any unitary matrix $U = [U_{ij}]$, we have
\begin{multline} \label{eqn:useful-00mtx-elt}
    \Tr [ \Pi_0 U A U^\dagger] =\\
    |U_{00}|^2 A_{00} + 2 \Re [ U_{00}^* U_{01} A_{10}] + |U_{01}|^2 A_{11} .
\end{multline}
Similarly, one can show that
\begin{multline} \label{eqn:useful-11mtx-elt}
    \Tr[ \Pi_{1} U A U^\dagger ] =\\
     |U_{10}|^{2} A_{00} + 2 \Re [ U_{11}^* U_{10} A_{01}] + |U_{11}|^2 A_{11} .
\end{multline}
If we further assume the single qubit parameterized unitary, $U(\boldsymbol \alpha)$, has the following decomposition: $R_z(2\alpha_1)R_y(2\alpha_2)R_z(2\alpha_3)$ (up to a global phase) \cite{nielsen_quantum_2010}, we get:
\begin{equation}\label{eqn:single_qubit_unitary_decompostion}
    U(\boldsymbol\alpha) = \left[ \begin{matrix}
        e^{i ( - \alpha_1  - \alpha_3)} \cos\alpha_2&
        -e^{i ( - \alpha_1+ \alpha_3 )} \sin \alpha_2\\
        e^{i ( \alpha_1  - \alpha_3 )} \sin \alpha_2 & 
        e^{i ( \alpha_1  + \alpha_3)} \cos \alpha_2 \\
    \end{matrix} \right]
\end{equation}
Therefore, we get the various terms to be:
\begin{align*}
    &|U_{00}|^2 = \cos^2(\alpha_2) \\
    &|U_{01}|^2 = |U_{10}|^2 = \sin^2(\alpha_2) \\
    &|U_{11}|^2 = \cos^2(\alpha_2)\\
     &U_{00}^*U_{01} = -e^{i2\alpha_3}\cos(\alpha_2)\sin(\alpha_2) = -\frac{1}{2}e^{i2\alpha_2}\sin(2\alpha_2) \\
     &U_{11}^*U_{10}= e^{-i2\alpha_3}\cos(\alpha_2)\sin(\alpha_2) = \frac{1}{2}e^{-2i\alpha_3}\sin(2\alpha_2)
\end{align*}
So the conditions, (\ref{eqn:useful-00mtx-elt}, \ref{eqn:useful-11mtx-elt}) become:
\begin{multline} 
  \Tr [ \Pi_0 U A U^\dagger] \\
  =|U_{00}|^2 A_{00} + 2 \Re [ U_{00}^* U_{01} A_{10}] + |U_{01}|^2 A_{11}  \\
  = \cos^2\left(\alpha_2\right)A_{00} + \sin^2\left(\alpha_2\right)A_{11} - \Re[e^{2i\alpha_3}\sin\left(\alpha_2\right)A_{10}] \label{eqn:useful-00mtx-elt-unitary-filled} 
\end{multline}
\begin{multline}
        \Tr[ \Pi_{1} U A U^\dagger ] 
        =|U_{10}|^{2} A_{00} + 2 \Re [ U_{11}^* U_{10} A_{01}] + |U_{11}|^2 A_{11} \\
      =\sin^2\left(\alpha_2\right)A_{00} +\cos^2\left(\alpha_2\right)A_{11}
      +\Re[e^{-2i\alpha_3}\sin\left(2\alpha_2\right)A_{01}] \label{eqn:useful-11mtx-elt-unitary-filled}
\end{multline}



\section{Proofs and Additional Results} \label{app:further_proofs}

Here we give the explicit proofs for the remaining theorems (which we also repeat here for completeness) in the main text, and some others introduced here.


\subsection{Robustness to measurement noise} \label{app_ssec:meas_noise}

Just as the case of quantum compilation~\cite{sharma_noise_2020}, we can deal with measurement noise in the classifier:

\begin{definition} \label{defn:measurement_noise}
\textit{Measurement noise} (MN) is defined as a modification of the standard POVM basis elements, $\{\Pi_0 = \ketbra{0}{0}, \Pi_1 = \ketbra{1}{1}\}$ by the channel $\mathcal{E}_{\vec{p}}^{\textnormal{meas}}$ with assignment matrix $\vec{p}$ for a single noiseless qubit:
\begin{align}
    \Pi_0 &= \ketbra{0}{0} \overset{\mathcal{E}_{\vec{p}}^{\textnormal{meas}}}{\rightarrow} \tilde{\Pi}_0 = p_{00}\ketbra{0}{0} + p_{01}\ketbra{1}{1}\nonumber\\
    \Pi_1 &= \ketbra{1}{1} \overset{\mathcal{E}_{\vec{p}}^{\textnormal{meas}}}{\rightarrow} \tilde{\Pi}_1 = p_{10}\ketbra{0}{0} + p_{11}\ketbra{1}{1} \label{eqn:noisy_povm_elements}\\
    \vec{p}& :=  \left(\begin{array}{cc}
         p_{00} & p_{01}  \\
         p_{10} & p_{11} 
    \end{array}\right)\nonumber
\end{align}
where $p_{00} + p_{10} = 1, p_{10} +p_{11} = 1$, and hence $p_{kl}$ is the probability of getting the $k$ outcome given the $l$ input. Furthermore, we assume that $p_{kk} > p_{kl}$ for $k \neq l$.
\end{definition}

The definition for the general case of $n$ qubit measurements can be found in \cite{sharma_noise_2020}, but we shall not need it here, since we only require measuring a single qubit to determine the decision function. More general classifiers which measure multiple qubits (e.g. and then take a majority vote for the classification) could also be considered, but these are outside the scope of this work. Now, we can show the following result in a similar fashion to the above proofs:

\begin{theorem} \label{thm:measurement_noise_robustness}
Let $\mathcal{E}_{\vec{p}}^{\textnormal{meas}}$ define measurement noise acting on the classification qubit and consider a quantum classifier on data from the set $\mathcal{X}$. 
Then, for any encoding $E: \mathcal{X} \rightarrow \mathcal{D}_2$, we have complete robustness
    \begin{equation}
        \mathcal{R} (\mathcal{E}_{\vec{p}}^{\textnormal{meas}}, E, \yhat) = \mathcal{X}
    \end{equation}
     if the measurement assignment probabilities satisfy $p_{00} > p_{01}, p_{11} > p_{10}$.
\end{theorem}

\begin{proof}
We can write the measurement noise channel acting on the POVM elements as:
\begin{align}
    \mathcal{E}_{\vec{p}}^{\textnormal{meas}}(\Pi^{(c)}_0\otimes I^{\otimes n-1})  = (p_{00}\ketbra{0}{0} + p_{01}\ketbra{1}{1}) \otimes I^{\otimes n-1} \nonumber
\end{align}
Again, if we had correct classification before the noise, $\Tr(\ketbra{0}{0}) \geq 1/2$, then:
\begin{align*}
    &\Tr[\mathcal{E}_{\vec{p}}^{\textnormal{meas}}(\Pi^{(c)}_0\otimes I^{\otimes n-1})\rhotildex]\\ 
    &= \Tr[(\{p_{00}\ketbra{0}{0} + p_{01}\ketbra{1}{1}\} \otimes I^{\otimes n-1})\rhotildex]\\
    &= p_{00}\Tr[(\ketbra{0}{0}\otimes I^{\otimes n-1})\rhotildex] + p_{01}\Tr[(\ketbra{1}{1} \otimes I^{\otimes n-1})\rhotildex]\\
    &= p_{00}\Tr[(\ketbra{0}{0}\otimes I^{\otimes n-1})\rhotildex] \\ 
    &+ p_{01}(1- \Tr[(\ketbra{0}{0}\otimes I^{\otimes n-1}) \rhotildex])\\
    &= (p_{00} - p_{01})\Tr[(\ketbra{0}{0}\otimes I^{\otimes n-1})\rhotildex] + p_{01}\\
    &\geq (p_{00} - p_{01})1/2 + p_{01} = 1/2(p_{00} + p_{01}) = 1/2
\end{align*}
where in the last line, we used the fact that $p_{00}+p_{01}=1$ and our assumption that $p_{00} >  p_{01}$. The same result holds if the vector was classified as $1$, and hence the classifier is robust to measurement noise. 
\end{proof}
Just as above, we can replace the ideal state, $\rhotildexi$ with a noisy state, $\mathcal{E}(\rhotildexi)$, where the operator accounts for other forms of noise, not including measurement noise. We can see this allows us to take a model which is robust without measurement noise, and `upgrade' it to one which is. However, we may be able to find looser restrictions by considering different types of noise together, rather than in this modular fashion.

To illustrate the results of \theref{thm:measurement_noise_robustness} we focus on the dense angle encoding, which can achieve nearly 100\% accuracy on the ``vertical'' dataset. We then compute the percentage which would be misclassified as a function the assignment probabilities in the noisy projectors in \eqref{eqn:noisy_povm_elements}. The results are seen in \figref{fig:measurement_noise_results}.

\begin{figure}[!ht]
\begin{tikzpicture}
  \node (img)  {\includegraphics[width = 0.3\textwidth, height = 0.22\textwidth]{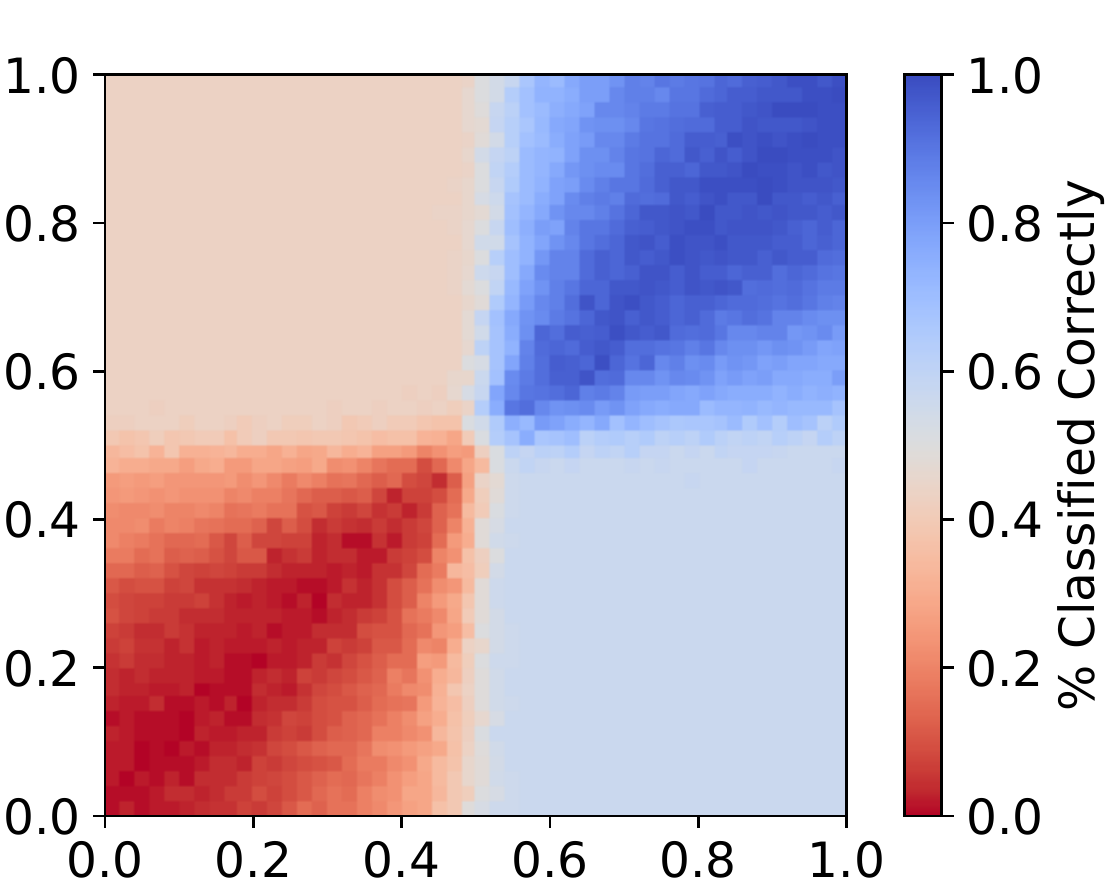}};
  \node[below=of img, node distance=0cm, yshift=1.1cm,font=\color{black}] {\hspace{-2.5em}\large{$p_{11}$}};
  \node[left=of img, node distance=10cm, rotate=90, anchor=center,yshift=-0.9cm,font=\color{black}] {\large{$p_{00}$}};
\end{tikzpicture}
\caption{Misclassification percentage as a result of (b) measurement noise as a function of probabilities $\{p_{00}, p_{11}\}$. If either $p_{00}$ or $p_{11}$ is less than $1/2$, then half the correctly classified points will be misclassified, with the probability increasing as expected, with the number of misclassified points increasing as the off diagonal terms, $p_{01}, p_{10} \rightarrow 1$, as expected from \theref{thm:measurement_noise_robustness}. By classified `correctly' in this context, we mean the fraction of points which are classified \emph{the same} with and without noise. 
}
\label{fig:measurement_noise_results}
\end{figure}

 
\subsection{Robustness for Factorizable Noise} \label{app:factorizable-noise}

\begin{theorem}\label{thm:factorizable-noise-before-meas}
If $\E$ is any noise channel which factorizes into a single qubit channel, and a multiqubit channel as follows:
\begin{align}
    \E(\rho) = \E_{\Bar{c}}(\rhotildex^{\Bar{c}}) \otimes \E_c(\rhotildex^c)
\end{align}
where WLOG $\E_c$ acts only on the classification qubit $(\rhotildex^c = \Tr_{\Bar{c}}(\rhotildex))$ after encoding and unitary evolution, and $\E_{\Bar{c}}$ acts on all other qubits arbitrarily, $(\rhotildex^{\Bar{c}} = \Tr_c(\rhotildex))$. Further assume the state meets the robust classification requirements for the single qubit error channel $\E_c$. Then the classifier will be robust to $\E$.
\end{theorem}

\begin{proof}
The correct classification again depends on the classification qubit measurement probabilities:
$\Tr(\ketbra{0}{0}^{(c)}\otimes I^{\otimes n-1}\rhotildex), \Tr(\ketbra{1}{1}^{(c)}\otimes I^{\otimes n-1} \rhotildex)$. If $\rhotildex$ is robust to the single qubit error channel $\E_c$, this means 
\begin{align*}
\Tr(\Pi_0^c\rhotildex) \geq 1/2 &\implies \Tr(\Pi_0^c \E_c(\rhotildex)) \geq 1/2 \\
\Tr(\Pi_1^c \rhotildex) < 1/2  &\implies \Tr(\Pi_1^c \E_c(\rhotildex)) < 1/2
\end{align*}
Then WLOG, assume the point $\vec{x}$ classified as $y(\rhotildex) = 0$ before the noise, then:
\begin{align*}
\Tr\left(\Pi_0^c\E(\rhotildex)\right)&= \Tr\left(\Pi_0^c \left[\E_{\Bar{c}}(\rhotildex^{\Bar{c}} \otimes \E_c(\rhotildex^c)\right]\right) \\
&=\Tr_{\Bar{c}}\left(\E_{\Bar{c}}(\rhotildex^{\Bar{c}})\right) \Tr_{c}\left(\ketbra{0}{0} \E_c(\rhotildex^c)\right) \\
&=\Tr \left(\ketbra{0}{0} \E_c(\rhotildex^c)\right) \geq 1/2 \qedhere
\end{align*}
\end{proof}

The above theorem is a simple consequence of causality in the circuit, only errors which have to happen before the measurement can corrupt the outcome. As such, outside of single qubit errors, we only need to consider errors before the measurement which specifically involve the classification qubit.

\subsection{Fidelity Bound} \label{app_ssec:fidelity_bound}
Here we derive the fidelity bound, \eqref{eqn:fidelity_bound_average} in a similar fashion to \eqref{eqn:fidelity_bound_mixed_state}:
\begin{align}
    \Delta_\E C &:= \left| C_\E - C \right| \nonumber \\[1.0ex]
    &= \left|  \Tr [ D ( \E (\tilde{\sigma} ) - \tilde{\sigma}) ] \right| \nonumber\\   
    &\leq \frac{1}{M}\sum\limits_{i=1}^M| \Tr(D\left[ \E(\rhotildexi)\otimes\ketbra{y_i}{y_i} - \rhotildexi\otimes\ketbra{y_i}{y_i}\right])| \nonumber\\
    &\leq \frac{1}{M}\sum\limits_{i=1}^M||D||_{\infty} ||\left[\E(\rhotildexi)- \rhotildexi\right]\otimes\ketbra{y_i}{y_i}||_1 \nonumber\\
    &\leq \frac{2}{M}\sum\limits_{i=1}^M \sqrt{1-F(\E(\rhotildexi), \rhotildexi)} \label{eqn:fidelity_bound_average_appendix}
\end{align}
Again, we use H\"{o}lders, the Fuchs-van de Graaf and the triangle inequalities, with $||D||_\infty := \max_j | \lambda_j(D) | = 1$.


\section{More Details on Numerical Results} \label{app:numerical_results}

In this section, we present supplementary numerical results to those in the main text. Firstly, in \figref{fig:single_qubit_datasets}, we illustrate the three single qubit datasets we employ here, namely the ``vertical'', ``diagonal'' and ``moons''. The former two are linearly separable, whereas the ``moons'' dataset is nonlinear.
\begin{figure}
\subfloat[\label{subfig:random_vertical_boundary_dataset}]{
         \begin{tikzpicture}
              \node (img)  {\includegraphics[width = 0.13\textwidth, height = 0.12\textwidth]{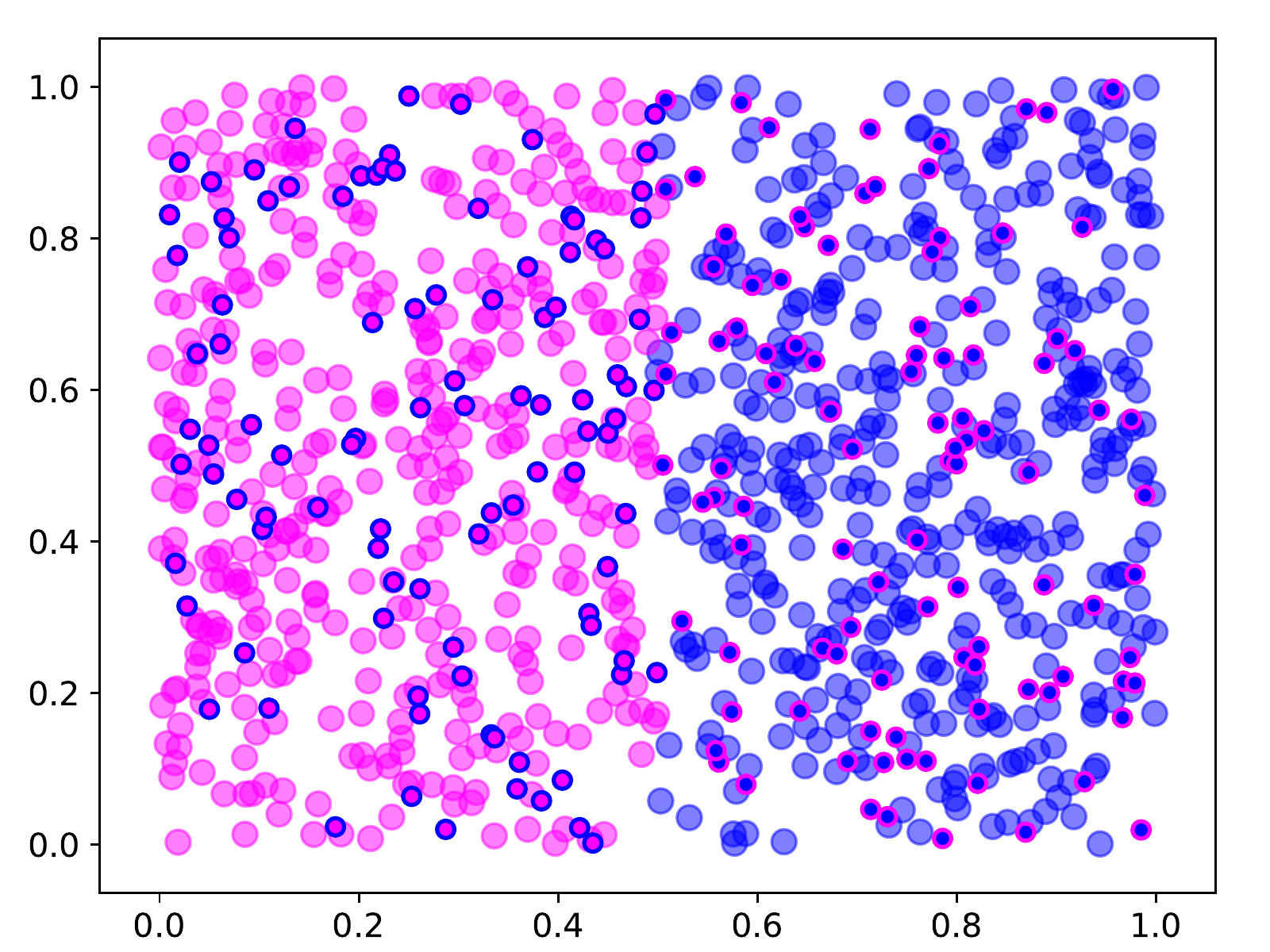}};
              \node[below=of img, node distance=0cm, yshift=1.2cm,font=\color{black}] {$x_1$};
              \node[left=of img, node distance=0cm, rotate=90, anchor=center,yshift=-0.95cm,font=\color{black}] {$x_2$};
        \end{tikzpicture}
    }\kern-1.5em
    \subfloat[ \label{subfig:random_diagonal_boundary_dataset_}]{
         \begin{tikzpicture}
      \node (img)  {\includegraphics[width=0.13\textwidth, height = 0.12\textwidth]{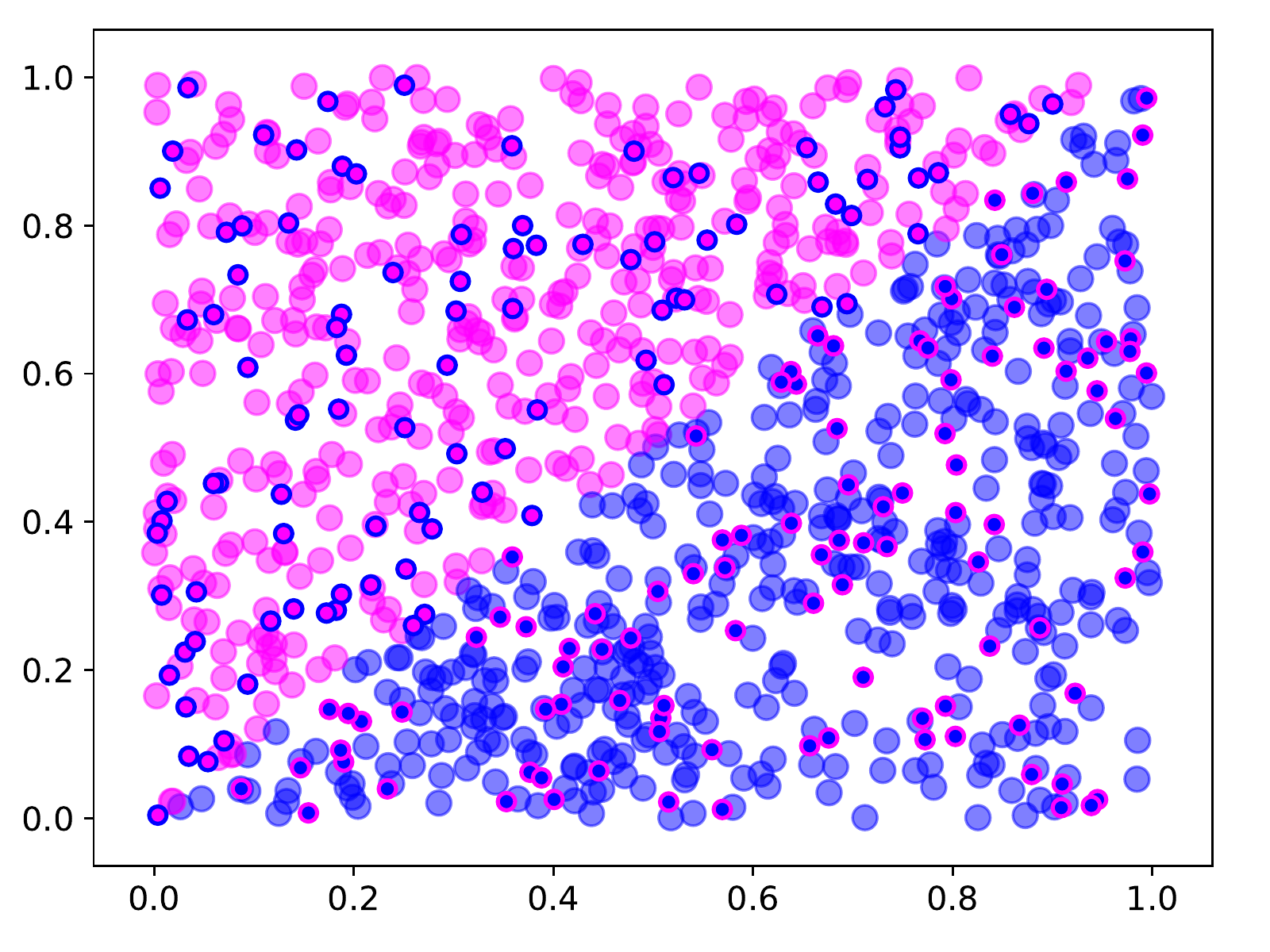}};
      \node[below=of img, node distance=0cm, yshift=1.2cm,font=\color{black}] {$x_1$};
      \node[left=of img, node distance=0cm, rotate=90, anchor=center,yshift=-0.95cm,font=\color{black}] {};
    \end{tikzpicture}
    }\kern-1.5em
    \subfloat[\label{subfig:moons_dataset}]{
             \begin{tikzpicture}
      \node (img)  {\includegraphics[width=0.13\textwidth, height = 0.12\textwidth]{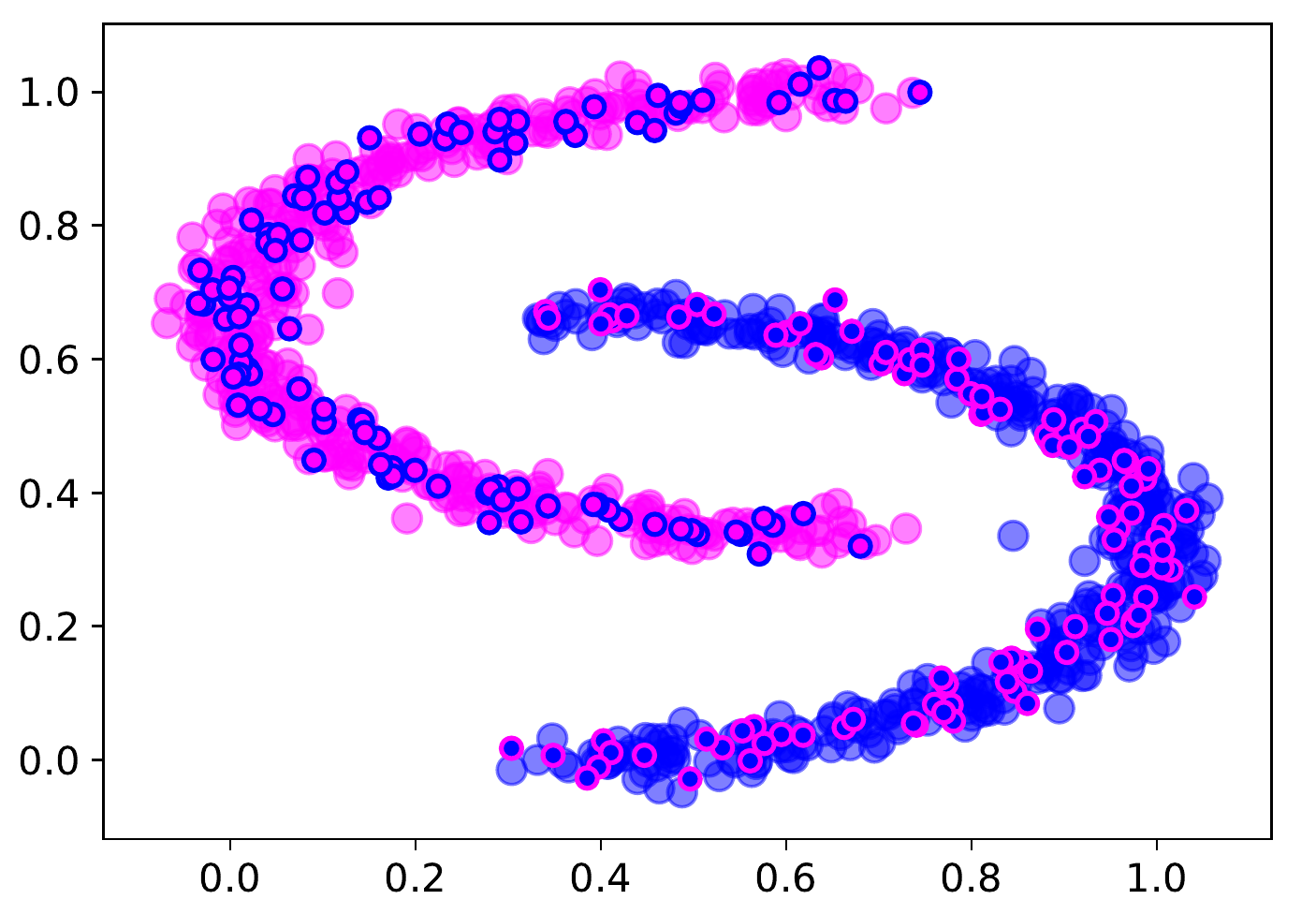}};
      \node[below=of img, node distance=0cm, yshift=1.2cm,font=\color{black}] {$x_1$};
      \node[left=of img, node distance=0cm, rotate=90, anchor=center,yshift=-0.95cm,font=\color{black}] {};
    \end{tikzpicture}
}
\caption{Three single qubit (two dimensional) datasets which we use. (a) ``vertical'', (b) ``diagonal'' and (c) the ``moons''  dataset from scikit-learn \cite{pedregosa_scikit-learn_2011} rotated by $90^{\circ}$ with a noise level of $0.05$. 20\% of each set is test data, indicated by the points circled with the opposite color. We chose the latter two due to the fact that the ``moons'' and ``vertical'' datasets can be well classified by the dense angle encoding, while the ``diagonal'' dataset can be well classified by the wavefunction encoding, which can be seen by studying the decision boundaries generated in \figref{fig:random_decision_boundaries_encodings}.}
\label{fig:single_qubit_datasets}
\end{figure}
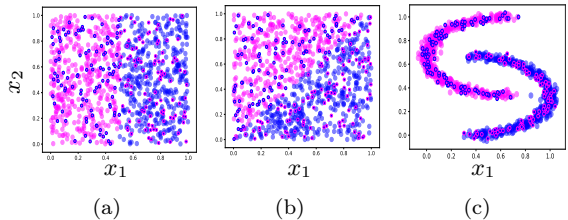

Secondly, to complement the results of \figref{fig:dae_encoding_learnability_versus_robustness} in the main text, in Table~\ref{table:vertical_boundary_encodin_params_plot} we provide the best parameters found in the procedure. Each set of parameters (each row, measured in radians) performs optimally in one of three areas. The first is the noiseless environment, in which a $\theta\approx 2.9$ parameter performs optimally. The second is the amplitude-damped environment, in which $\theta\approx 1.6$ achieves the best accuracy, and finally, $\theta=0$ is the most robust point to encode in, for the whole dataset. For each of these parameter sets, we also test them in the other scenarios, for example, the best parameters found in the noisy environment ($[\theta, \phi] = [1.6, 3.9]$) have a higher $\delta$-Robustness ($81\%$) than those in the noiseless environment ($70\%$), since these parameters force points to be encoded closer to the $\ket{0}$ state, i.e., the fixed point of the channel in question.  
\begin{table}[ht]
\centering
 \begin{tabular}{||c | c | c | c||} 
 \hline
 Parameters &  Accuracy & Accuracy  & $\delta$-Robustness\\
  &   w/o Noise & w/ Noise & \\ [0.5ex] 
 \hline\hline
 $[\theta, \phi] = [2.9, 2.9]$ & $100\%$ & $84\%$ & $70\%$ \\ 
 \hline
 $[\theta, \phi] = [1.6, 3.9]$ & $49\%$ & $100\%$ & $81\%$ \\ 
 \hline
 $[\theta, \phi] = [0, 0]$ & $43\%$ & $43\%$ & $100\%$\\ 
 \hline
\end{tabular}
\caption{Optimal parameters $[\theta, \phi]$ for dense angle encoding (with parameters in $U(\vec{\alpha})$ trained in noiseless environment) in (a) noiseless environment, (b) noisy environment (i.e.\@ amplitude damping channel is added) and (c) for maximal robustness. Optimal parameters in noisy environment are closer to fixed point of amplitude damping channel ($\ket{0}$, i.e.\@ $\theta \equiv 0$) and give a higher value of $\delta$-robustness. }\label{table:vertical_boundary_encodin_params_plot}
\end{table}

\figref{fig:specific_classifier_circuits} illustrates the specific decompositions for the single and two qubit classifiers we utilize for the numerical results in the text. For the matrix representation of the circuit shown in \figref{fig:specific_classifier_circuits}(a), see~\eqref{eqn:single_qubit_unitary_decompostion}.

\begin{figure}[t!]
    \includegraphics[width = 0.45\textwidth, height = 0.15\textwidth]{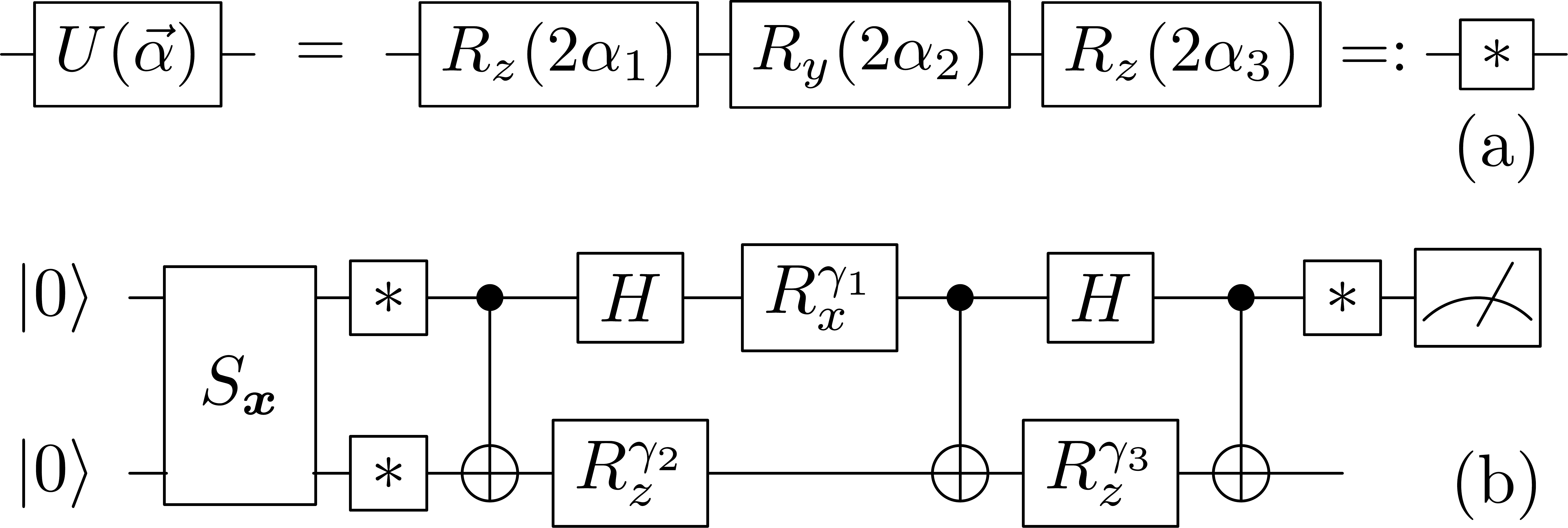}
    \caption{Circuit diagrams for the QNN ansatze we use to obtain numerical results. (a) Ansatz for the single qubit classifier. Here, $R_z^\alpha$ and $R_y^\beta$ denote rotations around the $z$-axis and $y$-axis with angles $\alpha, \beta$, respectively, of the Bloch sphere. We note that any element of $U(2)$ can be represented by this ansatz~\cite{nielsen_quantum_2010}. (b) Ansatz for the two qubit classifier, with $12$ parameters, $\{\alpha_k\}_{k=1}^{12}$. The first $6$ parameters, $\alpha_1,\dots,\alpha_6$, are contained in the first two single qubit unitaries, and $\alpha_{10}, \dots, \alpha_{12}$ are the parameters of the final single qubit gate. The parameters of the intermediate rotations are defined by $\{\gamma_1, \gamma_2, \gamma_3\} := \{2\alpha_7+\pi, 2\alpha_8, 2\alpha_9\}$. This decomposition can realize any two qubit unitary~\cite{vidal_universal_2004} up to global phase with the addition of a single qubit rotation on the bottom qubit at the end of the circuit. We omit this rotation since we only measure the first qubit for classification. As such, we reduce the number of trainable parameters ($\vec{\alpha}$) from $15$ to $12$.}
\label{fig:specific_classifier_circuits}
\end{figure}

To illustrate the results of \theref{thm:robustness_pauli_noise_xy}, we focus on the dense angle encoding, which can achieve nearly 100\% accuracy on the ``vertical'' dataset. We then compute the percentage which would be misclassified as a function of and Pauli noise parameters. The results are seen in \figref{fig:pauli_noise_results}, similar to that observed in \figref{fig:measurement_noise_results}. We note here, that for values of $p_X + p_Y > 1/2$, one has two strategies to achieve robustness. The first is to adjust the measurement basis as per \corrref{corr:pauli_x_robustness} and requires changing the model itself. Alternatively, one can apply an extra step of post processing and relabel every output `$\yhat = 0$' to `$\yhat = 1$', and vice versa.

\begin{figure}[!ht]
\begin{tikzpicture}
  \node (img)  {\includegraphics[width = 0.3\textwidth, height = 0.2\textwidth]{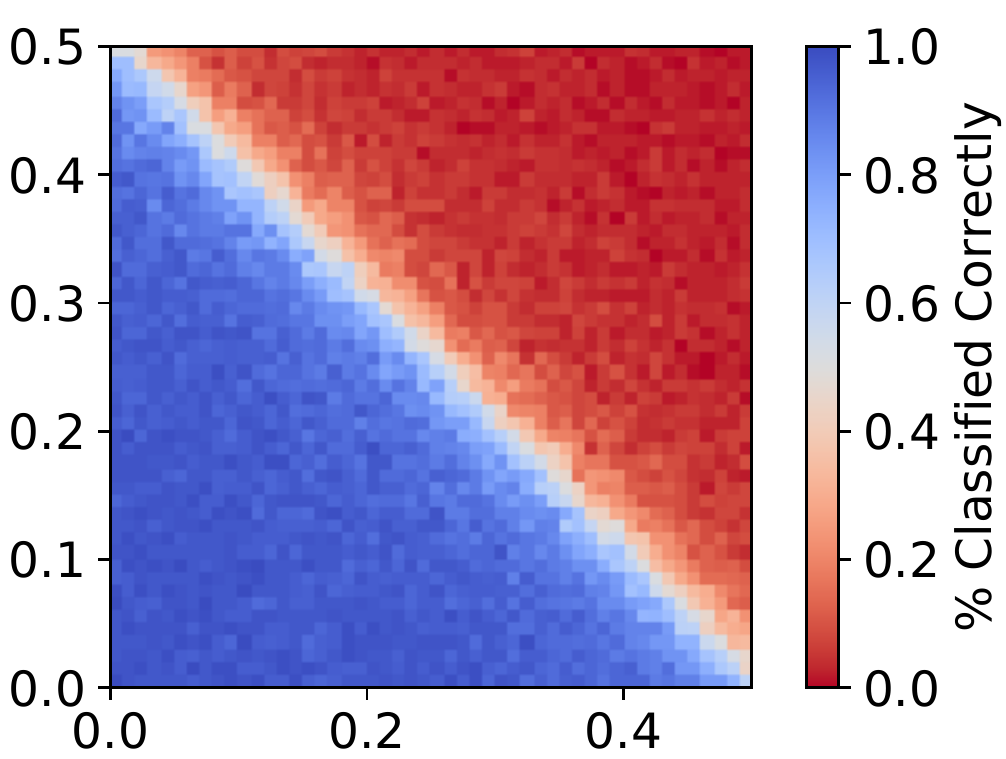}};
  \node[below=of img, node distance=0cm, yshift=1.2cm,font=\color{black}] {\hspace{-2.5em}\large{$p_X$}};
  \node[left=of img, node distance=10cm, rotate=90, anchor=center,yshift=-0.95cm,font=\color{black}] {\large{$p_Y$}};
\end{tikzpicture}
    \caption{Misclassification percentage as a result of Pauli noise with strengths $\{p_X, p_Y, p_Z = 0\}$ As expected, classification is robust for $p_X + p_Y < 1/2$ based on Theorem \ref{thm:robustness_pauli_noise_xy}, and a sharp transition occurs when this contraint is violated to give maximal misclassification. By classified `correctly' in this context, we mean the fraction of points which are classified \emph{the same} with and without noise.  
}
\label{fig:pauli_noise_results}
\end{figure}

Algorithm~\ref{alg:robust-algorithm} contains pseudocode for the encoding learning algorithm discussed in the main text.

\begin{algorithm}[H] \label{alg:robust-algorithm}
\SetAlgoLined
\SetKwInOut{Input}{Input}\SetKwInOut{Output}{output}
\Input{Noise parameters, $\vec{p}$, parameterized quantum circuit, $U(\boldsymbol\alpha)$, $M$ Labelled data examples $(\vec{x}_i, y_i)_{i=1}^M$, encoding set $\{f_l, g_l\}_{l=1}^K$, cost function C.}
\KwResult{Optimized encoding for noise and dataset.}
 Initialize encoding, $(f^*, g^*) \leftarrow \{f_l, g_l\}_{l=1}^K$ and parameters, $\{\vec{\theta}_j\} \leftarrow (0, 2\pi ]_j ~\forall j$ heuristically or at random. Initialize $C^*= M$\; 
 \For{j=1\dots K}{ \do
  Select subset of data: $(\vec{x}_i, y_i)_{i=1}^D \leftarrow(\vec{x}_i, y_i)_{i=1}^M$\;
  Encode each sample using encoding choice, $(f_j, g_j)$: prepare $\{\rhotildexi^{\boldsymbol\alpha, \vec{\theta}_j}\}_{i=1}^D$\;
  $\boldsymbol\alpha^* \leftarrow  \argmin_{\boldsymbol\alpha} C_D(\boldsymbol\alpha, \vec{\theta}_j)$\;
    Add noise with parameters $\vec{p}$: $\{\rhotildex^{\boldsymbol\alpha^*, \vec{\theta}_j}\}_{i=1}^D \leftarrow \{\E_{\vec{p}}(\rhotildex^{\boldsymbol\alpha^*, \vec{\theta}_j^*})\}_{i=1}^D$\;
    $\{\vec{\theta}^*_j\}  \leftarrow \argmin_{\vec{\theta}_j} C_D(\boldsymbol\alpha^*, \vec{\theta}_j)$\; 
    $C^*_j \leftarrow  C_D(\boldsymbol\alpha^*, \vec{\theta}_j^*)$\;
  \If{$C^*_j \leq C^*$}{
  $C^* \leftarrow C^*_j$\;
  $(f^*, g^*) \leftarrow f_j(\vec{\theta}_j^*), g_j(\vec{\theta}_j^*$) \;
  }
  }
 \Output{$C^*, \vec{\alpha}^*, f^*, g^*$}
 \caption{Quantum Encoding Learning Algorithm (QELA)}
\end{algorithm}



\end{document}